\documentclass[12pt]{article}
\usepackage{mheckII}

\usepackage[T1]{fontenc}

\usepackage{manfnt}
\usepackage{longtable}

\usepackage{fancybox}

\usepackage{multirow}

\usepackage{blkarray}

\usepackage{tikz} 
\usetikzlibrary{matrix}

\theoremstyle{theorem}
\newtheorem*{thm}{Theorem}

\newtheorem{fact}{Fact}
\newtheorem{ans}{Answer}
\newtheorem*{proper}{Property}

\newtheorem{corl}{Corollary}
\newtheorem{pro}{Proposition}

\newtheorem{lem}{Lemma}
\newtheorem*{cla}{Claim}

\newtheorem{que}{Question}
\newtheorem*{queos}{Question}
\newtheorem{cri}{Criterion}

\newtheorem*{rprin}{Rough physical principle}

\newtheorem*{defthm}{Fact/Definition}

\theoremstyle{definition}
\newtheorem*{defn}{Definition}
\newtheorem*{rem}{Remark}

\newtheorem*{cproof}{Comments on the proof}

\newtheorem*{cave}{\emph{Caveat}}
\newtheorem*{exe}{Example}

\def\Dsl{\,\raise.15ex\hbox{/}\mkern-13.5mu D}
\def\dsl{\,\raise.25ex\hbox{/}\mkern-10.5mu \partial}

\title{Special Geometry and the Swampland 
}

\authors{Sergio Cecotti\footnote{e-mail: {\tt cecotti@sissa.it}}\vskip 9pt

\centerline{SISSA, via Bonomea 265, I-34100 Trieste, ITALY}
}

\abstract{
In the context of 4d effective gravity theories with 8 supersymmetries, we propose to unify, strenghten, and refine the  several swampland conjectures into a single statement: the \emph{structural criterion,} modelled on the structure theorem in Hodge theory.
In its most abstract form the new swampland criterion applies to all 4d $\cn=2$ effective theories (having a quantum-consistent UV completion) whether supersymmetry is \emph{local} or \emph{rigid}: indeed it may be regarded as the more general version of Seiberg-Witten geometry which holds both in the rigid and local cases. 

As a first  application of the new swampland criterion we show that a quantum-consistent $\cn=2$ supergravity with a cubic pre-potential is necessarily a truncation of a higher-$\cn$ \textsc{sugra}. More precisely: its moduli space is a Shimura variety of `magic' type. In all other cases a quantum-consistent special K\"ahler geometry is either an arithmetic quotient of the complex hyperbolic space $SU(1,m)/U(m)$ or  has no \emph{local} Killing vector. 

Applied to Calabi-Yau 3-folds
 this result implies (assuming mirror symmetry) the validity of the Oguiso-Sakurai conjecture in Algebraic Geometry: all Calabi-Yau 3-folds $X$ without rational curves have Picard number $\rho=2,3$; in facts they are finite quotients of Abelian varieties. More generally: the K\"ahler moduli of $X$ do not receive quantum corrections if and only if $X$ has infinite fundamental group.
In all other cases the K\"ahler moduli have instanton corrections in (essentially) all possible degrees.}

\begin{document}
\maketitle

\tableofcontents

\newpage

\section{Introduction and overview:\\ swampland conjectures {\it vs.\!} VHS structure theorem}

Passing from Quantum Field Theory (QFT) to
Quantum Gravity (QG) requires a radical change of paradigm:
in the words of Cumrun Vafa \cite{personal} it is  like going from Differential Geometry (DG)
to the much deeper Number Theory.
The fundamental principles of quantum physics
-- such as the concept of ``symmetry'' --
 take new (and much subtler) meanings in QG, and the mathematical formulation of the theory requires totally new foundations in which \emph{arithmetics} is expected to play a central role.

The swampland program \cite{Vaf,OoV} (see \cite{Rev1,Rev2} for nice reviews) aims to characterize  the effective field theories which describe the low-energy limit of consistent quantum gravities inside the much larger class of effective theories which \emph{look} consistent from the viewpoint of the traditional field-theoretic paradigm. The consistent  QG effective theories form a very sparse zero-measure subset of the naively consistent models. An effective field theory which \emph{looks}
consistent but cannot be completed to a consistent QG is said to belong to the \emph{swampland} \cite{Vaf,OoV,Rev1,Rev2}.

Up to now the swampland program has taken the form of a dozen or so (conjectural) \textit{necessary conditions} which all consistent effective theories of gravity should obey; we refer to them as the \textit{swampland conjectures} \cite{Vaf,OoV,Rev1,Rev2}.
Although these conjectures are arguably the deepest principles in physics, indeed in all science, their current  formulation looks a bit unsatisfactory: they consist of several disjoint statements, which are neither logically independent nor related by a clear web of mutual implications.
Ideally, one would like to unify the several conjectures into a \emph{single} basic physical principle.

The reason why the swampland conjectures do not look ``elegant'' is that they are  formulated entirely in the language of the old QFT/DG paradigm: they are expressed using words like ``distance'', ``volume'', ``geodesic'', \emph{etc.}, that is, in the strict  DG jargon, while one expects that a language appropriate for the gravitational paradigm should emphasize more the arithmetic aspects. It is natural to think that, if we wish 
to unify and strengthen the several conjectures 
in one more fundamental swampland principle, we must state it in a suitable `arithmetic' language. 

To build the appropriate formalism in full generality requires a lot of foundational work in which the usual field-theoretic concepts get  uplifted to subtler ones endowed with loads of new additional structures most of which arithmetic in nature \cite{Notes}. 
The story is long, technically sophisticate, and an enormous quantity of more work is required.

Luckily enough, however, there is one important special case where we can dispense with new  foundations: namely the 4d effective theories invariant under 8 supercharges, more specifically, the vector-multiplet sector of an ungauged 4d $\cn=2$ supergravity. 
At the level of the traditional QFT/DG paradigm
such a sector is described by a special K\"ahler geometry (see e.g.\! \cite{specgeo}).
In this special case
 the swampland program reduces to the simple question:

\begin{que}\label{qqyaetion}
Which special K\"ahler geometries belong to the swampland?
\end{que}   
      
\textbf{Question 1} may be re-formulated in two
well-known alternative  languages
which already take care of the relevant arithmetic aspects. This will give us two (equivalent)
\textbf{Answers} to \textbf{Question 1}.
\medskip

The second \textbf{Answer} gives a necessary criterion for quantum consistency which 
makes sense also for low-energy effective theories $\mathscr{L}_\text{eff}$ having \emph{rigid} $\mathcal{N}=2$ supersymmetry.
In $\cn=2$ QFT one can prove rigorously the validity of this criterion (see \S\S.\,\ref{relSW} and \ref{direct}): it turns out to be equivalent to requiring that the low-energy physics is described by a globally well-defined
Seiberg-Witten geometry \cite{SW1,SW2}.
The swampland criterion we propose may be seen as the more general version of Seiberg-Witten geometry
which applies both to rigid and local $\cn=2$ \textsc{susy}: in both cases a quantum-consistent
effective theory is described by such a higher Seiberg-Witten geometry. In spite of this,
the local and rigid cases differ
dramatically: in the gravity case the proposed criterion 
implies the validity of all the standard swampland conjectures \cite{Vaf,OoV}
(see \S.\,\ref{recovering}), in particular
the distance conjecture which requires \emph{infinite} towers of light states at infinity, while in the rigid 
supersymmetry the \emph{same} argument shows that at most finitely many states may
become massless in any limit (cfr.\! \S.1.3).

\subsection{First alternative language: BPS branes in $tt^*$}
One shows\footnote{\ \label{frobfoot}In the math literature this statement is sometimes  known as the (generalized) Simpson theorem \cite{simpson}.} \cite{families,tt*} that the  special K\"ahler geometries are in one-to-one correspondence with the solutions to the $tt^*$ equations \cite{tt*} for the universal deformation of an abstract 2d (2,2) chiral ring $\mathcal{R}_s$ ($s\in S$) which -- in addition to the usual properties\footnote{\ By the usual properties we mean that $\{\mathcal{R}_s\}_{s\in S}$ must be a family of finite-dimensional, commutative, associative Frobenius $\C$-algebras with unity which satisfies the axioms of a Frobenius manifold \cite{Dubrovin:1994hc}.} --
is required to be a family of \emph{local-graded} $\C$-algebras\footnote{\ Here $\mathrm{tr}_s$ denotes the Frobenius trace map, i.e.\! the TFT one-point function on the sphere $S^2$.}, that is,
\begin{gather}\label{R1}
\mathcal{R}_s=\mathcal{R}_s^0\oplus \mathcal{R}_s^1\oplus\mathcal{R}_s^2\oplus\mathcal{R}_s^3,\qquad s\in S\\
\mathcal{R}_s^q\cdot\mathcal{R}_s^{q^\prime}\subseteq\mathcal{R}_s^{q+q^\prime},\\
\mathcal{R}_s^0=\C\cdot1,\qquad \mathrm{tr}_s\colon \mathcal{R}_s^3\overset{\sim}{\to} \C.\label{R3}
\end{gather} 
The family of rings $\{\mathcal{R}_s\}_{s\in S}$ is parametrized by the universal deformation space $S$ (the ``moduli'' space, a.k.a.\! the 2d conformal manifold). For the convenience of the reader, the identification of local-graded $tt^*$ geometry with special K\"ahler geometry is reviewed in section 2.1 below (some technical detail being deferred to appendix A).
\smallskip

In the $tt^*$ language the swampland program asks for the characterization of the solutions to the $tt^*$ PDEs for local-graded chiral rings which  ``arise from physics'' out of the much bigger space of \emph{all} such solutions which is equivalent to  the space of all special K\"ahler geometries. The physical subset contains the special geometries of (2,2) SCFTs with $\hat c=3$, and in particular the special geometry of 2d supersymmetric $\sigma$-models with target space a Calabi-Yau 3-fold $X$,
which coincide with the special geometry of  the $\cn=2$ \textsc{sugra} obtained by compactification of Type IIB on $X$ \cite{families}.  

So rephrased, the $\cn=2$ swampland program becomes a special instance of the more general problem of characterizing the \emph{physical} $tt^*$ geometries. Ref.\!\!\cite{Cecotti:1992rm} studies this problem in the opposite situation in which $\mathcal{R}_s$ is semi-simple instead of local-graded (for a 2d (2,2) QFT this means gapped versus conformal).
The physical semi-simple $tt^*$ geometries are characterized by diophantine equations, that is, by arithmetic conditions \cite{Cecotti:1992rm}. 
In principle, one should be able to get the desired arithmetic formulation of the swampland conditions by a ``straightforward'' generalization of the analysis of ref.\!\!\cite{Cecotti:1992rm} to arbitrary, that is not necessarily semi-simple,
universal families of chiral rings $\{\mathcal{R}_s\}_{s\in S}$,
and then specializing the general answer  to the particular case of local-graded
rings of the form \eqref{R1}-\eqref{R3}. This leads to a first answer to \textbf{Question 1}:
\begin{ans}
A special K\"ahler geometry which does \emph{not} belong to swampland satisfies the ``straightforward''  \emph{local-graded} generalization of the $tt^*$ diophantine conditions of ref.\!\!{\rm\cite{Cecotti:1992rm}}.\footnote{\ We stress that this is still a necessary condition, not a sufficient one, albeit it is expected to ``come close'' to be sufficient.}
\end{ans}

To understand the nature of the required  ``straightforward'' generalization, let us recall the physical origin of the arithmetic conditions on the physical $tt^*$ geometries for  semi-simple chiral rings. The $tt^*$ PDEs are the consistency conditions of a system of linear differential equations
(depending on an additional twistor parameter $\zeta\in\mathbb{P}^1$) \cite{tt*,Cecotti:1992rm,dubrovin,iogaiotto}\footnote{\ Eqn.\eqref{laxtt*} is written more explicitly in eqn.\eqref{pqw12hh}.}
\be\label{laxtt*}
\Big(\nabla^{(\zeta)}+\overline{\nabla}^{(\zeta)}\Big)\Psi(\zeta)=0.
\ee
When the $tt^*$ geometry is physical, the 
solutions $\Psi(\zeta)$ to the linear problem \eqref{laxtt*} have the physical interpretation of half-BPS brane amplitudes\footnote{\ The twistor parameter $\zeta$ labels the \textsc{susy} sub-algebra which leaves the brane invariant.} \cite{Hori:2000ck}.
The BPS branes belong to a linear triangle category \cite{Hori:2000ck,categories} whose numerical homology and $K$-theoretical invariants induce a canonical integral structure on the space of solutions to \eqref{laxtt*} which implies the diophantine conditions of \cite{Cecotti:1992rm}. Nowhere in this argument one uses the fact that $\mathcal{R}_s$ is semi-simple: the $tt^*$ equations can be written in the linear form \eqref{laxtt*} for arbitrary chiral rings, and -- whenever the $tt^*$ geometry arises from a physical situation -- its brane amplitudes
$\Psi(\zeta)$ describe objects of a BPS brane category with a finitely generated Grothendieck group from which the $tt^*$ geometry inherits an arithmetic structure. The precursor papers 
\cite{Cecotti:1992rm,Hori:2000ck} focused on semi-simple chiral rings for just one reason: in that case, one may identify half-BPS branes with Lefschetz thimbles \cite{Hori:2000ck}, and then use Picard-Lefschetz theory \cite{Cecotti:1992rm} to relate them in a simple way to the BPS spectrum and wall-crossing phenomena.
The relation with the BPS spectrum is lost when $\mathcal{R}_s$ is local-graded, since there are no non-trivial BPS states in a (2,2) SCFT; nevertheless the half-BPS branes are still there and, when the corresponding special K\"ahler geometry arises from physics, they are objects in some nice linear triangle category, so that their special K\"ahler geometry inherits interesting arithmetic properties.
In facts, the arithmetics of local-graded $tt^*$ amplitudes turns out to be much richer, nicer  and deeper than the Picard-Lefshetz one for semi-simple chiral rings,
see appendix B. 

In conclusion, in the special $\cn=2$ case the swampland conditions may be easily guessed from the $tt^*$ viewpoint. However there is a second alternative language which makes the story even more straightforward. In the rest of this note we shall mostly adopt this second viewpoint, except in \S.\ref{intuitive} where we use the equivalent $tt^*$ language to give an intuitive interpretation of the swampland criterion we are proposing.

\begin{rem} Requiring that the $tt^*$ brane amplitudes $\Psi(\zeta)$ satisfy the arithmetic conditions is \emph{a priori} weaker than requiring that they are physical brane amplitudes. The arithmetic conditions essentially see only the $K$-theoretic aspects of the relevant linear triangle category of branes.
The arithmetic conditions are then expected to be \textit{mere} necessary conditions. It is conceivable that one gets \emph{sufficient} conditions for the existence of a quantum gravity UV completion by requiring that the $tt^*$ branes carry the full-fledged categorical structure  and not merely its numerical avatar. Morally speaking, the procedure of UV completion of a low-energy effective theory   looks akin to reconstructing an algebraic variety $Y$ out of its  derived sheaf category $D(Y)$.   
\end{rem}

\subsection{Second alternative language: ``motivic'' VHS}
In \textsc{Table} \ref{comparison} (page \pageref{comparison}) we compare the problem of describing the set of consistent 4d $\cn=2$ quantum gravities with the mathematical problem of describing all ``algebro-geometric objects'' of ``Calabi-Yau type''. Notice that we do not speak of Calabi-Yau manifolds (or varieties) but more generally of ``algebro-geometric objects''.
Indeed experience with moduli of Abelian varieties has shown that, in order to obtain a deep and nice theory, one has to enlarge the class of geometric objects of interest beyond the actual manifolds \cite{milne1,milne3}. The appropriate class of ``objects'' to consider is some sort of ``Calabi-Yau motives''.\footnote{\ For instance, in ref.\!\cite{draco} Van Straten proposes to consider \emph{abstract} Calabi-Yau operators \guillemotleft
as describing something like \textit{a rank four Calabi-Yau motive over $\mathbb{P}^1$}\guillemotright\ 
 even when they are not the Picard-Fuchs operators of an actual one-parameter family of 3-CY.}
 Likewise, if one wishes a classification
of consistent ``Quantum Gravities'' with nice functorial properties, one should accept some more general animals than the strict QG theories such as, for instance, consistent truncations of consistent quantum gravities. Both classification problems are somehow \emph{open,} in the sense that the precise boundaries of the class of ``natural'' objects to be called ``physical'' (resp.\! ``geometric'') is not fully understood. In particular, we do not have an explicit set of axioms which define in a precise way  what a ``Quantum Gravity'' is supposed to be.
\medskip

At the naive level of the DG paradigm, on the physical side we know that the effective theory of a 4d $\cn=2$ quantum gravity  must be a 4d $\cn=2$ (ungauged) supergravity, whose vector-multiplet sector is described in DG language by a special K\"ahler manifold $S$. Saying that the effective Lagrangian $\mathscr{L}$ is $\cn=2$ supersymmetric amounts to saying that its couplings, seen as functions of the scalar fields (``moduli''), satisfy a set of differential relations
which are the defining property of special geometry \cite{cec,stro}.

\begin{table}
\centering
\begin{tabular}{l|c|c|l}\cline{2-3}
& \textsc{Quantum Gravity} & \textsc{Algebraic Geometry} & \\\cline{2-3}
\multirow{2}{7em}{classification problem} & \multirow{2}{10em}{consistent 4d $\cn=2$ ``quantum gravities''} & \multirow{2}{14em}{\centering CY algebro-geometric ``objects'' up to deformation type}  &\\
&&&
\\\cline{2-3}
\multirow{2}{7em}{naive paradigm:
DG description} & \multirow{2}{13em}{4d $\cn=2$ ungauged \textsc{sugra} $\equiv$ special K\"ahler geometries}
& \multirow{2}{16em}{\centering real variations of Hodge structure (VHS) of weight 3 with $h^{3,0}=1$} & $\equiv$\\
&&&\\\cline{2-3}
\multirow{2}{7em}{new paradigm:
deep question} & \multirow{2}{14em}{which special geometry can be completed to a consistent QG?}
&\multirow{2}{10em}{\centering which VHS arise from CY-\textit{ish} ``motives''?} &$\overset{?}{\equiv}$\\
&&&
\\\cline{1-3}
\begin{small}\textbf{what is known}\end{small} & & &\\\cline{1-1}
examples & Type II compactifications & families of actual CY 3-folds & $\checkmark$\\\cline{2-3}
\multirow{2}{6em}{necessary conditions}& \multirow{2}{10em}{\centering the swampland \textit{conjectures}} & \multirow{2}{10em}{\centering the structure \textit{theorem}} &
$\Leftarrow$\\
&&&\\\cline{2-3}
\end{tabular}
\caption{\label{comparison}\footnotesize The comparison of two classification problems. Symbols in the right margin describe the logical relation between the two boxes in the same row: $\equiv$ means that the two are known to be strictly equivalent, $\overset{?}{\equiv}$ means that the two are conjectured to be equivalent, while $\checkmark$ means that the two are trivially the same; finally, $\Leftarrow$ means that the statement in the second column implies the ones in the first column.}
\end{table} 

Passing to the second column of  \textsc{Table} \ref{comparison}, we know that the DG description of a deformation family of CY ``objects'' is given by a polarized real variation of Hodge structure ($\R\text{-VHS}$) \cite{GII,Gbook,deligne,reva,revb,MT4,periods,book} of pure weight 3 and Hodge numbers \be\label{hhidge}
\text{$h^{3,0}=1$ and $h^{2,1}=m\equiv \dim_\C S$.}
\ee
The $\R\text{-VHS}$ is encoded in the Griffiths period map \cite{GII,Gbook,deligne,reva,revb,MT4,periods,book}
\begin{equation}\label{griper}
p\colon S\to \Gamma\big\backslash Sp(2m+2,\R)\big/[U(1)\times U(m)],
\end{equation}
where $S$ is the ``moduli'' space and $\Gamma\subset Sp(2m+2,\Z)$ the monodromy group. The map $p$ is required to satisfy a set of differential relations, known as the Griffiths infinitesimal period relations (IPR)  \cite{GII,Gbook,deligne,reva,revb,periods}. One shows that these differential relations are exactly equivalent to the ones defining $\cn=2$ supersymmetry \cite{cec,stro}, so that
 the problems in the two columns of \textsc{Table} \ref{comparison} are strictly equivalent at the naive level of the DG paradigm. We write $\mathfrak{V}$ for the set of abstract weight-3 polarized $\R\text{-VHS}$ with the Hodge numbers \eqref{hhidge}, equivalently $\mathfrak{V}$ is the set of consistent-looking $\cn=2$ special K\"ahler geometries.

It turns out that almost all 4d $\cn=2$ supergravities in $\mathfrak{V}$ cannot arise as low-energy limits of consistent theories of quantum gravity. We write $\mathfrak{M}_\text{QG}\subset \mathfrak{V}$ for the extremely  minuscule subset which \emph{can} be physically realized (the complementary set $\mathfrak{V}\setminus \mathfrak{M}_\text{QG}\subset \mathfrak{V}$ being the vast $\cn=2$ swampland).
Likewise, almost all $\R\text{-VHS}$ in $\mathfrak{V}$
do not describe actual families of algebro-geometric objects. 
Again we have a minuscule subset
$\mathfrak{M}_\text{Mot}\subset \mathfrak{V}$
of ``motivic'' $\R\text{-VHS}$ which do arise from algebraic geometry. 
In particular, a variation of Hodge structure in
$\mathfrak{M}_\text{Mot}$  should be defined over $\mathbb{Q}$ not $\R$, i.e.\!\! it should be a (pure, polarized) variation of Hodge structure in Deligne's sense (a VHS for short) \cite{reva,revb,MT4,periods,deligne}: this is a first hint of non-trivial arithmetic structures entering in the game.

Both $\mathfrak{M}_\text{QG}$ and $\mathfrak{M}_\text{Mot}$ are God-given, extremely sparse, subsets of the same naive space $\mathfrak{V}$, and both are expected to be 
characterized by subtle arithmetic-like properties.

It is natural to ask whether there is any simple  relation between these two special       
subsets. The meta-conjecture is that they indeed coincide (for a suitable choice of what we are willing to call a ``consistent quantum gravity''). 
This will follow, for instance, if we assume the string lamppost principle (SLP) \cite{slp}. 
\medskip

Characterizing $\mathfrak{M}_\text{QG}$ and $\mathfrak{M}_\text{Mot}$ inside $\mathfrak{V}$ is a fundamental problem in (respectively) Theoretical Physics and Algebraic Geometry. What is known about these two fundamental problems?

First of all we have a large supply of explicit examples in the form of moduli spaces of known Calabi-Yau 3-folds. A deformation family of 3-CYs, $\{X_s\}_{s\in S}$, yields an element of $\mathfrak{M}_\text{Mot}$ as well as an element of $\mathfrak{M}_\text{QG}$ (by compactifying Type IIB on $X_s$). The information we get from such explicit examples tautologically corroborates the idea that 
$\mathfrak{M}_\text{Mot}=\mathfrak{M}_\text{QG}$. 
\medskip

The two columns of \textsc{Table} \ref{comparison} 
look very much the same when restricted to their first five rows.
The last row describes the second piece of ``known'' information about the deep problems:
 a set of \emph{necessary conditions} an element $\mathfrak{v}\in\mathfrak{V}$ should satisfy in order to belong to the sub-set $\mathfrak{M}_\text{QG}$ or, respectively, $\mathfrak{M}_\text{Mot}$. In this last row the two columns look quite \emph{different}. In the column of Quantum Gravity the necessary conditions take the form of a dozen or so conjectural statements: they are the  usual swampland \textit{conjectures} (specialized to the case of ungauged $\cn=2$ \textsc{sugra}) \cite{Vaf,OoV,Rev1,Rev2}. In the Algebraic Geometry column we have a \emph{single} statement which, moreover, 
has the logical status of a mathematical \emph{theorem}: this fundamental result is known under the name of  the \textit{structure theorem} of (global) variations of Hodge structure, see e.g. \S.\,IV of \cite{reva}, \S.\,II.B of \cite{revb}, \S.\,III.A of \cite{MT4}, or \textbf{Theorem 15.3.14} in \cite{periods}. We review the theorem in \S.\ref{jqwaslll} below
following the two nice surveys \cite{reva,revb}. 

It is an established fact that all VHS arising from Algebraic Geometry enjoy the properties stated in the structure theorem.
In other words, the statement of the theorem yields a \emph{proven} necessary condition for a VHS
to belong to the God-given subset $\mathfrak{M}_\text{Mot}$.

In \S.(III.B.7) of ref.\!\!\cite{revb} the authors ask whether this necessary condition \emph{suffices} to completely characterize the set $\mathfrak{M}_\text{Mot}$ of ``motivic'' VHS: they answer in the negative, but their general feeling is that the structure theorem comes ``close'' to the holy Grail of actually determining $\mathfrak{M}_\text{Mot}$, thus ``almost'' completing the algebro-geometric analogue of the 4d $\cn=2$ swampland program.\footnote{\ Translated in the physical language, the basic source of non-sufficiency is that  the structure theorem does not say which points $\tau$ in the upper half-plane may be realized as periods of harmonic (3,0)-forms on \emph{rigid} 3-CY. Roughly speaking, the problem is that in this case the moduli space $S$ is a single point, and the global aspects of $S$ (central to the swampland program \cite{OoV}) do not restrict the period map $p$.  However there are good reasons to believe that in the physical context we may strengthen the statement and get rid of this inadequacy: one compactifies the 4d $\cn=2$ theory with constant graviphoton coupling $\tau=\theta/2\pi+ 4\pi i/g^2$ on a circle $S^1$. In the resulting 3d $\cn=4$ theory $S$ is promoted to a quaternionic-K\"ahler manifold  $Q_\tau$ of real dimension 4 \cite{CFG} whose geometry depends on $\tau$. Applying the 3d $\cn=4$ version of the swampland conditions to the 3d moduli space  $Q_\tau$ should yield some useful  constraint on the allowed $\tau$'s.}

What is the logical relation between the two boxes in the last row of \textsc{table} \ref{comparison}?

It is easy to see that a (global) special geometry which satisfies the VHS structure theorem automatically satisfies all the applicable swampland conjectures.\footnote{\ For some conjecture this actually holds after some minor technical refinement of the original statement.} 
This is hardly a surprise since Ooguri and Vafa \cite{OoV} used the general properties of Calabi-Yau moduli spaces as motivating examples for their conjectures.
However the inverse implication is \emph{false,} i.e.\! the structure theorem is actually a stronger constraint than the several swampland conjectures combined. 

The comparison of the two columns of  \textsc{Table} \ref{comparison} suggests
 to unify and refine the several swampland conjectures into the single statement:

\begin{ans}[The structural swampland conjecture] A (global) special K\"ahler geometry belongs to the swampland unless its underlying Griffiths period map $p$ (cfr.\!\! eqn.\eqref{griper}) satisfies the VHS structure theorem. 
\end{ans}

\textbf{Answer 2} and \textbf{Answer 1} are  mutually consistent, see \S.\ref{intuitive}.

\subsection{Relation with Seiberg-Witten theory in $\cn=2$ QFT}\label{relSW}

The structural swampland criterion holds in rigid $\cn=2$ supersymmetry as well. Here we sketch the story without entering in technical details (see \S.\,\ref{direct} for more).
The Lagrangian of a \emph{formal} $\cn=2$ low-energy effective QFT has the form
\be
\mathscr{L}= \int d^2\theta\, F(X^a) +\text{hypermultiplets}+\text{h.c.}
\ee
where $X^a$ ($a=1,2,\dots, h$) are restricted $\cn=2$ chiral superfields containing the field strengths of the $h$ Abelian gauge fields and $F$ is the holomorphic pre-potential. Let $M$ be the Coulomb branch where the scalars\footnote{\ We use the same symbol to denote the superfield and its scalar first component.} $X^a$ take value and $\widetilde{M}$ its smooth simply-connected cover, so that $M=\cg\backslash \widetilde{M}$ for some discrete group of isometries $\cg$. We have a commutative diagram\footnote{\ \textbf{Notation:} throughout this paper double-headed arrows $\xymatrix@1{\ar@{->>}[r]&}$ stand for canonical projections while hook-tail arrows $\xymatrix@1{\ar@{^{(}->}[r]&}$ represent canonical inclusions.}
\be
\begin{gathered}
\xymatrix{\widetilde{M}\ar[rr] \ar@{->>}[d] \ar@/^1.6pc/[rrrr]^(0.45){\widetilde{p}}&& \widetilde{S}\;\ar@{^{(}->}[rr]\ar@{->>}[d]&& Sp(2h,\R)/U(h)\ar@{->>}[d]
\\
M\ar[rr]\ar@/_1.6pc/[rrrr]_(0.45){p} && S \;\ar@{^{(}->}[rr] &&\Gamma\backslash Sp(2h,\R)/U(h)}
\end{gathered}\qquad \text{where}\quad\left[\begin{aligned} & \widetilde{S}\equiv \widetilde{p}(\widetilde{M})\\
& S\equiv p(M)=\Gamma\backslash \widetilde{S}\\
&\Gamma\subset Sp(2h,\Z)
\end{aligned}\right.
\ee
whose covering map $\widetilde{p}$
\be
\widetilde{p}\colon (X^1,\cdots, X^h)\mapsto \partial_a\partial_b F(X^c)\in Sp(2h,\R)/U(h)
\ee
satisfies all the axioms of a weight-1 local VHS with Hodge numbers $h^{1,0}=h^{0,1}=h$.
Again we may pose a swampland question:  describe the sub-set of formal VHS $\widetilde{p}$ which define low-energy effective Lagrangians having a UV completion which is a fully consistent QFT.

A necessary condition follows directly from QFT  first principles. The operator algebra of a 4d $\cn=2$ QFT contains a  \textsc{susy} protected subfactor, the chiral ring $\mathscr{R}_{4d}$, which is a finitely-generated, commutative, unital $\C$-algebra. Then a 4d $\cn=2$ QFT comes equipped with an affine variety, namely the spectrum $\mathrm{Spec}\,\mathscr{R}_{4d}$ of its chiral ring.
The (global) smooth Coulomb branch $M$ is the complement in $\mathrm{Spec}\,\mathscr{R}_{4d}$ of the discriminant divisor $D$  (the locus where additional degrees of freedom become massless).
Write $\mathrm{Spec}\,\mathscr{R}_{4d}=\overline{M}\setminus Y_\infty$ for $\overline{M}$ projective and $Y_\infty$ an effective  divisor\footnote{\ For explicit examples of these  algebraic-geometric gymnastics in 4d $\cn=2$ QFT see \cite{specialarithmetics}.}. Hence
\be\label{oftheform}
M=\overline{M}\setminus Y,\qquad Y=\overline{D}+Y_\infty.
\ee 
Then the group $\Gamma\subset Sp(2h,\Z)$ satisfies the properties of the monodromy group for a VHS over a quasi-projective base of the form \eqref{oftheform}, and then the period map $p$ satisfies the VHS structure theorem. (For succinct  surveys of the underlying mathematical facts, see e.g.\! \S\S.5-7 of \cite{peter} or \S.\,IV of \cite{reva}). Then we have the following

\begin{fact}\label{rigidswamp} A 4d $\cn=2$ low-energy effective field theory which has a quantum-consistent UV completion is described by a period map $p\colon M\to \Gamma\backslash Sp(2h,\R)/U(h)$ which satisfies the VHS structure theorem.
\end{fact}

When the above necessary condition is satisfied, we may realize $p$ as the period map of an algebraic family of Abelian varieties parametrized by $M$. Thus we may rephrase \textbf{Fact \ref{rigidswamp}}  by saying that
a quantum-consistent 4d $\cn=2$ effective theory necessarily arises from a Seiberg-Witten geometry
\cite{SW1,SW2}. Experience in low dimension \cite{cl1,cl2,cl3,cl4,cl5,cl6} suggests that this condition is almost sufficient. Roughly, the idea is that the UV completion of a low-energy effective theory, satisfying the structural condition, is obtained by using its Seiberg-Witten geometry to engineer the full $\cn=2$ QFT in string theory \cite{swklemm}. 

\subparagraph{Quantum-consistent $\cn=2$ gravity versus $\cn=2$ QFT.} Although the structural swampland criterions for $\cn=2$ QFT and $\cn=2$ \textsc{sugra} are \emph{in principle} identical, they lead to very different physical consequences.
In the \textsc{sugra} case the structural swampland condition implies the validity of the Ooguri-Vafa swampland conjectures \cite{OoV}  
while nothing of that sort applies in QFT.
The most striking differences in the implications of the structural swampland condition in the two physical contexts are:
\begin{itemize}
\item[A)] the \emph{volume conjecture:} in a quantum-consistent $\cn=2$ gravity theory the volume of the scalars' space is finite, while this is not true in quantum-consistent $\cn=2$ QFT, as the example of free-field theory shows;
\item[B)] the \emph{distance conjecture:} in quantum-consistent $\cn=2$ gravity an infinite tower of states becomes light at infinity in scalars' space,  whereas nothing of that sort may happen in a UV-complete QFT which has finitely-many effective degrees of freedom in \emph{all} regimes.
\end{itemize}  These dramatic differences arise from three simple elements:
\begin{itemize}
\item[1.] In $\cn=2$ QFT the VHS has weight 1 (non-zero Hodge numbers $h^{1,0}=h^{0,1}$), whereas in $\cn=2$ \textsc{sugra} the VHS has weight 3
(non-zero Hodge numbers $h^{3,0}=h^{0,3}=1$ and $h^{2,1}=h^{1,2}$). Let $\gamma_i$ be the monodromy around a prime component $Y_i$ of the \textsc{snc} divisor $Y=\sum_j Y_j$.
By the hard monodromy theorem \cite{periods}  $\gamma_i$  satisfies\footnote{\ Here we are glossing over some technicality discussed in the main text below. In particular, we have assumed   (without loss) that the monodromy group $\Gamma$ is neat.} \be\label{eqrt}
\big(\log \gamma_i\big)^{k_i}\neq0\quad \big(\log \gamma_i\big)^{k_i+1}=0\quad \text{with }1\leq k_i\leq \begin{cases} 1 & \text{QFT}\\
3 & \text{SUGRA}
\end{cases}
\ee
As we approach the support of a divisor $Y_i$ with  $k_i=3$ an infinite tower of states becomes massless. This cannot happen in QFT in view of \eqref{eqrt}. When we approach a divisor with $k_i=1$, in the gravity case either an infinite tower or just finitely many states may become massless, depending on the details of the degenerating mixed Hodge structure along $Y_i$ (see \S.\,\ref{recovering} and references therein).  The same analysis shows that in QFT at most finitely many degrees of freedom may become massless
along any $Y_i$; in showing this eqn.\eqref{hodgeversusG} below plays a crucial role.
\item[2.] In \textsc{sugra} the covering period map $\widetilde{p}$ is a local immersion (the local Torelli theorem), while in QFT  $\widetilde{p}$ just factorizes through the immersion $\widetilde{S}\to Sp(2h,\Z)/U(h)$ of the ``Torelli'' space $\widetilde{S}$ into the Siegel upper half-space. (In many examples $\widetilde{M}\equiv  \widetilde{S}$ but this does not hold, e.g., in the free theory where $\widetilde{S}$ is a point). 
 \item[3.] In consequence of 2. we have a very different relation between the  VHS Hodge metric $K_{i\bar j}$ (to be defined in \S.2.1\textbf{(VII)} below)
 and the special K\"ahler metric $G_{i\bar j}$ which enters in the scalars' kinetic terms. Indeed, the Hodge metric $K_{i\bar j}$ induces a nice positive-definite metric
 on the Torelli space $\widetilde{S}$ but not necessarily on the (covering) Coulomb branch $\widetilde{M}$. In the \textsc{sugra} case we have $\widetilde{M}\equiv \widetilde{S}$ by the Torelli property,
 and $K_{i\bar j}$ is a genuine K\"ahler metric on the covering scalars' space $\widetilde{M}$.
 The (possibly singular in the rigid case) pulled back metric on $\widetilde{M}$ is\footnote{\ By abuse of notation we use the same symbol for the metric and its pull-back, partly because in the rigid case we are mainly interested to special geometries with $\det R_{i\bar j}\not\equiv 0$ so that $\widetilde{S}$ and $\widetilde{M}$ may be locally identified.}
 (for details see \S.2.1\textbf{(VII)})
\be\label{hodgeversusG}
K_{i\bar j}= \begin{cases} R_{i\bar j} & \mathrm{QFT}\\
(m+3)G_{i\bar j}+R_{i\bar j} & \mathrm{SUGRA}
\end{cases}
\ee 
where $R_{i\bar j}$ is the Ricci tensor of $G_{i\bar j}$. If we insert back the Newton constant, the
first term $(m+3)G_{i\bar j}$ will carry a relative  factor
$M_p^{-2}$ ($M_p$ being the Planck mass) so that the two formulae agree in the limit $M_p\to\infty$.
\end{itemize}

From eqn.\eqref{hodgeversusG} we see that being at infinite distance in the physical metric $G_{i\bar j}$ has quite different meanings from the viewpoint of the intrinsic Hodge metric geometry in QFT versus  SUGRA.
In particular, the volume of the space $S\equiv\Gamma\backslash \widetilde{S}$, as computed with the Hodge metric $K_{i\bar j}$, is \emph{finite} in both cases, but this implies that the volume of the scalars' manifold $M$, as computed with the physical metric $G_{i\bar j}$, is also finite only in SUGRA
(see \S.\,\ref{recovering} for details).
\smallskip

Notice that in QFT we may locally identify $\widetilde{M}$ and $\widetilde{S}$ around points where the Ricci tensor is non-singular i.e.\! $\det R_{i\bar j}\neq0$. 
\medskip

Despite the different physical implications of the structural swampland condition in QFT and SUGRA, its general consequences hold equally in both cases. For instance, the dichotomy introduced in the next subsection holds for a quantum-consistent Torelli manifold $\widetilde{S}$ in both cases. 
Assume there is \emph{one} point in the smooth, simply-connected, irreducible cover $\widetilde{M}$ of the Coulomb branch where $\det R_{i\bar j}\neq0$; then 
either the rigid special K\"ahler manifold $\widetilde{M}$ is symmetric or has no Killing vector.
The first possibility is ruled out for a non-free $\cn=2$ QFT, since a symmetric rigid special K\"ahler manifold is necessarily flat corresponding to a free QFT  (see appendix C).

\subsection{A survey of the first applications}  From the structural conjecture one can easily extract some novel property that all consistent $\cn=2$ effective theories should satisfy (in addition to the usual swampland conditions).
\medskip 

To state the first one, we need some generality about ``symmetry'' in the QG paradigm. For more precise definitions and statements, see \S.\ref{preliminary} below.

\paragraph{``Symmetry'' in QG.}
In QG there are \emph{no} symmetries in the strict sense of the term \cite{bsei}. However we can still consider the
\emph{local} Killing vectors of the scalars' space $S$
(with respect to the metric in their kinetic terms). These local Killing vectors are \emph{never} globally defined on $S$ (so they do \emph{not} generate symmetries) but they \emph{do  constrain} the couplings in the effective Lagrangian $\mathscr{L}$ and in particular the scalars' metric. For instance, if we compactify Type II on a six-torus we get a 4d $\cn=8$ effective theory whose scalars' manifold $S_{\cn=8}$ has 133 local Killing vectors forming the Lie algebra $\mathfrak{e}_{7(7)}$, and this fact determines $\mathscr{L}$ almost uniquely \cite{cremmerjulia1}.
\smallskip

In the $\cn=2$ context, the scalars' space has the form $S=\Gamma\backslash \widetilde{S}$ for a simply-connected special K\"ahler manifold $\widetilde{S}$ and a discrete group of isometries $\Gamma\subset \mathsf{Iso}(\widetilde{S})$. The (naive) symmetry group $\mathsf{Sym}(\widetilde{S})$ of the covering special geometry $\widetilde{S}$ is a closed sub-group 
\be
\mathsf{Sym}(\widetilde{S})\subset Sp(2m+2,\R)\quad \text{where}\quad  m=\dim_\C S.
\ee
The elements of its Lie algebra $\mathfrak{sym}(\widetilde{S})\subseteq \mathfrak{iso}(\widetilde{S})$  are our \emph{local} Killing vectors on $S$.
As we shall see, when the special geometry $S$ does  not belong to the swampland the group $\mathsf{Sym}(\widetilde{S})$
and its algebra $\mathfrak{sym}(\widetilde{S})$ -- if non-trivial -- carry interesting arithmetic structures.

\paragraph{The dichotomy.} A first consequence of the structural conjecture is the following

\begin{fact}[Dichotomy]\label{cido} Assume {\bf Answer 2}. Let $S=\Gamma\backslash \widetilde{S}$ be a special K\"ahler geometry which does \textbf{not} belong to the swampland,
where $\widetilde{S}$ is its simply-connected cover and $\Gamma$ its monodromy group (a.k.a.\! duality group). Then: $S$ has a non-zero \emph{local} Killing vector $\Rightarrow$ $S$ is locally symmetric. More precisely: 
\begin{itemize}
\item \emph{either} $\widetilde{S}$ is Hermitian symmetric, its isometry group $\mathsf{Iso}(\widetilde{S})$ is one of the Lie groups in
\textsc{Table} \ref{tableiso}, $\mathsf{Sym}(\widetilde{S})=\mathsf{Iso}(\widetilde{S})$,
and $\Gamma\subset \mathsf{Iso}(\widetilde{S})$ is an \emph{arithmetic} subgroup; \item \emph{or} $\mathfrak{sym}(\widetilde{S})=0$, that is, the symmetry group
$\mathsf{Sym}(\widetilde{S})$ of the covering special geometry $\widetilde{S}$ is \emph{discrete.}
\end{itemize}
In other words: \emph{$\mathfrak{sym}(\widetilde{S})\neq0$
if and only if $S$ is a Shimura variety of rank-1 or `magic' type}.\footnote{\ By a \textit{Shimura variety} \cite{milne1,milne3,milne2,shimurakerr} we mean an \emph{arithmetic} quotient of a non-compact Hermitian symmetric manifold, i.e.\! $\Gamma\backslash G/K$ where $G$ is a non-compact semi-simple real Lie group, $K\subset G$ a maximal compact subgroup of the form $U(1)\times H$, and $\Gamma\subset G$ an \emph{arithmetic} subgroup (in the strict algebraic-group sense). We say that a Shimura variety
$\Gamma\backslash G/K$ is of \textit{`magical' type} if  $G$ is a `magical' Lie group, i.e.\! one of the groups in the right hand side of \textsc{Table \ref{tableiso}}. }
\end{fact}

\begin{table}
\setlength{\tabcolsep}{18pt}
\centering
\begin{tabular}{c|ccc}\hline\hline
rank-1 (quadratic) &  \multicolumn{3}{c}{`magic' (cubic)}\\\hline
$SU(1,m)$ & $SL(2,\R)\times SO(2,k)$ & $SL(2,\R)$ & $Sp(6,\R)$\\
& $U(3,3)$ & $SO^*\mspace{-1mu}(12)$ & $E_{7(-15)}$ \\\hline\hline
\end{tabular}
\caption{\label{tableiso}Isometry groups $G$ of symmetric special K\"ahler manifolds \cite{cremmer}: $\widetilde{S}\equiv G/K$ with $K\subset G$ a maximal compact subgroup. One has $k=m-1\geq1$. The pre-potential $\cf$ is quadratic and respectively cubic.
`Magic' $SL(2,\R)$ has rank-1 as a symmetric space, but it has ``rank-3'' as  a special geometry \cite{book} as the rest of `magic' special geometries.}
\end{table}

\begin{rem} Shimura varieties are the geometries with the richer and most interesting arithmetics \cite{milne1,milne3,milne2,shimurakerr}: indeed they are the very {paradise} of arithmetics. Their simplest (and most classical) examples are the (non-compact)  \emph{modular curves} \cite{diamond}. The fact that Shimura varieties arise so naturally provides further evidence for the relevance of arithmetics in QG. Note that last two groups in \textsc{Table \ref{tableiso}} do not correspond to `classical' Shimura varieties, that is, they are not moduli spaces of Abelian motives. 
\end{rem}

There are plenty of non-symmetric simply-connected special K\"ahler geometries $\widetilde{S}$ with a non-trivial  Killing vector, $\mathfrak{sym}(\widetilde{S})\neq0$. They and \emph{all} their quotients belong  to the swampland. In particular, 
\textbf{Fact \ref{cido}} is bad news for the author of ref.\!\!\cite{Talg}: all quotients of the several infinite series of homogeneous
special K\"ahler geometries constructed there,
however beautiful and elegant, belong to the swampland! (This negative result is already obvious at the naive DG level: see \S.\,\ref{tubedom} below).
\medskip

\textbf{Fact \ref{cido}} has important \textbf{Corollaries}.
The first one is in facts a confirmation in the
special $\cn=2$  case of a more general and profound prediction by the authors of  ref.\!\cite{vvvaf}:

\begin{corl}[Completeness of instanton corrections]\label{cominst} Suppose we have a pre-potential with an asymptotic expansion at $\infty$ of the form
\be\label{cubex}
\begin{aligned}
\cf(X^0,X^i)=&\overbrace{-\frac{d_{ijk}X^iX^jX^k}{6X^0}}^\text{classical}+\overbrace{\frac{1}{2}\,A_{IJ}X^IX^J+\frac{i\,\zeta(3)}{2(2\pi)^3}\,\chi\,(X^0)^2}^\text{loop}+\\
&+\overbrace{\sum_{\vec\lambda\in \Z_+^\rho} c_{\vec\lambda}\, \mathrm{Li}_3(e^{2\pi i\vec \lambda \cdot\vec X/X^0})}^\text{instanton corrections},\qquad\text{where }d_{ijk},\;A_{IJ}\in\mathbb{Q},
\end{aligned}
\ee
which does \textbf{not} belong to the swampland. 
Then 
\begin{itemize}
\item[\it (i)] either $c_{\vec\lambda}\neq0$ for essentially all possible $\vec\lambda$; 
\item[\it(ii)] or $c_{\vec\lambda}=0$ for all $\vec\lambda$, the ``Euler characteristic'' $\chi=0$,
$A_{IJ}\in\Z$, and $d_{ijk}z^iz^jz^k$ is  the determinant form of a rank-3 real Jordan algebra (in particular, $d_{ijk}\in\Z$).
\end{itemize}
\end{corl}

\begin{rem} The (possibly reducible) rank-3 real Jordan algebras are in 1-to-1 correspondence with the cubic  `magic' groups in the right-hand side of \textsc{Table} \ref{tableiso}.
\end{rem}

In eqn.\eqref{cubex} we labelled the various terms according to their origin in the special case where the pre-potential $\cf$ describes Type IIA compactified on a Calabi-Yau 3-fold $X$. In this particular application,  the vector-multiplet scalars $z^i\equiv X^i/X^0$ parametrize the quantum K\"ahler moduli of $X$:  
the first term is the world-sheet classical ($\equiv$ large volume) answer, the second one the perturbative loop corrections, and the last term the world-sheet instantons \cite{candelas,KKV}. In this specific geometric set-up the coefficients $d_{ijk}$ are integers (being intersection indices in cohomology \cite{CFG2}) and $\chi$ is the Euler characteristic of $X$.
\medskip

Roughly speaking \textbf{Corollary 1} says that we must have all possible world-sheet instanton corrections unless our $\cn=2$ \textsc{sugra} is a consistent truncation of a $\cn>2$ supergravity. This fact was predicted in \cite{vvvaf} as an extension of the swampland conjectures.
We make the statement more precise in the next three \textbf{Corollaries}.

\begin{corl} Suppose the pre-potential \eqref{cubex} describes the quantum K\"ahler moduli of a 2d (2,2) superconformal $\sigma$-model (with $\hat{c}=3$).
Then
the ``instanton corrections'' in eqn.\eqref{cubex} vanish \emph{if and only if} the absence of quantum corrections in the world-sheet theory is  implied by the non-renormalization theorem of a higher $(p,q)>(2,2)$ 2d supersymmetry. In this case, also the loop corrections must vanish.\footnote{\ The statement requires specifications since $A_{IJ}$ depends on the electro-magnetic duality frame. What we mean is that  $A_{IJ}$ vanishes in \emph{some} duality frame.
}
\end{corl}

\begin{corl}\label{asqw} Let $S=\Gamma\backslash \widetilde{S}$ be a special K\"ahler manifold, consistent with our swampland structural criterion, which is described locally by a \emph{strictly cubic} pre-potential $\cf_{\mspace{-2mu}\text{\rm cub}}$.
Then \emph{at least one} of the following possibility applies:
\begin{itemize}
\item $S$ is a geodesic submanifold of the scalar's space of some consistent $\cn=4$ supergravity. That is, the $\cn=2$ \textsc{sugra} is a consistent truncation of $\cn=4$ to configurations invariant under a certain discrete symmetry group;
\item $S$ is a geodesic submanifold of the scalar's space of some consistent $\cn=8$ supergravity. That is, the $\cn=2$ \textsc{sugra} is a consistent truncation of $\cn=8$ to configurations invariant under a certain discrete symmetry group; 
\item $S$ is a Shimura variety which is an arithmetic quotient of Cartan's  domain $E_{7(-25)}/[U(1)\times E_6]$. 
\end{itemize}
\end{corl}

Using results from \cite{milne1,milne3} we get
a characterization of the QG cubic scalar manifolds $S$ with well-known moduli spaces of nice algebro-geometric ``objects'':

\begin{corl} $S$ as in {\rm\textbf{Corollary \ref{asqw}}} and $\dim_\C S\neq 15,27$.
Then $S$ is the moduli space of some
family of Abelian motives. {\rm Typically, they may be realized as the special geometries of the untwisted sector moduli
of Type II/heterotic compactifications on tori \cite{CFG}.}
\end{corl}

Note that, up to finite covers, there are just five finite-volume, cubic, special K\"ahler manifolds $S$ which are locally symmetric and locally irreducible. They have complex dimensions
\be\label{listdim}
1,\ 6,\ 9,\ 15,\ \text{and}\ 27.
\ee 
So if $S$ does not belong to the swampland and is locally irreducible
with $\dim_\C S$ not in the list \eqref{listdim}, we conclude that either instanton corrections are present
or there is no asymptotic limit in the space $S$ where  $\cf$ gets cubic, i.e.\! the special geometry is ``orphan''.

\subsubsection*{Applications to Algebraic Geometry}

\textbf{Fact \ref{cido}} has important implications
in Algebraic Geometry. For instance: 

\begin{corl} Let $X$ be a Calabi-Yau 3-fold which admits a mirror $X^\vee$. Suppose that $X$ does not contain rational curves. Then $X$ is a finite quotient of an Abelian variety; in particular $X$ has Picard number $\rho=2$ or 3.
\end{corl}

The assumption of the existence of a mirror may be replaced by some much weaker technical requirement.
Modulo this aspect, \textbf{Corollary 2} answers  in the positive the \textbf{Question} posed by Oguiso and Sakurai in ref.\!\!\cite{OSaku}. We stress that under the stated assumption, the result is  mathematically fully rigorous, since it is a direct consequence of the VHS structure theorem applied to the universal deformation space of  $X^\vee$.

 \subsection{Organization of the paper}
 
 The rest of this paper is organized as follows. In section 2 we review special K\"ahler geometry from a viewpoint convenient for studying its global and arithmetic aspects, emphasizing its relations with $tt^*$ geometry and Griffiths' theory of variations of Hodge structures. For later convenience we collect here some facts about ``symmetries'' in special geometry.
 In section 3 we present (for comparative purposes) two purely differential-geometric results which go in the same directions as our main conclusions, but are significantly weaker. In subsection 3.3 we collect some facts about the geometry of \emph{symmetric} special K\"ahler geometries for later use. In section 4 we introduce our \textit{structural swampland criterion} first at an intuitive level in the language of $tt^*$ and then systematically using the mathematical framework of VHS.
Then we discuss how this statement implies, in suitable senses, the validity of the original Ooguri-Vafa swampland conjectures \cite{OoV}. In section 5
we show how the criterion may -- in principle -- be used to compute the pre-potential $\cf$ of a $\cn=2$ supergravity which does not belong to the swampland. In section 6 we prove the main result of the present paper: the dichotomy for quantum-consistent special geometries. Section 7 describes the behaviour of a quantum-consistent special geometry when we approach at infinity a MUM (maximal unipotent monodromy) point, and presents applications to Type IIA compactifications on Calabi-Yau 3-folds. In section 8 we discuss the Oguiso-Sakurai question and argue that the answer is positive. In section 9 we briefly comment on the case of Picard number 1. In appendix A we review the equivalence between the $tt^*$ PDEs and the condition that the ``Weil map'' $w$ is pluri-harmonic. In appendix B we show that the $tt^*$ brane amplitudes of \emph{superconformal} 2d (2,2) models are characterized by special arithmetic properties not shared by their massive counterparts.  


\section{Review of special K\"ahler geometry}

In this section we review the basic facts of special K\"ahler geometry mainly to fix language and notation.
By convention, ``special K\"ahler manifold'' stands for 
``\emph{integral} special K\"ahler manifold'', that is,
one whose global structure is consistent with Dirac quantization of electro-magnetic fluxes and charges. In particular the underlying variations of Hodge structures are \emph{genuine} VHS and not mere $\R$\text{-VHS}.

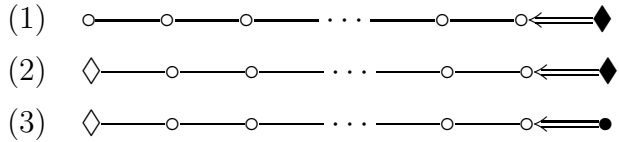
\begin{figure}
\begin{equation*}
\begin{aligned}
&(1) &&\xymatrix{\circ\mspace{-5mu}\ar@{-}[r]&\mspace{-5mu}\circ\mspace{-5mu}\ar@{-}[r]&\mspace{-5mu}\circ\mspace{-5mu}\ar@{-}[r]&\cdots \ar@{-}[r]&\mspace{-5mu}\circ\mspace{-5mu}\ar@{-}[r]&\mspace{-5mu}\circ\mspace{-5mu}&\mspace{-5mu}\blacklozenge\ar@{=>}[l]}\\
&(2) &&
\xymatrix{\lozenge\mspace{-5mu}\ar@{-}[r]&\mspace{-5mu}\circ\mspace{-5mu}\ar@{-}[r]&\mspace{-5mu}\circ\mspace{-5mu}\ar@{-}[r]&\cdots \ar@{-}[r]&\mspace{-5mu}\circ\mspace{-5mu}\ar@{-}[r]&\mspace{-5mu}\circ\mspace{-5mu}&\mspace{-5mu}\blacklozenge\ar@{=>}[l]}\\
&(3) &&
\xymatrix{\lozenge\mspace{-5mu}\ar@{-}[r]&\mspace{-5mu}\circ\mspace{-5mu}\ar@{-}[r]&\mspace{-5mu}\circ\mspace{-5mu}\ar@{-}[r]&\cdots \ar@{-}[r]&\mspace{-5mu}\circ\mspace{-5mu}\ar@{-}[r]&\mspace{-5mu}\circ\mspace{-5mu}&\mspace{-5mu}\bullet\ar@{=>}[l]}
\end{aligned}
\end{equation*}
\caption{\label{thisfigure}The diagrams of the three Klein geometries $\mathscr{D}^{(a)}_m\equiv G_m(\R)/H^{(a)}_m$, $a=1,2,3$. The nodes are decorated by a  color (white vs.\! black) and a shape (circle vs.\! lozenge). Forgetting all decorations we recover the Dynkin graph of $C_{m+1}$ which specifies the underlying  complex Lie group $G_m(\C)$. Forgetting  the shape of the nodes, one gets the Vogan diagram \cite{knapp} of the real Lie algebra $\mathfrak{g}_m$ of the group $G_m(\R)$. Forgetting the color, one gets the diagrams \cite{penrose} of the complex parabolic subgroups $P^{(a)}_m\subset G_m(\C)$, while forgetting nothing one gets those of the reductive subgroup $H_m^{(a)}\cong U(1)^\#\times L_m^{(a)}\subset G_m(\R)$ where $\#$ is the total number of lozenges (black and white) and $L_m^{(a)}$ is the real Lie group whose Vogan diagram is obtained by deleting all lozenges. For (1) and (2) the subgraph over the circle nodes is totally white, hence $H_m^{(a)}$ is compact so for $a=1,2$ the domain $\mathscr{D}^{(a)}_m$ has a homogeneous metric structure.} 
\end{figure}

\subparagraph{Notations and conventions.}
Through this paper, $S$ will denote a (integral) special K\"ahler manifold,  $\widetilde{S}$ its smooth, simply-connected cover, and $m$ its complex dimension.
We write $G_m(\mathbb{K})$ for the group  $Sp(2m+2,\mathbb{K})$ where $\mathbb{K}$ is one of the rings $\C,\R,\mathbb{Q}, \Z$. We see $G_m(\mathbb{K})$ as the group of $\mathbb{K}$-valued points in the universal   Chevalley group-scheme\footnote{\ The standard sources for Chevalley groups are  \cite{chev1,chev2}; a readable introduction is chapter VII of \cite{hump}.} of type $C_{m+1}$.
 $G_m(\C)$ and $G_m(\R)$ have (in particular) the structure of (respectively) a complex and a real Lie group, and $G_m(\Z)\subset G_m(\R)$ is a maximal arithmetic subgroup with respect to the natural structure of algebraic group over $\mathbb{Q}$ provided by\footnote{\ With some abuse, in this paper we do not distinguish between an algebraic group defined over $\mathbb{Q}$ and the group of its $\mathbb{Q}$-valued points.} $G_m(\mathbb{Q})$. We stress that all these fancy algebraic-geometric   structures on $Sp(2m+2,\R)$ have an intrinsic physical meaning, being unambiguously implied by Dirac quantization of charge/fluxes.\footnote{\ See, for instance, the wonderful computations in the appendix of \cite{soule} for the particular case of Type II compactified to 4d on a six-torus.}

\subsection{Basic geometric structures}
\label{basic}

We are concerned with three geometries which eventually will turn out to be equivalent:
\begin{itemize}
\item[(1)] $tt^*$ geometry for local-graded chiral rings as in eqns.\eqref{R1}-\eqref{R3}; 
\item[(2)] variations of (effective) polarized weight-3 Hodge structures with $h^{3,0}=1$;
\item[(3)] special K\"ahler geometry in the sense of $\cn=2$ supergravity. 
\end{itemize}
We begin by introducing their respective Klein models.

\subparagraph{(I) Three Klein geometries and their Penrose correspondence.}
Besides $G_m(\R)$ itself, we introduce the following three reductive, $G_m(\R)$-homogeneous, complex Klein geometries (called henceforth the \emph{domains}):
\begin{itemize}
\item[(1)] $\mathscr{D}_m^{(1)}\overset{\rm def}{=} G_m(\R)/U(m+1)$, i.e.\! the \textit{Siegel upper half-space}, namely the period domain for polarized weight-1 Hodge structures with Hodge number $h^{1,0}=m+1$;  
\item[(2)] $\mathscr{D}_m^{(2)}\overset{\rm def}{=} G_m(\R)/[U(1)\times U(m)]$, i.e.\!\! the \textit{Griffiths period domain} for polarized weight-3 Hodge structures with Hodge numbers $h^{3,0}=1$, $h^{2,1}=m$;
\item[(3)] $\mathscr{D}_m^{(3)}\overset{\rm def}{=} G_m(\R)/[U(1)\times G_{m-1}(\R)]$ which we call the \textit{holomorphic contact domain.}
\end{itemize}

The diagrams of these three geometries are represented in Figure \ref{thisfigure}.
Diagram (1) yields the Siegel upper half-space, i.e.\! the space of gauge couplings $\tau_{ab}$ (symmetric matrices with positive-definite imaginary part); 
(2) is obtained by making a lozenge the node associated to the Dirac representation $\boldsymbol{(2m+2)}$ of the electro-magnetic charges, while one gets (3) from (2) by eliminating all lozenges shared by (1) and (2).   
From the figure one sees that the $G_m(\C)$-homogeneous spaces
\be
\check{\mathscr{D}}^{(a)}_m\overset{\rm def}{=}G_m(\C)/P^{(a)}_m,\qquad a=1,2,3
\ee
are rational, projective (hence compact), complex manifolds\footnote{\ For these basic facts, see e.g.\! chapter 2 of \cite{russian}.}
whose homogeneous (holomorphic) bundles are
in one-to-one correspondence with the possible assignments of an integer $n_i\in\Z$ to each node of the corresponding graph in Figure \ref{thisfigure}  with the restriction that $n_i\geq0$ for the \emph{circle} nodes.
 $\check{\mathscr{D}}^{(a)}_m$ is called the \emph{compact dual} of  $\mathscr{D}^{(a)}_m$.
By construction $\mathscr{D}^{(a)}_m$ is an open domain in $\check{\mathscr{D}}^{(a)}_m$; the domain $\mathscr{D}^{(a)}_m$ gets its canonical complex structure
by restricting the one of $\check{\mathscr{D}}^{(a)}_m$, and the $G_m(\R)$-homogeneous bundles over the domain $\mathscr{D}^{(a)}$ 
likewise get their holomorphic structure from the corresponding bundles over the projective variety $\check{\mathscr{D}}^{(a)}_m$ \cite{homore}. The homogeneous vector bundles $\mathscr{V}^{(a)}\to \mathscr{D}^{(a)}_m$ are in one-to-one correspondence with the
$H^{(a)}_m$-modules $\boldsymbol{m}^{(a)}$ \cite{homore}; we write $\co(\boldsymbol{m}^{(a)})$ for the bundle over $ \mathscr{D}^{(a)}_m$ associated to $\boldsymbol{m}^{(a)}$.

Note that we have canonical fibrations between the domains (which are the restrictions of the corresponding fibrations between their compact duals) 
\be\label{canporo}
\begin{gathered}
\xymatrix{&& \mathscr{D}^{(2)}_m\equiv \frac{Sp(2m+2,\R)}{U(1)\times U(m)}\ar@/_1pc/@{->>}[lld]_{\varpi_1}\ar@/^1pc/@{->>}[drr]^{\varpi_3}\\
\mathscr{D}^{(1)}_m\equiv \frac{Sp(2m+2,\R)}{U(m+1)}\ar@{<..>}[rrrr]^{\text{Penrose correspondence}} &&&&\mathscr{D}^{(3)}_m\equiv \frac{Sp(2m+2,\R)}{U(1)\times Sp(2m,\R)}}
\end{gathered}
\ee
where $\varpi_3$ is holomorphic while $\varpi_1$ is just smooth.
 From the diagrams in Figure \ref{thisfigure}, 
we see that the correspondence
$\xymatrix@1{\mathscr{D}^{(1)}_m\ar@{<..>}[r]&\mathscr{D}^{(3)}_m}$  implied by the diagram  \eqref{canporo} (``supersymmetry'') is nothing else than the Penrose correspondence \cite{penrose}. In particular the Griffiths domain $\mathscr{D}^{(2)}_m$ embeds in $\mathscr{D}^{(1)}_m\times \mathscr{D}^{(3)}_m$.  The fiber of $\varpi_1$ is
$P\C^m=SU(m+1)/U(m)$ hence compact,\footnote{\ By general theory the fiber of $\varpi_1$ is always a complex submanifold of the Griffiths period domain.} while
the one of $\varpi_3$ is the Siegel upper half-space $\mathscr{D}^{(1)}_{m-1}=Sp(2m,\R)/U(m)$ and is non compact.

\subparagraph{(II) The holomorphic contact structure.} As already anticipated, the domain $\mathscr{D}^{(3)}_m$ carries a canonical structure of $G_m(\R)$-homogeneous, holomorphic, contact manifold.
This is obvious from the form of the isotropy group $H^{(3)}_m=U(1)\times L_m^{(3)}$. We give a more explicit description. First,
\be\label{presswzqw}
\check{\mathscr{D}}^{(3)}_m\equiv G_m(\C)/P^{(3)}_m=P\C^{2m+1},
\ee
since $P^{(3)}_m\subset G_m(\C)$ is the subgroup fixing a line $\ell$ in the $\boldsymbol{2m+2}$ fundamental representation of $G_m(\C)\equiv Sp(2m+2,\C)$.
The presentation \eqref{presswzqw} induces on $P\C^{2m+1}$ a canonical structure of holomorphic contact manifold; 
in homogeneous coordinates
$\boldsymbol{z}^a$ ($a=0,1,\dots,2m+1$) this contact structure is generated by the holomorphic one-form with coefficients in $\co(2)$
\be
\lambda\overset{\rm def}{=}Q(\boldsymbol{z},d\boldsymbol{z})\equiv \omega_{ab}\, \boldsymbol{z}^a\,d\boldsymbol{z}^b,
\ee 
where $\omega=S\otimes \boldsymbol{1}_{m+1}$ is
the standard symplectic matrix. 
Written in the homogeneous coordinates $\boldsymbol{z}^a$, the open domain $\mathscr{D}^{(3)}_m\subset P\C^{2m+1}$
is the locus where 
\be\label{quadrppos}
-i\,Q(\boldsymbol{z},\bar{\boldsymbol{z}})\equiv -i\,\omega_{ab}\,\boldsymbol{z}^a\bar{\boldsymbol{z}}^b>0.
\ee 
The one-dimensional representation of the isotropy group on the line $\ell$
defines a tautological holomorphic line bundle
\be\label{yyyyaqc2}
\cl\to \mathscr{D}_{m}^{(3)},
\ee
whose fibers are the lines $t\mspace{-2mu}\cdot\mspace{-2mu} \boldsymbol{z}\in \C^{2m+2}$. $\cl$ is clearly $Sp(2m+2,\R)$-homogeneous; it  is also equipped with the canonical, invariant, positive-definite Hermitian form
$-i\,Q(\boldsymbol{z},\bar{\boldsymbol{z}})$.
\smallskip

We recall that a holomorphic Legendre submanifold $L$ in the complex contact manifold $(\mathscr{D}^{(3)}_m,\lambda)$ of dimension $2m+1$ is a holomorphic immersion  $z\colon L\to \mathscr{D}^{(3)}_m$ such that
$z^\ast \lambda=0$ and $\dim_\C L=m$.

\subparagraph{(III) Griffiths infinitesimal period relations (IPR).} The weight lattice of $C_{m+1}$ is 
\begin{equation}
\Lambda_W=\bigoplus_{a=1}^{m+1}\Z\mspace{1mu} e_a\quad\text{
with inner product $(e_a,e_b)=\delta_{ab}$.}
\end{equation}
The simple roots (ordered as in Figure \ref{thisfigure}) are 
\be
\alpha_a=e_a-e_{a+1}\ \ \text{for }1\leq a\leq m\ \ \text{and }\ \alpha_{m+1}=2\,e_{m+1}.
\ee 
Let $\Phi\colon \Lambda_W\to\Z$ be the integral  linear form\footnote{\ $\Phi$ is the linear form which defines the relevant Hodge structure on the Lie algebra $C_{m+1}\equiv\mathfrak{g}_m$, for more details see \cite{reva,revb} (where it is written $\Psi$). The basic property of the linear form $\Phi$ defining a Hodge representation is  that $\Phi(\alpha)$ is $0\bmod4$ (resp.\! $2\bmod4$) when $\alpha$ is a compact (resp.\! non-compact) root.} 
\be\label{defffphi}
\Phi\colon \sum_{a=1}^{m+1}n_a\,e_a\mapsto -3\,n_1+\sum_{a=2}^{m+1} n_a.
\ee
$\Phi$ evaluated on the simple roots of graph (2) returns zero if the root is a white circle, it returns 2 if the node is black, and $-4$ if it is a white lozenge. $\Phi$ induces a $\Z$-grading on the complexified Lie algebra $\mathfrak{g}^\C_m$ of $G_m(\C)$. There is a unique  $\boldsymbol{Q}\in\mathfrak{g}^\C_m$ such that\footnote{\ For these statements see \S.\,X.3 in the second edition of \cite{knapp}. In eqn.\eqref{pqawe11} $\mathfrak{g}_\alpha$ stands for the root space of the root $\alpha$.} 
\be\label{pqawe11}
[\boldsymbol{Q},X_\alpha]=\frac{1}{2}\mspace{1mu}\Phi(\alpha)X_\alpha\quad \text{for}\quad  X_\alpha\in\mathfrak{g}_\alpha.
\ee 
Borrowing the $tt^*$ terminology, we shall refer to the Lie algebra element $\boldsymbol{Q}$ as the \textit{(superconformal) $U(1)_R$ charge}.
Then
\begin{align}\label{qqqqwhat}
&\mathfrak{g}_m^\C=\bigoplus_{k=-3}^3 \mathfrak{g}^{-k,k}_m, &&g^{k,-k}_m\overset{\rm def}{=}\Big\{X\in\mathfrak{g}_m^\C\;:\; [\boldsymbol{Q},X]=k X\Big\}\\
&\big[\mathfrak{g}_m^{k,-k},\mathfrak{g}^{l,-l}_m\big]\subseteq \mathfrak{g}^{k+l,-k-l}_m, &&\mathfrak{g}_m^{0,0}=\mathfrak{h}_m^{(2)}\otimes\C\equiv\mathfrak{Lie}(H^{(2)}_m)\otimes\C.
\end{align}
The adjoint representation induces on each direct summand of $\mathfrak{g}_m^\C$
the structure of a $H^{(2)}_m$-module,
hence each direct summand $\mathfrak{g}^{-k,k}_m$ defines a $G_m(\R)$-homogeneous holomorphic bundle over the Griffiths domain
\begin{equation}
\co(\mathfrak{g}^{-k,k}_m)\to \mathscr{D}^{(2)}_m.
\end{equation} 
E.g.\! the holomorphic (1,0) tangent bundle is \be
T\mathscr{D}^{(2)}_m=\co\mspace{-4mu}\left(\oplus_{k=1}^3\mspace{2mu} \mathfrak{g}^{-k,k}_m\right).
\ee Griffiths' horizontal holomorphic bundle \cite{Gbook,periods} is the homogenous bundle  $\co(\mathfrak{g}_m^{-1,1})$.
Let $S$ be a complex manifold. We say that a map
$p\colon S\to \Gamma\backslash \mathscr{D}^{(2)}_m$ satisfies Griffiths' IPR iff 
\be\label{uuuuqw12}
p_\ast(TS)\subseteq \co(\mathfrak{g}_m^{-1,1}),
\ee
where $TS$ is the holomorphic tangent bundle to $S$. In particular $p$ must be holomorphic.

\subparagraph{(IV) $tt^*$ equations and branes.} $S$ a complex manifold, $G$ a semi-simple real Lie group without compact factors $K\subset G$ a
maximal compact subgroup, and $\Gamma\subset G$ a discrete subgroup.  A map
\be
w\colon S\to \Gamma\backslash G/K
\ee
is said to be \textit{a solution to the $tt^*$ PDEs with monodromy group $\Gamma\subset G$} if it is pluri-harmonic (with respect to the symmetric metric), i.e.
\be\label{pqw12x}
\overline{D} \partial w=0.
\ee
see appendix \ref{revtau}. The $tt^*$ equations describe an isomonodromic problem \cite{ising}
and their solutions are essentially determined by
the monodromy group $\Gamma$. 
We specialize to the case relevant for special geometry, i.e.\! $G\equiv G_m(\R)$ 
\be
w\colon S\to \Gamma\backslash \mathscr{D}^{(1)}_m\equiv \Gamma\backslash G_m(\R)/U(m+1).
\ee
Let $\boldsymbol{\Psi}\colon \widetilde{S}\to G_m(\R)$ be any lift of $w$, and let $\omega\in \Lambda^1(G_m(\R))\mspace{-3mu}\otimes\mspace{-1.5mu} \mathfrak{g}_m$ be the Maurier-Cartan form of $G_m(\R)$. The usual $tt^*$ (Berry) connection $A$ and the $tt^*$ (anti)chiral ring valued 1-forms $C$, $\bar C$ \cite{tt*} are
\be\label{qwertUU}
A=\boldsymbol{\Psi}^*\omega|_{\mathfrak{k}_m},\qquad C=\boldsymbol{\Psi}^*\omega|_{\mathfrak{s}_m}\Big|_{(1,0)},\qquad \overline{C}=\boldsymbol{\Psi}^*\omega|_{\mathfrak{s}_m}\Big|_{(0,1)}
\ee
where
\be
\mathfrak{k}_m=\bigoplus_{k\ \text{even}}\mathfrak{g}^{-k,k}_m,\qquad \mathfrak{s}_m=\bigoplus_{k\ \text{odd}}\mathfrak{g}^{-k,k}_m.
\ee
It is well known that $A$, $C$, $\overline{C}$ solve the $tt^*$ equations iff $w$ is pluri-harmonic \cite{jap,dubrovin};
for the benefit of the reader we review the argument in appendix A. Note that two lifts of $w$ yield gauge-equivalent $A$, $C$, $\bar C$ describing the same $tt^*$ geometry.
We say that \textit{$w$ is a $U(1)_R$-preserving solution to $tt^*$} if, in addition,
\be\label{pqawe12}
[\boldsymbol{Q},C]=-C,\qquad [\boldsymbol{Q},\overline{C}]=\overline{C}.
\ee
Clearly this condition may be satisfied if and only if the chiral rings $\{\mathcal{R}_s\}_{s\in S}$ are local-graded ($\boldsymbol{Q}$ being their grading operator).
If we have a $U(1)_R$-preserving solution $w$, we may construct a $S^1$-family of lifts by acting with $e^{i\theta\boldsymbol{Q}}$, $\theta\in [0,2\pi)$. The $S^1$-family is then extended to a full $\mathbb{P}^1$ twistorial  family of \emph{complexified} solutions by analytic continuation\footnote{\ The inverse in eqn.\eqref{rrrrczq} is needed to convert right vs.\! left actions in order to match the conflicting conventions in $tt^*$ and Hodge theory.}
\be\label{rrrrczq}
\Psi(\zeta)\overset{\rm def}{=} \big(\boldsymbol{\Psi}\zeta^{\boldsymbol{Q}}\big)^{-1},\qquad \zeta\in\mathbb{P}^1.
\ee
Extending from the equator $|\zeta|=1$ to the full twistor sphere, $\Psi(\zeta)$ becomes valued in the complex group $G_m(\C)$ instead of the real group $G_m(\R)$ as dictated by unitarity. $\Psi(\zeta)$ then satisfies a twistorial reality condition
\be
\Psi(1/\zeta^*)=\Psi(\zeta)^*.
\ee

The maps $\Psi(\zeta)\colon \widetilde{S}\to Sp(2m+2,\C)$ are called (covering) \textit{$tt^*$ brane amplitudes.}
They satisfy the linear PDEs (the $tt^*$ Lax equations) \cite{tt*,iogaiotto} 
\be\label{pqw12hh}
\big(\nabla^{(\zeta)}+\overline{\nabla}^{\mspace{2mu}(\zeta)}\big)\Psi(\zeta)\overset{\rm def}{=}\big(d+A+\zeta^{-1}\,C+\zeta\,\overline{C}\,\big)\Psi(\zeta)=0.
\ee
\begin{cave} The brane amplitudes here are written in the \emph{real} gauge not in the more traditional holomorphic gauge \cite{tt*}. In particular, on the twistor equator $|\zeta|=1$ they are real. 
\end{cave}

Let $S=\Gamma\backslash \widetilde{S}$ for a discrete group $\Gamma$ of isometries acting freely (so $\pi_1(S)\cong\Gamma$). There is a monodromy representation $\varrho\colon \Gamma\to G_m(\Z)$ such that
the covering brane amplitudes $\Psi(\zeta)\colon \widetilde{S}\to G(\C)$ satisfy the $\Gamma$-equivariance property
\be
\xi^*\mspace{2mu}\Psi(\zeta)^{-1}=\varrho(\xi)\,\Psi(\zeta)^{-1},\qquad\forall\;\xi\in\Gamma,
\ee
which just reflects the fact that the brane amplitudes $\Psi(\zeta)\colon S\to \varrho(\Gamma)\backslash G(\R)$
are well defined on the physical space $S$.

\begin{rem} We see that the specifying a $U(1)_R$ grading in the sense of $tt^*$, i.e.\! $\boldsymbol{Q}$, is identically to the math procedure of constructing a Hodge representation in Hodge theory (by the integral form $\Phi$) \cite{reva,revb,MT4}. In other words, the equator of the $tt^*$ twistor sphere gets identified with the Deligne circle of Hodge theory at the reference point in the period domain.
\end{rem}

\subparagraph{(V) Definition of special geometry.} $S$ a complex manifold of dimension $m$, and
$\Gamma\subset G_m(\Z)$ a subgroup to be called the \textit{monodromy}  (or \textit{$U$-duality}) group. Modulo commensurability (i.e. up to finite covers) we may (and do) assume $\Gamma$ torsion-free (and  even neat \cite{morris,borel}).
\medskip

Consider the commutative diagram
\be\label{gasqwe4534}
\begin{gathered}
\xymatrix{S \ar@/^2.5pc/[rrrr]^w\ar[rr]_{p}\ar@/_1.5pc/[rrd]_{z} && \Gamma\backslash \mathscr{D}^{(2)}_m\ar@{->>}[rr]_{\varpi_1} \ar@{->>}[d]^{\varpi_3} && \Gamma\backslash \mathscr{D}^{(1)}_m\\
&& \Gamma\backslash \mathscr{D}^{(3)}_m}
\end{gathered}
\ee
where, as before, the double-headed arrows $\xymatrix@1{\ar@{->>}[r]^{\varpi_a}&}$ stand for the Penrose canonical projections \eqref{canporo}. 
The arrow $w$ describes $tt^*$ geometry,
the arrow $p$ is Griffiths' period map describing the variation of Hodge structure (VHS), and the arrow $z$ encodes the supergravity geometry.

Work by many people \cite{simpson,jap,dubrovin,families,cec,stro,tt*,ising,bgrif,bgrif2} may be summarized into:

\begin{defthm} The following three conditions are essentially equivalent for the  commutative diagram \eqref{gasqwe4534}:
\begin{itemize}
\item[(1)] $p$ satisfies Griffiths' IPR;
\item[(2)] $w$ is a $U(1)_R$-preserving solution to $tt^*$;
\item[(3)] $z$ is a Legendre submanifold.
\end{itemize}
If one (hence all) of the conditions is satisfied, we say that the diagram \eqref{gasqwe4534} is a \emph{special} (K\"ahler) \emph{geometry.}   
\end{defthm}

\noindent In this statement ``essentially equivalent'' means that if we are given any one of the three arrows $p$, $w$ or $z$, which satisfies the relevant condition in the above list, one may complete, in an essentially unique way, the  commutative diagram by arrows satisfying the stated condition. For instance, $p$ is canonically the $1^\text{st}$ prolongation of $z$ \cite{bgrif,bgrif2}.
Then $z_\ast$ is everywhere of maximal rank, and
 the pull-back
of the curvature of the Chern connection of the tautological bundle $\cl$ \eqref{yyyyaqc2}, equipped with the Hermitian form \eqref{quadrppos}, is positive and hence it defines a K\"ahler metric on $S$ (the so-called special K\"ahler metric).

If we are given $p$, we construct the other two arrows, $w$ and $z$, by compositing it with the two canonical projections;  so the most convenient description of a special geometry  is through its period map $p$. Then \textbf{Question \ref{qqyaetion}} is more conveniently stated in the form

\begin{que} Which period maps $p$, satisfying the IPR, do arise from quantum-consistent $\cn=2$ theories of gravity?
\end{que}

\subparagraph{(VI) The local pre-potential $\cf$.}
Let $\widetilde{S}$ be the simply-connected cover of $S$. 
$p$ and $z$ may be lifted to their respective covering maps
\be\label{jazsqwrr}
\begin{gathered}
\xymatrix{\widetilde{S} \ar[rr]^{\widetilde{p}}\ar@/_1.5pc/[rrd]^{\widetilde{z}} && \mathscr{D}^{(2)}_m\ar@{->>}[d]\\
&& \mathscr{D}^{(3)}_m}
\end{gathered}
\ee 
 with $\widetilde{L}\equiv \widetilde{z}(\widetilde{S})\subset \mathscr{D}_{m}^{(3)}$
 a holomorphic Legendre submanifold.
We write 
\be
\pi\colon \mathscr{D}_{m}^{(3)}\to \mathbb{P}^m,
\qquad \pi\colon (\boldsymbol{z}^0:\boldsymbol{z}^1:\dots: \boldsymbol{z}^{2m+1})\mapsto (\boldsymbol{z}^{m+1}:\boldsymbol{z}^{m+2}:\cdots: \boldsymbol{z}^{2m+1})
\ee for the projection on the last $(m+1)$ homogeneous coordinates $\boldsymbol{z}^a$ of $\mathscr{D}_{m}^{(3)}$. Following tradition
\cite{dewitproyen}, we shall denote the homogeneous coordinates of $\mathbb{P}^m$ as $X^I$ ($I=0,1,\dots,m$).
A generic Legendre submanifold $\widetilde{z}(\widetilde{S})$ intersects transversally the generic fiber of $\pi$; hence the differential of the map $\pi \widetilde{z}\colon \widetilde{S}\to \mathbb{P}^m$
is generically of maximal rank. 
The usual coordinates on $\mathbb{P}^m$ then may serve as local holomorphic coordinates in (small enough) local charts $U\subset \widetilde{S}$. We may write the local Legendre manifold $\widetilde{z}(U)\subset \widetilde{L}$  in terms of the homogeneous coordinates on $\mathscr{D}_{m}^{(3)}$ in the form  \cite{bgrif}
\be
\begin{gathered}
\widetilde{z}(U)=\Big\{(\cf_I, X^J)\in \mathscr{D}^{(3)}_{m}\colon (X^J)\in \pi\widetilde{z}(U)\Big\}\subset \widetilde{L},\\
 \cf_I\overset{\rm def}{=}\partial_{X^I}\mspace{-1mu}\cf\mspace{-0.2mu}(X^J),\quad\cf(X^J)\equiv\frac{1}{2} X^I\cf_I(X^J),\quad I,J=0,1,\cdots,m, 
\end{gathered}
\ee
for some holomorphic Hamilton-Jacobi function  $\cf(X^I)$, homogeneous of degree 2, which is called the \emph{pre-potential} in the $\cn=2$ supergravity context.
\emph{A priori}
the map $\pi\widetilde{z}\colon \widetilde{S}\to \mathbb{P}^m$ may be multi-to-one, and then each local branch of this  multi-cover has its own local pre-potential. 
In other words: the prepotential $\cf$ may be \emph{multivalued} (even after replacing $S$ by its universal cover $\widetilde{S}$).
Since the swampland conditions refer to the
\emph{global} geometry of the scalars' manifold $S$, we cannot ignore the issue of the several branches of $\cf$.
This item will be discussed in paragraph \textbf{(IX)}. As a preparation, we recall that the Hamilton-Jacobi function $\cf$ is  not invariant \emph{in value} under a $G_m(\mathbb{K})$
($\mathbb{K}=\R,\mathbb{Q},\Z$) rotation but only under
its block-diagonal subgroup $GL(m+1,\mathbb{K})\subset G_m(\mathbb{K})$:
\be\label{changeF}
\begin{split}
&\text{the action (in value) of}\ \ \gamma\equiv \begin{pmatrix}
{A_I}^K & B_{IL}\\
C^{JK} & {D^J}_L
\end{pmatrix}\in Sp(2m+2,\mathbb{K})\ \ \text{on $\cf$ is}\\
&\cf\xrightarrow{\ \gamma\ } \cf+ \frac{1}{2}\cf_I({A_K}^IC^{K\mspace{-2mu}J})\cf_J+\cf_I(C^{KI}B_{K\mspace{-2mu}J})X^J+\frac{1}{2}X^I(B_{KI}{D^K}_{\mspace{-2mu}J})X^J.
\end{split}
\ee

\subparagraph{(VII) WP and Hodge metrics.} 
The diagram \eqref{gasqwe4534} induces on $S$ a countable family of distinct K\"ahler metrics; the first two are relevant for our present purposes. The first one, $G_{i\bar j}$, is called the Weil-Petersson (WP) metric (a.k.a.\! the normalized $tt^*$ metric \cite{tt*}); its K\"ahler form is the pull-back \emph{via} $z$ of the curvature of the positive line bundle \eqref{yyyyaqc2}. This is the special K\"ahler metric which appears in the $\cn=2$ \textsc{sugra} kinetic terms \cite{cec,stro}.
The second one, $K_{i\bar j}$,
called in the math literature the Hodge metric \cite{GII,Gbook,periods,3cymmet}, was introduced and studied from the $tt^*$ viewpoint in \cite{ising,Bershadsky:1993ta}. Its K\"ahler form $\omega_H$ is the pull-back \emph{via} $p$ of the positive curvature of the Griffiths' canonical line bundle $\cl_\text{can}$ over the period domain $\mathscr{D}^{(2)}_m$ \cite{GII}. In terms of the matrix one-forms defined in \eqref{qwertUU} one has \cite{ising}
\be
\omega_H=\frac{i}{2}\,\mathrm{tr}\big(C\wedge \overline{C}\big).
\ee 
For a special K\"ahler geometry of complex dimension $m$, the relation between the two K\"ahler metrics is \cite{ising,3cymmet}
\be\label{cjqer}
K_{i\bar j}=(m+3)\mspace{2mu} G_{i\bar j}+R_{i\bar j}.
\ee
The two metrics $G_{i\bar j}$ and $K_{i\bar j}$ have the same isometry group.
The metric $K_{i\bar j}$ is better behaved than $G_{i\bar j}$ in several senses. E.g.\!\! the  curvatures of $K_{i\bar j}$ have much better non-positivity properties: for a \emph{general} special geometry the Ricci curvature of $K_{i\bar j}$ is negative and  away from zero by a known constant \cite{3cymmet}, and its  curvature tensor is non-positive in the Griffiths sense \cite{GII,Gbook,periods} (as well as in the Nakano sense \cite{Notes}). 

\subparagraph{(VIII) Torelli properties.} In the particular case where our special geometry describes an actual family of Calabi-Yau 3-folds, it is convenient to identify the simply-connected cover $\widetilde{S}$ of $S$
with the completion, with respect to the Hodge metric $K_{i\bar j}$, of the \emph{Torelli space} $\ct^\prime$ which is  the moduli of polarized and marked CY's, 
see refs.\!\!\cite{glo-torelli,glo-torelli1,glo-torelli2}.
$\widetilde{S}$
 should not be confused with the Teichm\"uller space
$\ct$ i.e.\! the universal cover of the uncompleted space $\ct^\prime$.
 $\widetilde{S}$ is biholomorphic to a  domain of holomorphy in $\C^m$ \cite{glo-torelli1}. We stress that $\widetilde{S}$ is metrically complete with respect to the Hodge metric but not necessarily with respect to the Weil-Petersson one. The implications of this fact for physics will be discussed in \S.\,\ref{recovering}\textbf{(2)}. 
  
The analysis of  \cite{glo-torelli,glo-torelli1,glo-torelli2} extends from this geometric situation to general $\cn=2$ supergravities with quantized electro-magnetic fluxes.
Then the ``global Torelli'' theorem of refs.\!\!\cite{glo-torelli,glo-torelli1,glo-torelli2}
asserts that the lifted period map $\widetilde{p}$ in eqn.\eqref{jazsqwrr} is injective.\footnote{\ We stress that the strong version of global Torelli is \emph{false} for Calabi-Yau 3-folds, as shown by the Aspinwall-Morrison quintic counterexample \cite{AM,AM2}.} 
Again, this holds in any naively-consistent $\cn=2$ \textsc{sugra} including the ones in the swampland. More precisely, we take injectivity of $\widetilde{p}$ as part of our definition of an \emph{integral} special geometry (by convention, all our geometries are integral).
Most of our arguments are independent of this property: they rest on  the CY strong  
\emph{local} Torelli theorem (see e.g.\! \textbf{Theorem 16.9} in ref.\!\!\cite{torelli}) which clearly holds for all \emph{formal} special K\"ahler geometry.\footnote{\   For these facts expressed in the ``old'' supergravity language see \cite{dewitproyen}.} 

Henceforth we shall  
identify $\widetilde{S}$ with its isomorphic image $\widetilde{p}(\widetilde{S})$. Then $S=\Gamma\widetilde{S}$, where the monodromy group $\Gamma$ acts on $\widetilde{S}\equiv \widetilde{p}(\widetilde{S})$ in the obvious way.

\subparagraph{(IX) Branching of $\cf$.}
Let us consider the pull-back of the local holomorphic function $\cf$ to $\widetilde{S}$ (again written $\cf$).
When defined as in \textbf{(VIII)},  $\widetilde{S}$ is biholomorphic to a  domain of holomorphy  in $\C^m$; if the holomorphic function $\cf$ has no singularity in $\widetilde{S}$, it has a uni-valued global analytic extension to all $\widetilde{S}$. $\cf$ may become singular
only at a locus $B\subset \widetilde{S}$ where  the Legendre sub-manifold $\widetilde{L}$ ceases to be transverse to the fibers of $\pi$. On the locus $B$  the holomorphic function
$\det(\partial_I\partial_J \cf)$ has a pole; hence $B$ has complex codimension at least 1 in $\widetilde{S}$. 
Indeed, $B\subset \widetilde{S}$ is just the branch locus of the holomorphic covering map 
\be
\pi\mspace{1mu}\varpi_3\mspace{2mu}\widetilde{p}\colon \widetilde{S}\to \mathbb{P}^m
\ee 
where distinct branches of the cover coalesce together. Around such a locus $\cf$ should become multivalued to represent the several local branches of the Legendre cover $\widetilde{L}\to \mathbb{P}^m$.
Since $B$ has codimension $\geq1$, 
the several branches of $\cf$ in $\widetilde{S}\setminus B$  all arise from the (multi-valued) analytic continuation of a \emph{single} holomorphic function.

\subsection{Special K\"ahler symmetries}\label{preliminary}

The special K\"ahler space of a general $\cn=2$ supergravity has the form 
\be\label{kasqw1}
S=\Gamma\backslash \widetilde{S}
\ee 
with $\widetilde{S}$ a smooth, Hodge-metric complete, simply-connected, special K\"ahler manifold. 
It is convenient to work with the covering $\cn=2$ \textsc{sugra}, whose scalar fields take value in  the covering special K\"ahler geometry $\widetilde{S}$, and  think of $\Gamma$ as a discrete symmetry of the covering \textsc{sugra} which we must gauge in order to get the actual low-energy quantum gravity with moduli space $S$. In other words: two field configurations in the covering theory which differ by the action of $\Gamma$ are regarded as the \emph{same} physical configuration in QG. We write $\mathsf{Iso}(\widetilde{S})$ for the Lie group of holomorphic\footnote{\ One can show that the connected component of the isometry group of a special K\"ahler manifold is made of holomorphic isometries, i.e.\! all Killing vectors are holomorphic.} isometries of the special K\"ahler metric on $\widetilde{S}$,
and 
$\mathsf{Sym}(\widetilde{S})\subseteq \mathsf{Iso}(\widetilde{S})$ for the \emph{naive}\footnote{\ By the \emph{naive} symmetry group we mean the would-be symmetry group of the supergravity  theory with the same Lagrangian but  whose vector fields $A^I$ are regarded as one-form fields instead of Abelian connections (the difference being that in the second case the electro-magnetic fluxes are quantized \emph{\'a la} Dirac).} symmetry group of the covering $\cn=2$ \textsc{sugra}. 
For all naive symmetry $\xi\in\mathsf{Sym}(\widetilde{S})$ we have
\be
\widetilde{p}(\xi\mspace{-2mu}\cdot\mspace{-2mu} s)= S_\xi \cdot \widetilde{p}(s)\qquad s\in \widetilde{S}, \quad \xi\in \mathsf{Sym}(\widetilde{S})\subseteq
\mathsf{Iso}(\widetilde{S}),\quad S_\xi\in G_m(\R),
\ee
where the naive symmetry $\xi$ acts on the electro-magnetic field strengths by the \emph{real}  duality rotation $S_\xi$.
The map $\xi\mapsto S_\xi$ yields a Lie group homomorphism 
\be
\kappa\colon\mathsf{Sym}(\widetilde{S})\to G_m(\R)\equiv Sp(2m+2,\R)
\ee 
whose kernel is trivial by ``global Torelli'', so that the naive symmetry group 
$\mathsf{Sym}(\widetilde{S})$ gets identified with a closed subgroup of $G_m(\R)$, that is: \textit{all naive symmetries of $\widetilde{S}\subset\mathscr{D}^{(2)}_m$ arise from automorphisms of the ambient space} 
\be
\mathscr{D}^{(2)}_m\equiv Sp(2m+2,\R)/[U(1)\times U(m)].
\ee
The Lie algebra $\mathfrak{sym}(\widetilde{S})$
of $\mathsf{Sym}(\widetilde{S})$
is then seen as a sub-algebra of $\mathfrak{g}^\R_m\equiv\mathfrak{sp}(2m+2,\R)$. In facts, the injectivity of
\be
\mathfrak{sym}(\widetilde{S})\to\mathfrak{sp}(2m+2,\R)
\ee 
already follows from strong local Torelli (\!\!\cite{torelli} \textbf{Theorem 16.9}). 

 The actual symmetry group of the covering special K\"ahler geometry $\widetilde{S}$ 
is the subgroup of the naive one which is consistent with the quantization of the electro-magnetic fluxes, i.e.
\be\label{hhhasqwp}
\mathsf{Sym}(\widetilde{S})_\Z\overset{\rm def}{=} \mathsf{Sym}(\widetilde{S})\cap G_m(\Z).
\ee 
The discrete gauge group satisfies $\Gamma\subset \mathsf{Sym}(\widetilde{S})_\Z$. The normalizer $\cn(\Gamma)$ of $\Gamma$ in $\mathsf{Sym}(\widetilde{S})_\Z$ is an ``emergent'' symmetry of the QG, and the quotient
\be
\cn(\Gamma)/\Gamma
\ee
is the ``emergent'' global symmetry of the effective theory. Here ``emergent'' means that $\cn(\Gamma)/\Gamma$ is a symmetry of the low-energy effective theory truncated at the 2-derivative level, which may (and should, according to the swampland conjectures) be explicitly broken by higher derivative couplings.\footnote{\ Compactification of Type IIB on the Aspinwall-Morrison 3-CY \cite{AM} is such an example with an ``emergent'' $\Z_5$ symmetry at the 2-derivative level  which is explicitly broken by higher derivative operators.}

It is natural and convenient to consider the larger group of \emph{rational} symmetries
\be
\mathsf{Sym}(\widetilde{S})_\mathbb{Q}\overset{\rm def}{=} \mathsf{Sym}(\widetilde{S})\cap G_m(\mathbb{Q}).
\ee
Rational symmetries are not symmetries in the strict sense of the world, since they do not preserve the Dirac symplectic lattice $\Lambda$ of electro-magnetic charges. However, given a rational symmetry $\xi\in \mathsf{Sym}(\widetilde{S})_\mathbb{Q}\setminus \mathsf{Sym}(\widetilde{S})_\Z$ there is a finite-index sublattice $\Lambda_\xi\subset \Lambda$ such that
\be
\xi \,\Lambda_\xi\subset \Lambda,
\ee
and hence $\xi$ sets a correspondence between the  subsector of the theory whose electro-magnetic charges are in the sublattice $\Lambda_\xi$
and the subsector with electro-magnetic charges in the sublattice $\xi\,\Lambda_\xi$ (having the same index in $\Lambda$ as $\Lambda_\xi$). The correspondence
$\Lambda_\xi\leftrightarrow \xi\,\Lambda_\xi$ leaves invariant the classical physics, so it is a kind of sector-wise symmetry.  More importantly, its existence has observable consequences. 
For instance (assuming $\xi$ is a rational symmetry of the full theory and not just of its 2-derivative truncation\,!!) it implies that the entropy of an extremal black holes with charge $v\in\Lambda_\xi$ is equal to the entropy of the corresponding extremal black hole with charge $\xi v\in \xi\mspace{2mu}\Lambda_\xi$.
\medskip

We see $G_m(\mathbb{Q})\equiv Sp(2m+2,\mathbb{Q})$ as  a connected (linear) algebraic group defined over the field $\mathbb{Q}$ \cite{milnebook}. The subgroup 
\be
\mathsf{Sym}(\widetilde{S})_\mathbb{Q}\subset G_m(\mathbb{Q})
\ee 
then contains a maximal connected $\mathbb{Q}$-algebraic subgroup 
\be
\mathsf{Q}(\widetilde{S})\subset \mathsf{Sym}(\widetilde{S})_\mathbb{Q}.
\ee 
This allows to introduce the notion of the $\mathbb{Q}$-Lie algebra of ``infinitesimal symmetries'' of the special K\"ahler geometry $\widetilde{S}$,  namely the $\mathbb{Q}$-Lie algebra  $\mathfrak{q}(\widetilde{S})$ of the algebraic group $\mathsf{Q}(\widetilde{S})$ \cite{milnebook}  
\be
\mathfrak{q}(\widetilde{S})\subset \mathfrak{sp}(2m+2,\mathbb{Q})\quad \text{and}\quad \mathfrak{q}(\widetilde{S})\otimes_\mathbb{Q}\mspace{-2mu}\R\subseteq \mathfrak{sym}(\widetilde{S}).
\ee
We shall see that in a quantum-consistent
special geometry the last inclusion is an equality.

\section{Warm-up: Two weak results in the DG paradigm}\label{dummies}

Before addressing the issues of main interest, we state two `elementary' results which are weaker versions of 
\textbf{Fact \ref{cido}} and \textbf{Corollary \ref{cominst}}, respectively.
The merit of these statements is that they may be proven remaining inside the naive DG paradigm
by assuming the swampland conjectures in the original Ooguri-Vafa form \cite{OoV}. In facts, it suffices to impose that the moduli space $S$ is non-compact of finite volume (the other conditions then hold automatically in this special context).
\medskip

The material in this section will be only marginally used in the rest of the paper, and can be omitted in a first reading. However, comparison of the swampland structural criterion with some well-known phenomena in Differential Geometry may help to understand the nature and effects of the QG  `arithmetic' paradigm (see \textbf{Comments on the proof} below).
\medskip

 The first result is weaker than \textbf{Fact \ref{cido}} because it makes additional geometric assumptions on the special K\"ahler manifold $S$, so that it applies only to a small  \emph{subclass} of special geometries. The second result holds in full generality, but it is  less precise than \textbf{Corollary \ref{cominst}} in a significant way.
 
 Since both results will be subsumed by the stronger \textbf{Fact \ref{cido}}, in this section we shall be rather sketchy with the proofs, and be cavalier with several fine points.   

\subsection{A restricted class of special K\"ahler geometries}

As before, $\mathsf{Iso}(\widetilde{S})$ stands for the Lie group of {holomorphic}
isometries of the special K\"ahler manifold $\widetilde{S}$,
and $\mathfrak{iso}(\widetilde{S})$ for its Lie algebra generated by the holomorphic Killing vectors.
All Killing vectors of $\widetilde{S}$ belong to $\mathfrak{iso}(\widetilde{S})$.

\begin{fact}\label{weakre} Let $S=\Gamma\backslash \widetilde{S}$ be a special K\"ahler manifold which is non-compact with $\mathsf{Vol}(S)<+\infty$. Here $\widetilde{S}$ is its smooth simply-connected cover.\footnote{\ Recall that -- replacing $S$ by a finite cover, if necessary -- we are assuming $\Gamma$ to be neat \cite{morris,borel}, hence torsion-free. Then $\widetilde{S}$ is the universal cover and $\pi(S)=\Gamma$.} \emph{Assume, in addition, that the Riemannian sectional curvatures $S$ are non-positive.} Then we have the dicothomy:
\begin{itemize}
\item if $\widetilde{S}$ has a non-trivial Killing vector (i.e.\! $\mathfrak{iso}(\widetilde{S})\neq0$) then
$\widetilde{S}$ is Hermitian symmetric: $\widetilde{S}=G/K$ (with $G$ a group in \textsc{Table} \ref{tableiso}) and $\Gamma\subset G$ is an \emph{arithmetic} subgroup;
\item otherwise $\mathsf{Iso}(\widetilde{S})$ is a discrete group of rank 1, and $\Gamma\subset \mathsf{Iso}(\widetilde{S})$ is a finite-index subgroup.
\end{itemize}
\end{fact}

\noindent The definition of the \emph{rank} of an abstract group $\Gamma$ is   recalled in the box on page \pageref{gggggamma}.

\begin{rem} \textbf{Fact \ref{weakre}} is \textbf{Fact \ref{cido}} except that we restrict to the very small sub-class of \emph{non-positively curved} special K\"ahler manifolds. In this restricted class of manifolds the conclusion is stronger since we have the extra information that in the non-symmetric case the monodromy group $\Gamma$ has rank 1. \end{rem}

Similar results hold under other special assumptions on $\widetilde{S}$: see  e.g.\!  \textbf{Theorem 1.7}
in \cite{lu*}. 

\begin{figure}
\centering
\doublebox{\begin{minipage}{330pt}
\begin{scriptsize}
\textbf{\underline{The rank of an abstract group $\Gamma$}} \cite{hadam} \vskip 6pt

Let $\Gamma$ be an abstract group. For $\sigma\in\Gamma$, we write $Z_\Gamma(\sigma)$ for its centralizer in $\Gamma$. We write $A_i(\Gamma)$ for the set of the elements $\sigma\in\Gamma$ such that $Z_\Gamma(\sigma)$ contains a free Abelian subgroup of rank $\leq i$ as a subgroup of finite index. Then
$$
r(\Gamma):=\min\!\bigg\{i\colon \text{there exist finitely many }\gamma_j\in\Gamma\ \text{such that }\Gamma=\bigcup_j \gamma_j \cdot A_i(\Gamma)\bigg\}
$$
and set
$$
\mathrm{rank}\,\Gamma:=\sup\big\{ r(\Gamma^*)\colon \Gamma^*\subseteq \Gamma\ \text{is a finite index subgroup}\big\}.
$$
Properties:
\begin{itemize}
\item Commensurable groups have the same rank;
\item if $\Gamma=\Gamma_1\times\cdots\times \Gamma_s$ is a product, then $\mathrm{rank}\,\Gamma=\sum_j \mathrm{rank}\,\Gamma_j$;
\item free groups have rank 1;
\item if $G_\mathbb{Q}$ is a $\mathbb{Q}$-algebraic group of $\mathbb{Q}$-rank $\boldsymbol{r}$ and $\Gamma\subset G_\mathbb{Q}$ is an arithmetic subgroup \cite{morris}, then $\mathrm{rank}\,\Gamma=\boldsymbol{r}$.
\end{itemize}

\end{scriptsize}
\end{minipage}
}
\label{gggggamma} 
\end{figure}

\medskip

By definition \cite{jost}, under the present special assumption, the simply-connected Riemannian manifold $\widetilde{S}$ is a Hadamard manifold without Euclidean factors.\footnote{\ The non obvious part of this statement is that $\widetilde{S}$ is metrically complete. By construction, $\widetilde{S}$ is complete for the Hodge metric; from eqn.\eqref{cjqer} we see that, when the sectional curvatures of $G_{i\bar j}$ are non-positive, completeness with respect to the Hodge metric implies completeness with respect to the WP metric.} Then \textbf{Fact \ref{weakre}} is a direct consequence of the \emph{rigidity theorems}  in Differential Geometry for finite-volume quotients of Hadamard manifolds, see refs.\!\!\cite{hada0,hada1,hada2,hadam}. When $S$ is locally symmetric these theorems reduce to the Mostow rigidity theorem \cite{morris,mostow} while the properties of $\Gamma$ follow from the Margulis  super-rigidity theorem \cite{morris,margulis}. 

We recall some definitions and the simplest rigidity statement \cite{hada1}.
A \emph{lattice} $\Gamma$ in a Hadamard manifold $H$ is a subgroup $\Gamma\subset \mathrm{Iso}(H)$ such that $\mathsf{Vol}(H/\Gamma)<\infty$. A lattice $\Gamma\subset\mathrm{Iso}(H)$ is said to be \emph{reducible} iff the manifold $M=H/\Gamma$ has a \emph{finite} cover which is reducible as a Riemannian manifold. 

\begin{thm}[P. Eberlein \cite{hada1}]\label{Hadatigidity} Let $H$ be a Hadamard manifold without Euclidean factors\,\footnote{\ \textsc{Proposition 4.4} of \cite{hada1} has no assumption on the Hadamard manifold $H$ but requires $\Gamma$ not to contain Clifford translations. Our assumption of no Euclidean factor implies the last condition, see \textsc{Theorem 2.1} of \cite{hada0}.} and let $\Gamma$ be an irreducible lattice in $H$. Then either:
\begin{itemize}
\item[\rm (1)] $\mathrm{Iso}(H)$ is discrete, $\Gamma$ has finite index in $\mathrm{Iso}(H)$ and $H$ is irreducible;
\item[\rm (2)] $H$ is isometric to a symmetric space of non-compact type.
\end{itemize}
\end{thm}

\begin{rem} In order to study the interplay between the isometry group and the global properties of a special K\"ahler geometry, it is convenient to replace its usual WP K\"ahler metric $G_{i\bar j}$ with the Hodge one $K_{i\bar j}$, eqn.\eqref{cjqer}. Since the two metrics $G_{i\bar j}$ and $K_{i\bar j}$ have the same isometry group,
 the statement and proof of \textbf{Fact \ref{weakre}} holds equally well for both metrics.  However as mentioned in \S.\,\ref{basic}\textbf{(VII)}
 $K_{i\bar j}$ has better chances of having non-positive curvatures. 
Unfortunately having a non-positive curvature in Griffiths sense is much weaker than the condition required in the above \textbf{Theorem}, i.e.\! non-positive Riemannian sectional curvatures, so that, while consideration of $K_{i\bar j}$ may seem to enlarge the class of special geometries to which \textbf{Fact \ref{weakre}} applies, it still consists of a very small portion of all special K\"ahler geometries.\footnote{\ See \textbf{Remark} at the end of this subsection.} 
\end{rem}

\begin{proof}[Proof of {\rm\textbf{Fact \ref{weakre}}}] Let $\widetilde{S}$ be the smooth simply-connected cover of our special K\"ahler manifold $S$ having (by hypothesis) non-positive Riemannian sectional curvatures. We need to show the following claim: the existence of a discrete subgroup $\Gamma\subset\mathrm{Iso}(\widetilde{S})$
such that $\mathsf{Vol}(\Gamma\backslash \widetilde{S})<\infty$ implies that either $\mathfrak{iso}(\widetilde{S})=0$ or $\widetilde{S}$ is symmetric.  
The simply-connected \emph{reducible} special K\"ahler manifolds were  classified in \cite{ferrara} (see also \cite{book}):
they are all symmetric spaces, corresponding to the non-simple Lie groups in \textsc{Table \ref{tableiso}}. Therefore we may assume without loss that $\widetilde{S}$ is an \emph{irreducible} Hadamard manifold without Euclidean factors.
The claim then follows from the \textbf{Theorem} quoted above.  
\end{proof}

\begin{cproof} The basic ingredient of the argument  is a \emph{rigidity theorem} in Differential Geometry. Already in ref.\!\!\cite{OoV} it was observed that the swampland conjectures are easy to prove for negatively-curved moduli spaces. Unfortunately, in general the special K\"ahler manifolds are \textbf{not} quotients of Hadamard spaces, and no useful general rigidity theorem for them is available \textit{while remaining in the DG paradigm.} Going to the ``new'' paradigm, we may replace in the above argument  the DG rigidity theorem by the subtler rigidity theorems of Hodge theory \cite{GII,Gbook,periods}
whose most convenient and powerful formulation is the VHS \emph{structure theorem.}
This fancier rigidity theorem has an arithmetic flavor and holds (conjecturally) for \emph{all} special K\"ahler geometries\footnote{\ In facts, more generally for all ``motivic'' VHS for arbitrary values of the Hodge numbers $h^{p,q}$. This fact allows to extended the swampland structural criterion to effective theories more general than $\cn=2$ supergravities \cite{Notes}.} which do \emph{not} belong to the swampland, independently of the sign of their curvature,
while having essentially the same implications as  the usual  rigidity theorems of Differential Geometry and Lie group theory. 
\end{cproof}

We close this subsection mentioning some other results for the non-positively curved case. We recall that the rank, $\mathsf{rank}(M)$ of a complete Riemannian manifold is defined as \cite{hadam}
\be
\mathsf{rank}(M)= \min_{|v|^2=1} \dim P\mspace{-4mu}J_v
\ee
where $P\mspace{-4mu}J_v$ is the vector space of parallel Jacobi fields along the geodesic with initial velocity vector $v$. If $M$ is locally symmetric this
reproduces the standard rank (i.e.\! the $\R$-rank of its  isometry group).

\begin{thm}[Burns--Spatzier \cite{hada2}]
Let $S$ be a complete connected Riemannian manifold of \emph{finite volume} and  non-positive sectional curvature without Euclidean factors. Then $S$ has a \emph{finite} cover which splits as a Riemannian product of rank 1 spaces and a locally symmetric space. 
\end{thm}

For a manifold $S$ as in the  \textbf{Theorem} we may define the rank of its fundamental group $\pi_1(S)$ in purely  abstract group-theoretical fashion, see the box on page \pageref{gggggamma} 
for the detailed definition.
 Under the present hypothesis on $S$:
\begin{pro}[Ballmann--Eberlein \cite{hadam}] $S$ a complete Riemannian manifold with non-positive sectional curvatures
and $\mathsf{Vol}(S)<\infty$. Then
\begin{equation}\mathrm{rank}\,\pi_1(S)=\mathrm{rank}(S).
\end{equation} 
In particular, $\mathrm{rank}(S)$ is a homotopy invariant.
\end{pro}

\begin{corl} Let $S$ be a special K\"ahler manifold with $\mathsf{Vol}(S)<\infty$ and non-positive Riemannian curvatures. Then either $S$ is locally symmetric or $\mathrm{rank}\,\pi_1(S)=1$.
\end{corl}

\begin{rem} Inverting the logic, if a finite-volume special K\"ahler manifold $S$ is not locally symmetric, while the monodromy group $\Gamma$ has rank $\geq2$, we conclude that its special K\"ahler metric $G_{i\bar j}$, as well as its Hodge metric $K_{i\bar j}$, should have  at some point a positive sectional curvature. This remark applies, say, to the dimension-101 moduli space of quintic hypersurfaces in $\mathbb{P}^4$ where $\mathrm{rank}\,\Gamma=102$ \cite{beauv}.
This observation implies that non-symmetric special K\"ahler manifolds with everywhere non-positive Riemannian sectional curvatures -- if they exist at all -- are quite rare.
\end{rem}

\subsection{The quantum corrections cannot be trivial}\label{tubedom}
   
 One has the following
 
 \begin{fact}\label{tubearg} $\widetilde{S}$ a special K\"ahler manifold with strictly cubic pre-potential 
 \be\label{cubcubc}
 \cf_{\rm{cub}}=-\frac{d_{ijk}X^iX^jX^k}{3!\,X^0}.
 \ee
  There exists a subgroup
 $\Gamma\subset\mathrm{Iso}(\widetilde{S})$ such that
 $S\equiv\Gamma\backslash \widetilde{S}$ has finite volume \emph{if and only if} $S$ is locally symmetric \emph{(hence an arithmetic quotient of a Hermitian symmetric space $G/K$ with $G$ a group  in the right column of \textsc{Table} \ref{tableiso}).}
 \end{fact}  
 
 In particular \textbf{Fact \ref{tubearg}} says that, in a Type IIA compactification on a Calabi-Yau 3-fold $X$, the K\"ahler moduli $S$
 should receive \emph{some} quantum correction,
 unless $S$ is locally symmetric, since otherwise the purely
 cubic classical pre-potential belongs to the swampland. However, \textbf{Fact \ref{tubearg}} says nothing about the \emph{nature} of these quantum corrections: it does not even rule out that perturbative corrections suffice. In this sense it is much weaker than \textbf{Corollary 1}.
 
 As anticipated, we shall be very cavalier with the argument. Consider first the case where the cover $\widetilde{S}$ is homogeneous (hence one of the special K\"ahler manifolds constructed in \cite{Talg}).
 Then $\widetilde{S}=\mathrm{Iso}(\widetilde{S})/I$,
 with $I\subset \mathrm{Iso}(\widetilde{S})$ the compact isotropy group of a chosen base point. Then,
 for all $\Gamma\subset \mathrm{Iso}(\widetilde{S})$,
 \be
 \mathsf{Vol}(\Gamma\backslash \widetilde{S})=
 \mathbf{Vol}\big(\Gamma\backslash\mathrm{Iso}(\widetilde{S})\big)\big/\mathbf{Vol}(I),
 \ee
 for some invariant measure $\mathbf{Vol}$
 on the Lie group $\mathrm{Iso}(\widetilde{S})$.
 Since $I$ is compact, $\mathbf{Vol}(I)<\infty$,
 and $S$ may not belong to the swampland only if there exists a $\mathrm{Iso}(\widetilde{S})$-invariant measure $\mathbf{Vol}$ and a discrete subgroup
 $\Gamma\subset \mathrm{Iso}(\widetilde{S})$ such that 
 $\mathbf{Vol}(\Gamma \backslash \mathrm{Iso}(\widetilde{S}))<\infty$. 
 
 The case of $S$ not locally homogeneous behaves in a similar way provided we can show that, after possibly  the excision of a zero-measure subset, the orbits of $\mathrm{Iso}(\widetilde{S})$ are regular, so that we may define a nice generic-orbit space $Y$ and
 \be\label{tobjust}
 \mathsf{Vol}(S)=\mathbf{Vol}\big(\Gamma\backslash \mathrm{Iso}(\widetilde{S})\big)\times \int_Y d\mu
 \ee
 where $d\mu$ is the appropriate  ``Fadeev-Popov'' induced measure on the generic-orbit space $Y$. We shall dispense with the technicalities involved in the justification of eqn.\eqref{tobjust}. Granted it, we again have that the Lie group $G\equiv \mathrm{Iso}(\widetilde{S})$ must admit  a \emph{lattice} that is, a pair $(\mathbf{Vol}, \Gamma)$ where $\mathbf{Vol}$ is an invariant measure on $G$ and $\Gamma\subset G$ a discrete subgroup with
$\mathbf{Vol}(\Gamma\backslash G)<\infty$. We recall a well known fact:

\begin{lem}[See e.g.\! \cite{lattice}]\label{unilllema} A Lie group $G$ admits a lattice if and only if it is unimodular.
\end{lem}

\subsubsection*{Tube domains in $\C^m$}

The covering cubic special K\"ahler manifold $\widetilde{S}$ described by \eqref{cubcubc} is a special case of a \emph{tube domain}\footnote{\ Also known as Siegel domain of the first kind \cite{gindik}.}. A tube domain $T(V)$ has the form
\be
T(V)=\Big\{z\in \C^m\colon \mathrm{Im}\,z\in V\Big\}\subset \C^m,
\ee
with $V\subset \R^m$ a strict, convex, open cone.
Its group of holomorphic automorphisms, $\mathsf{Aut}(T(V))$, contains
the group
\begin{equation}
\R^m\rtimes \mathsf{Aut}(V)\qquad (v, A)\colon z\mapsto Az+v,
\end{equation}
where $\mathsf{Aut}(V)\subset GL(m,\R)$ is the group of linear automorphisms of  the ambient $\R^m$ preserving the cone $V$.

When the cubic special K\"ahler manifold $T(V)$ arises as the \emph{classical} (large-volume) limit of the K\"ahler moduli of Type IIA compactified on a 3-CY $X$, the Abelian subgroup 
\be
\R^m\subset \mathsf{Aut}(T(V))
\ee 
describes the axionic shifts of the $B$-field by harmonic $(1,1)$-forms on $X$ -- a classical continuous symmetry which remains unbroken to \emph{all} orders in  world-sheet perturbation theory (while instanton corrections typically break it down to a discrete subgroup) -- 
while $\mathsf{Aut}(V)$ corresponds to the geometric automorphism of the classical K\"ahler cone.

Note that for all convex cone $V$ we have at least the symmetry $\R_{>0}\subset \mathsf{Aut}(V)$ corresponding in Type IIA to overall rescalings of the K\"ahler form. This classical rescaling symmetry may be  broken already by the loop corrections. Then, for all
cubic tube domain $\widetilde{S}\equiv T(V)$, we have
\begin{equation}\label{xxyyw}
\R_{>0}\ltimes \R^m \subset \mathsf{Iso}(\widetilde{S})\subset \mathsf{Aut}(\widetilde{S}).
\end{equation}
\textbf{Fact \ref{tubearg}} follows from \textbf{Lemma \ref{unilllema}} together with the following

\begin{lem}\label{dicooVV} $V\subset \R^m$ a strict, convex, open cone. Let $G$ be a connected Lie group such that
\be\label{lllwq}
\R_{>0}\ltimes \R^m\subseteq G\subseteq \mathsf{Aut}(T(V)). 
\ee
Then $G$ is either semi-simple or non-unimodular.
In the first case $G\equiv \mathsf{Aut}(T(V))$
and $T(V)$ is a Hermitian symmetric space. 
\end{lem}

\begin{proof} Consider the (real) Lie algebra
$\mathfrak{aut}(T(V))$ of $\mathsf{Aut}(T(V))$; its elements are holomorphic vector fields $f(z)^i\partial_{z^i}$ whose coefficients $f(z)^i$ are polynomials in the coordinates $z^i$ of $\C^m$ \cite{tube}.  Let $\partial\equiv z^i\partial_{z^i}$ be the Euler vector field, which generates the overall scaling symmetry $\R_{>0}$. The Lie algebra $\mathfrak{aut}(T(V))$ is graded by the adjoint action of $\partial$ \cite{tube}
\be
\mathfrak{aut}(T(V))= \mathfrak{aut}_{-1}\oplus\mathfrak{aut}_0\oplus\mathfrak{aut}_{+1},\qquad X\in\mathfrak{aut}_\ell\Leftrightarrow [\partial, X]=\ell\,X,
\ee
with $\R^m\equiv \mathfrak{aut}_{-1}$, $\partial\in\mathfrak{aut}_0\equiv\mathfrak{aut}(V)$,
while $\dim \mathfrak{aut}_{+1}\leq m$ with equality if and only if $T(V)$ is a symmetric domain \cite{tube}.
The Lie algebra $\mathfrak{g}$ of $G$ is also graded by $\partial$ 
\be \mathfrak{g}=\mathfrak{g}_{-1}\oplus\mathfrak{g}_0\oplus\mathfrak{g}_1\subseteq \mathfrak{aut}(T(V)),
\ee 
with  
\be
\mathfrak{g}_{-1} = \R^n,\quad \partial\in\mathfrak{g}_0,\quad \mathfrak{g}_{+1}=\mathfrak{g}\cap\mathfrak{aut}_{+1}.
\ee 
A necessary condition for $G$ to be unimodular, is that the trace of the adjoint action of $\partial$ on its Lie algebra $\mathfrak{g}$ vanishes. Since
\be
\mathrm{tr}_\mathfrak{g}\partial= -m+\dim\mathfrak{g}_{+1},
\ee
this happens iff $\mathfrak{g}_{+1}\equiv\mathfrak{aut}_{+1}$ and 
$\dim \mathfrak{aut}_{+1}=m$. Hence $T(V)$ is a symmetric domain and\footnote{\ Through the paper the symbol $G^\circ$ denotes the connected component of the group $G$.} $G=\mathsf{Aut}(T(V))^{\mspace{-1mu}\circ}$, the connected component of the holomorphic automorphism group of $T(V)$.  
\end{proof}

\begin{proof}[Proof of {\bf Fact \ref{tubearg}}] If $S=\Gamma\backslash \widetilde{S}$ does not belong to the swampland for some $\Gamma$, $\mathrm{Iso}(\widetilde{S})$ should be unimodular by \textbf{Lemma \ref{unilllema}}. Then, in view of eqn.\eqref{xxyyw}, \textbf{Lemma \ref{dicooVV}} applied to the Lie group $G\equiv\mathrm{Iso}(\widetilde{S})$ requires $\widetilde{S}\equiv T(V)$ to be biholomorphic to a  symmetric domain and $\mathrm{Iso}(\widetilde{S})=\mathsf{Aut}(T(V))$. On $T(V)$ there is a unique metric (up to overall normalization) which is invariant under $\mathsf{Aut}(T(V))$ -- the symmetric one -- and since $\mathsf{Iso}(\widetilde{S})=\mathsf{Aut}(T(V))$, the special K\"ahler metric on $\widetilde{S}$ is proportional to the unique $\mathsf{Aut}(T(V))$-invariant metric, that is, the special K\"ahler manifold $\widetilde{S}$ is symmetric.
\end{proof}

In particular, all quotients of the homogeneous non-symmetric special K\"ahler geometries constructed in \cite{Talg} -- all of which have cubic pre-potentials -- belong to the swampland.
\medskip

\subsection{Symmetric rank-3 tube domains}\label{3tube}

As already anticipated the symmetric rank-3 tube domains $T(V)$ are precisely the symmetric special K\"ahler manifold with a cubic pre-potential \eqref{cubcubc}. Their isometry groups $\mathsf{Iso}(T(V))$ are listed in the right-hand side of \textsc{Table} \ref{tableiso}. For later reference let us describe the symmetric 
cubic form $d\colon \odot^3\mspace{-1mu}\R^m\to \R$ (a.k.a.\! the Yukawa coupling).
For the cubic space $SL(2,\R)/U(1)$ the cubic form is given by $d(x)=x^3$.
For the spaces with 
$
\mathrm{Iso}(T(V))=SL(2,\R)\times SO(2,k)$
($k\geq 1$), it is given by 
\be\label{det1}
d(x_I)= 3\,x_0\,A_{ij}\mspace{2mu}x^ix^j,\qquad i,j=1,2,\dots, k,
\ee
where $A_{ij}$ is an integral quadratic form of signature $(1,k-1)$.
In the remaining four cases $\R^m$ is identified with the $\R$-space $\mathsf{Her}_3(\mathbb{F})$ 
 of Hermitian $3\times 3$ matrices over the normed $\R$-algebra $\mathbb{F}$, where $m\equiv \dim_\R \mathsf{Her}_3(\mathbb{F})=3(1+\dim_\R\mathbb{F})$,  see \textsc{Table} \ref{octionions}: 
 
 \begin{table}
$$
\begin{tabular}{c|cccc}\hline\hline
$\phantom{\Big|}\mathbb{F}$ & $\mathbb{R}$ & $\C$ & $\mathbb{H}$ & $\mathbb{O}$\\
$m$ & $6$ & $9$ & $15$ & $27$\\
\hline\hline
\end{tabular}
$$
\caption{\label{octionions}The normed $\R$-algebras. $m$ is the real dimension of 
$\mathsf{Her}_3\mspace{-0.5mu}(\mathbb{F})$.
}
\end{table}
 
\be\label{33matrix}
\R^{3(1+\dim_\R\mathbb{F})}\cong \mathsf{Her}_3(\mathbb{F})\overset{\rm def}{=} \left\{\begin{pmatrix} x_1 & z_{12} & z_{13}\\
z_{12}^* & x_2 & z_{23}\\
z_{13}^* & z^*_{23} & x_3
\end{pmatrix}\!\!,\ \ x_i\in\R,\ z_{ij}\in\mathbb{F}\right\}
\ee
where $z\mapsto z^*$ is the usual $\R$-linear complex conjugation in $\mathbb{F}$. 
As a \emph{real} symmetric cubic form $d\colon \odot^3\mspace{-1mu}\R^m\to \R$
is $6!$ times the determinant of the $3\times 3$ matrix \eqref{33matrix}. The determinant is given by the usual formula, but we need to pay attention to the order in which we perform the products since $\mathbb{F}$ is neither commutative nor associative in general. The correct expression over $\R$ is \cite{vinberg}
\be\label{det2}
6!\,d(x_i, z_{jk})= x_1x_2x_3-x_2z_{13}z^*_{13}-x_1z_{23}z^*_{23}-x_3z_{12}z^*_{12}+z_{12}(z_{23}z_{13}^*)+(z_{13}z^*_{23})z^*_{12}.
\ee

\subparagraph{Crucial \textit{caveat}.} In the applications to Quantum Gravity we are actually concerned with 
an \emph{integral} cubic form $d_\mathbb{Z}\colon \odot^3\mspace{-2mu}\mathbb{Z}^m\to \mathbb{Z}$ whose underlying \emph{real} cubic form is $\R$-equivalent to \eqref{det2}. 
Clearly, there are many inequivalent integral cubic forms $d_\mathbb{Z}$
with the same underlying real form $d$.
Their difference would look physically irrelevant according to the field-theoretic paradigm -- since they all define the same \emph{naive} Lagrangian $\mathscr{L}$ -- but in the Quantum Gravity setting different integral forms define physically distinct theories because of flux quantization. Indeed, a cubic special geometry may or may not belong to the swampland depending on the particular integral cubic form $d_\mathbb{Z}$ we choose in the real equivalence class of $d$. For instance, consider the quantum-consistent special geometry obtained from Type IIA compactified on the manifold $X\equiv (E\times K)/\Sigma$ where $E$ is an elliptic curve, $K$ a K3 surface and $\Sigma$ a freely acting symmetry such that $h^{2,0}(X)=0$:
there are precisely  eight diffeomorphism classes of such complex manifolds, see section 8.
The integral cubic form $d_\Z$ is a topological invariant of $X$.  The generic integral cubic form, equivalent over $\R$ to one of these eight quantum-consistent forms, belongs to  the swampland. This quantum inconsistency will not be detected by the usual swampland criteria, but the structural one is strong enough to discriminate between the different integral forms $d_\Z$.

\section{The swampland structural  criterion}\label{jqwaslll}

\subsection{The rough physical idea}
A basic physical principle in QFT states that the possible interactions are severely restricted by the gauge symmetries of the theory. 
One expects this principle to extend to QG, except that in QG the gauge symmetries enjoy two novel features:
\begin{itemize}
\item[\textit{(i)}] we can have infinite discrete gauge groups; 
\item[\textit{(ii)}] the gauge group $\cg$ carries additional arithmetic structures, that is, $\cg$ is a group-object in some subtler category $\mathfrak{G}$ of QG ``symmetries''. 
Typically $\mathfrak{G}$ is endowed with a canonical forgetful functor $\mathfrak{G}\to\boldsymbol{\mathsf{Lie}}$ which keeps only the underlying Lie-theoretic structures of the ``usual'' description of symmetry in field theory.
\end{itemize} 
One would guess 
the existence of a fundamental QG principle of the rough form:

\begin{rprin}
An effective field theory belongs to the swampland unless: 
\begin{itemize}
\item[\rm(a)] its gauge group belongs to the appropriate subtler category, $\cg\in\mathfrak{G}$;
\item[\rm(b)] all couplings are consistent with the gauge symmetry $\cg$.
\end{itemize}
Here \emph{consistent} means compatible with respect to \emph{all} the structures carried by $\mathfrak{G}$ not just the underlying usual consistency in
$\boldsymbol{\mathsf{Lie}}$. As a consequence, {\rm(b)} is a much stronger constraint than the corresponding fact in QFT, and $\cg$ ``almost'' determines the theory.\footnote{\ Indeed it uniquely determine the effective theory iff $\cg$ is ``big enough''.} 
\end{rprin}

We believe that a suitable technical version of this rough statement is the proper swampland condition \cite{Notes}. In the special case of an effective theory which is an ungauged 4d $\cn=2$ supergravity this physical principle takes the form  of the VHS structure theorem. 

Indeed the VHS structure theorem, when interpreted as a statement about the 4d $\cn=2$ supergravities having an algebro-geometric origin, has precisely the form
of the above \textbf{rough physical principle}.
The theorem consists of two parts: 
\begin{itemize}
\item[(a)] a list of properties that the $U$-duality group $\Gamma$ should satisfy in  all motivic VHS;
\item[(b)] strong constraints on the period map $p$ (i.e.\!\! on the special K\"ahler geometry) arising from $\Gamma$.
\end{itemize}
Item (b) follows from (a) together with the VHS rigidity theorems which may be regarded as the ultimate  extension of Seiberg's principle of the ``power of holomorphy'' \cite{power}.

\subsection{Intuitive view: $tt^*$ solitons and brane rigidity}\label{intuitive}

Although we mostly use the VHS language, it is worthwhile to briefly discuss the basic idea from the (equivalent)  viewpoint of $tt^*$ geometry which offers a different physical intuition of the same deep facts.
This subsection is written having in mind readers which prefer simple physical arguments to abstract mathematics. All others are invited to jump directly to the next subsection. Given its purpose, this subsection is sketchy and rough; but we stress that it is just a rephrasing of the technical proof of the structure theorem in Hodge theory, so it can be made fully rigorous by adding a little bit of math pedantry.

\subparagraph{An auxiliary $\sigma$-model.} We consider a space-time of the form $X\times \R$, where $\R$ is time and $X$ is a Riemannian manifold with metric $\gamma$ and coordinates $x^i$. To make the story a bit shorter, we assume from the start $X$ to be K\"ahler: in our application $X$ will be identified with the universal cover $\widetilde{S}$ of the \textsc{sugra} moduli space $S$. On this space-time we introduce a classical $\sigma$-model with target space some Riemannian manifold $Y$ with metric $g$ and coordinates $y^a$.
The (classical) static solitonic particles of this model are the (time-independent) solutions to the equations
\be\label{sssoleqert}
\gamma^{i \bar k} D_i \mspace{1mu}\partial_{\bar k}\mspace{2mu} y^a=0,
\ee
which extremize the functional
\be\label{ffwqma}
m=\int_X \sqrt{\gamma}\,\gamma^{i\bar k} g_{ab}\,\partial_i y^a\,\partial_{\bar k}y^b,
\ee
whose value at the extremum is the mass of the soliton. Thus a soliton of the $\sigma$-model is nothing else than a harmonic map $y\colon X\to Y$.
Since $X$ is K\"ahler, there is a special class of such solitons -- called \textit{pluri-harmonic} maps -- in which the full matrix vanishes  
\be\label{eqertjj}
D_i \mspace{1.5mu}\partial_{\bar k}\mspace{2mu} y^a=0
\ee
 not just its trace as in eqn.\eqref{sssoleqert}.
 We stress that the special solitons have the property of being solutions of the equations of motion \emph{for all} choices of the K\"ahler metric $\gamma_{i\bar j}$ on $X$; in the physical applications one says that these special configurations are protected by \textsc{susy}.\footnote{\ This is related to the deep fact uncovered in ref.\!\cite{iogaiotto}:
the $tt^*$ geometry defines a dual 2d (2,2) system, with field space the coupling-constant  space $X$ of the original theory, whose $tt^*$ geometry is the original 2d (2,2) model. Both $tt^*$ geometries are invariant under deformations of the $D$-terms of the dual model.} In this setting, all harmonic maps $y\colon X\to Y$ \eqref{sssoleqert} satisfy Simpson's
 Bochner-formula
 \cite{simpson,simp2}
 \be\label{hasqwrt}
 \begin{split}
 D_{\bar k}D_{\bar l}&\big(\gamma^{\bar k i}\gamma^{\bar l j}g_{ab}\, \partial_i y^a\,\partial_j y^b\big)=\\
 &=\gamma^{\bar k i}\gamma^{\bar l j}
g_{ab}\big(D_{\bar l}\partial_i y^a\big)\mspace{-2mu}\big(D_{\bar k}\partial_j y^b\big)-
\gamma^{\bar k i}\gamma^{\bar l j}\,R_{abcd}\, \partial_i y^a\, \partial_j^b\, \partial_{\bar k}y^c\,\partial_{\bar l}y^d
\end{split}
 \ee
 where $R_{abcd}$ is the Riemann tensor on $Y$.
 
 We now specialize to the case in which the target space $Y$ is locally isometric to an irreducible symmetric space $G/K$ of non-compact type. We write $\mathfrak{g}=\mathfrak{k}\oplus \mathfrak{p}$ for the corresponding Lie algebra decomposition, and identify $TY\cong \mathfrak{p}$ through the Maurier-Cartan form $\omega$ on $G$. Then $dy\in T^*\mspace{-2mu}X\otimes y^*TY$ is identified with
 \be
 C+\overline{C}= y^*\omega|_\mathfrak{p},
 \ee
where $C\equiv C_i\mspace{2mu} dx^i$ (resp.\! $\overline{C}\equiv \overline{C}_{\bar l}\mspace{2.7mu}d\bar x^{\bar l}$) are the two summands in the type decomposition 
\be
T^*\mspace{-2mu}X\otimes \C\cong T^{1,0}\mspace{-1mu}X\oplus T^{0,1}\mspace{-1mu}X.
\ee 
 $C_i$, $\overline{C}_{\bar l}$ are seen as matrices acting on some representation space $V$ of $G$.
Under this identification, eqn.\eqref{hasqwrt} becomes
\be
\overline{D}^i \overline{D}^j\mathrm{tr}(C_iC_j)=\|\overline{D}C\|^2+ \mathrm{tr}\Big([C_i,C_j][\overline{C}_{\bar k},\overline{C}_{\bar l}]\Big)\,\gamma^{i\bar l}\,\gamma^{j\bar k}.
\ee
Both terms in the \textsc{rhs} are point-wise positive semi-definite; the first one vanishes iff $\overline{D}C=0$, i.e.\!\! if and only if the harmonic map $y$ is actually pluri-harmonic, while the second term vanishes iff $[C_i,C_j]=0$, that is, iff
the matrices $C_i$ generate a \emph{commutative} algebra. In appendix \ref{revtau} we show, following \cite{jap,dubrovin}, that the first condition implies the second one, and then the full set of $tt^*$ PDEs \cite{tt*}
\be\label{rrqw123}
\begin{gathered}
DC=\overline{D}C=D\overline{C}=\overline{D}\mspace{3mu}\overline{C}=C\wedge C=\overline{C}\wedge \overline{C}=0\\
D^2=\overline{D}^2=D\overline{D}+\overline{D}C+C\wedge\overline{C}+\overline{C}\wedge C=0.
\end{gathered}
\ee
In other words, when $Y=G/K$ the special $\sigma$-model solitons \eqref{eqertjj} are exactly the solutions to the $tt^*$ PDEs.

\subparagraph{$tt^*$ $\boldsymbol{Q}$-solitons.}
The $tt^*$ arrow $w$ in the diagram \eqref{gasqwe4534} is a special soliton in the above sense for $G=Sp(2m+2,\R)$ and $V=\boldsymbol{2m+2}$, see eqn.\eqref{pqw12x}. From that diagram we see that a special K\"ahler geometry in the sense of $\cn=2$ supergravity is nothing else than a special soliton of our auxiliary $\sigma$-model with the extra feature of preserving the $U(1)_R$ charge $\boldsymbol{Q}$
to ensure that the unital commutative algebra $\mathcal{R}$ generated by the $C_i$ is a local-graded chiral ring, cfr.\! eqn.\eqref{pqawe12}. A $\boldsymbol{Q}$-preserving special soliton is nothing else than a VHS in Griffiths' sense: indeed, in view of eqn.\eqref{qqqqwhat}, eqn.\eqref{pqawe12} is equivalent to the Griffiths' IPR \eqref{uuuuqw12}.
To ensure that the VHS has the correct Hodge numbers
we must require\footnote{\ Otherwise we get the $tt^*$ geometry of a family of superconformal 2d (2,2) models whose $\hat c$ is not necessarily 3.} that the $U(1)_R$ charge $\boldsymbol{Q}$ has the correct spectrum
\begin{equation}
\begin{split}
\text{spectrum of }&\text{$\boldsymbol{Q}$ acting on $V$}=\\
&=\left\{ \frac{1}{2}\,\Phi(\lambda)\; \colon\;\lambda\in\text{(weights $\boldsymbol{2m+2}$)}\right\}\equiv\left\{\mp\frac{3}{2},\pm\frac{1}{2},\cdots,\pm \frac{1}{2}\right\},
\end{split}
\end{equation}
 see eqns.\eqref{defffphi},\eqref{pqawe11}.
We call such $\boldsymbol{Q}$-preserving special solutions to \eqref{sssoleqert} \textit{$tt^*$ $\boldsymbol{Q}$-solitons.} 

All $\cn=2$ special geometries are $tt^*$ $\boldsymbol{Q}$-solitons (and viceversa), so asking which special geometries belong to the  swampland is equivalent to asking which formal  $tt^*$ $\boldsymbol{Q}$-solitons are unphysical.

\begin{rem} When $X$ is compact, the integral of the total derivative in the \textsc{lhs} of eqn.\eqref{hasqwrt} vanishes, and all solitons are automatically special, i.e.\! solutions of the $tt^*$ PDEs
\eqref{rrqw123}. 
\end{rem}

\subparagraph{Physical $tt^*$ $\boldsymbol{Q}$-solitons.}
In physics we are not interested in \emph{all} the solutions to eqn.\eqref{sssoleqert} but only in solitons which satisfy two additional conditions: (1) they have a finite mass $m<\infty$, and (2) they are stable against small deformations. 
Only these solitons describe physical states in the spectrum of the theory.

Stability of solitons requires the existence of some non-trivial topological charge: the soliton of smallest mass within a given topological sector is stable. For instance, when $Y=G/K$, which is a contractible space diffeomorphic to $\R^n$, the only stable finite-mass solution is the vacuum, since there is no obstruction to continuosly decreasing the energy of the field  configuration down to zero. The only way to get some non-trivial stable soliton, while preserving the crucial fact that the equations \eqref{eqertjj} are equivalent to the  $tt^*$ PDEs of special geometry \eqref{rrqw123}, is to quotient the target space $G/K$ by a discrete subgroup of $G$. We are led to the following set-up:
we have a homomorphism $\varrho\colon \pi_1(X)\to G$ whose image $\Gamma\subset G$ is a discrete subgroup,
and we consider the \emph{twisted} maps
$\widetilde{y}\colon\widetilde{X}\to G/K$, with domain the universal cover $\widetilde{X}$, 
which satisfy the equivariance condition
\be\label{twisted}
\widetilde{y}(\xi\mspace{-2mu}\cdot\mspace{-2mu} x)=\varrho(\xi)\mspace{-2mu}\cdot\mspace{-2mu}\widetilde{y}(x)\qquad \text{for all}\ \xi\in \big(\text{deck group of }
\widetilde{X}\to X\big)\cong \pi_1(X),
\ee
so that $\widetilde{y}$ descends to a well-defined   map $y\colon X\to \Gamma\backslash G/K$.
If $\Gamma$ is not torsion, the target space 
$\Gamma\backslash G/K$ is no longer contractible, so now the solitons may be stabilized by their non-trivial homotopy class.

Intuitively, we have at most one stable soliton of finite mass per homotopy class of maps $y\colon X\to \Gamma\backslash G/K$, i.e.\! the one with the smallest possible mass. If this soliton happens to be  special, eqn.\eqref{eqertjj}, it represents a solution to the $tt^*$ equations, and hence may be lifted to a $tt^*$ brane amplitude $\Psi(\zeta)$, see \S.\,\ref{basic}\textbf{(IV)}.
Uniqueness of the soliton with given topological charge then implies rigidity of the finite-mass stable brane amplitude. All math rigidity theorems we use in this paper (including the VHS structure theorem) are manifestations of this basic physical idea.
If, in addition, the stable special soliton preserves $\boldsymbol{Q}$,
it describes a special K\"ahler geometry enjoying  the extra property of ``stability'' which is neither required by 4d $\cn=2$ supersymmetry nor definable while remaining  within the \textsc{sugra} language.
 
\medskip

Stability of non-trivial $tt^*$ solitons requires the subgroup $\Gamma\subset G$ not to be torsion, hence infinite. Seeing $G\equiv Sp(2m+2,\R)$
as an $\R$-algebraic group, we conclude that 
its Zariski closure\footnote{\ Here the Zariski closure is taken over the ground field $\R$. In the actual story of the swampland structural  criterion one should take the Zariski closure over $\mathbb{Q}$. Again this is a manifestation of the `arithmetic' nature of QG.} $\overline{\Gamma}^{\mspace{2mu}\R}\subseteq G$, identified with the real Lie group $M\equiv \overline{\Gamma}^{\mspace{2mu}\R}(\R)$, must have positive dimension.
The target space $\Gamma\backslash G/K$
may be continuously retracted to $\Gamma\backslash M/K_M$ ($K_M\equiv [K\cap M]$),
so that the stable solitons are naturally expected to take value in this smaller space. The special ones yield solutions of the $tt^*$ PDEs provided $M$ is a real semi-simple Lie group without compact factors. 
\smallskip

We are led to the following situation: 
we have a non-compact semi-simple real Lie group $M$
and a discrete monodromy subgroup $\Gamma\subset M$
whose $\R$-Zariski closure is the full $M$,
and we consider the twisted 
maps as in eqn.\eqref{twisted} with target $M/K_M$.
Without loss we may reduce to the case that $M$
is simple. 
Under these assumptions on  $\Gamma$ one shows (\!\!\cite{corlette} \textbf{Theorem 2.1}) that if there is any smooth map 
\be
y_\text{any}\colon X\to \Gamma\backslash M/K_M
\ee
whose lift $\widetilde{y}_\text{any}\colon \widetilde{X}\to M/K_M$ satisfies \eqref{twisted} and has finite mass \eqref{ffwqma}, then there is a unique finite-mass soliton in the same homotopy class.
One gets this soliton by starting with the given generic finite-mass map $y_\text{any}$ and  deforming it continuously to decrease its mass until we reach its minimal possible value. Its mass saturates the lower bound on the masses allowed in the given topological class, just like it was a BPS state. The Zariski-density condition of $\Gamma$ in $M$ guarantees that in this process
our finite-mass configuration does not escape to infinity in field space nor develops singularities. 

When $X$ is K\"ahler, the soliton will be automatically special (i.e.\! pluri-harmonic), and the algebra $\mathcal{R}$ generated by the $C_i$ commutative, provided we may justify
that when we integrate
the equality \eqref{hasqwrt} over $X$
the boundary term
\be
\int_X\sqrt{\gamma}\,\mspace{2mu} \overline{D}^i \overline{D}^j\mspace{1.5mu}\mathrm{tr}\big(C_iC_j\big)\,d^mx=
\int_{\partial X} D_{\bar l}\big(\gamma^{\bar k i}\gamma^{\bar l j}g_{ab}\, \partial_i y^a\,\partial_j y^b\big) dS^i\overset{?}{=}0
\ee 
vanishes. This is typically true for finite-mass solitons in spaces $X$ of finite volume (e.g.\! it holds when $X$ is the union of a compact set and finitely many cusps).
This condition is automatically satisfied by the special soliton $w$ in the diagram \eqref{gasqwe4534} as a consequence of $\cn=2$ supersymmetry. Indeed, all $\boldsymbol{Q}$-preserving solitons are automatically special since
\be
\mathrm{tr}\big(C_iC_j\big)=-\frac{1}{2}\mspace{2mu}\mathrm{tr}\mspace{-1mu}\Big(\big[\boldsymbol{Q},C_iC_j\big]\Big)=0.
\ee

All finite-mass stable solutions to the $tt^*$ equations arise in this way for some dense discrete group $\Gamma$ and  
they are uniquely determined\footnote{\ From a different point of view: the $tt^*$ PDEs are an isomonodromic problem, and the solutions are determined by the monodromy $\Gamma$ they preserve.} by $\Gamma$. This is the statement of rigidity: it is analogue to the rigidity of BPS configurations in \textsc{susy} theories. This is hardly a surprise: in the usual applications of $tt^*$ geometry, these solutions  \emph{do describe} finite-tension stable BPS branes. 

This leads to the following \emph{informal} version of the conjecture we are proposing: 
\vglue6pt

\noindent\textbf{Swampland structural  criterion} (Informal)\textbf{.} {\it A quantum-consistent $\cn=2$ special geometry is described by a \emph{physical} $tt^*$ $\boldsymbol{Q}$-soliton
\be
w\colon S\to \Gamma\backslash \mathscr{D}^{(1)}_m\equiv \Gamma\backslash Sp(2m+2,\R)/U(m+1),
\ee 
that is, \emph{a stable one of finite-mass.} In particular the monodromy group $\Gamma$ is Zariski dense in a Lie subgroup $M\subset Sp(2m+2,\R)$ which contains the image of a lift of $w$. All other special K\"ahler geometries belong to the swampland.}
\vglue6pt

\noindent The statement sounds very physical in its wording, but it refers to an auxiliary field theory which lives on the moduli space $S=\Gamma\backslash \widetilde{S}$. The physical meaning of the auxiliary field theory is unclear to me.

\subsection{The structure theorem of algebro-geometric VHS}
A basic fact about a VHS which has an algebro-geometric origin is that its complex base $S$ (the ``moduli'' space) is quasi-projective, that is, there is a  projective 
variety $\overline{S}$ and a divisor $Y$ so that
$S=\overline{S}\setminus Y$. The Hironaka desingularization theorem implies that the pair $(\overline{S},Y)$ may be chosen so that
$Y$ is a simple normal-crossing (\textsc{snc}) divisor.
Since $S$ is assumed to be complete for the Hodge metric, $Y=\sum_i Y_i$ is minimal in the following sense: the period map $p$ in eqn.\eqref{griper} cannot be extended along any prime component $Y_i$ of $Y$; this is equivalent to saying that the monodromy $\gamma_i$ around each component $Y_i$ is non-semi-simple \cite{periods}.

The monodromy group $\Gamma$ of an algebro-geometric VHS enjoys special properties \cite{peter}:
$\Gamma$ is finitely generated (in facts, in our particular context, even finitely presented),  and contains a neat\footnote{\ A group $\Gamma\subset GL(n)$ is called \emph{neat} if all elements $\gamma\in \Gamma$ (seen as matrices) have no eigenvalue $\neq1$ of norm 1. In particular, in a neat group there are no non-trivial finite order elements, and all quasi-unipotent elements are unipotent. In particular, after passing to the finite cover as in the main text, the monodromies $\gamma_i$ around the prime divisors $Y_i$ are all non-trivial unipotent elements of $Sp(2m+2,\Z)$.} (hence torsion-free) finite-index subgroup. We work modulo finite groups, so, by replacing the base $S$ of the VHS with a finite cover, we may (and do) assume $\Gamma$ neat. The monodromy $\gamma_i$ around each prime divisor at infinity, $Y_i$, is then non-trivial unipotent.

The most important property of $\Gamma$ is that it is semi-simple. 
To formulate this condition precisely, we introduce the notion of the $\mathbb{Q}$-Zariski closure of $\Gamma\subset G_m(\Z)$ in the
$\mathbb{Q}$-algebraic group $G_m(\mathbb{Q})$,
that we denote as 
$\overline{\Gamma}^{\mspace{2mu}\mathbb{Q}}$.
By definition, $\overline{\Gamma}^{\mspace{2mu}\mathbb{Q}}$ is the smallest $\mathbb{Q}$-algebraic subgroup of $G_m(\mathbb{Q})$ which contains  $\Gamma$. 

\begin{thm}[See e.g.\!  \cite{reva,revb,peter}]  When the VHS has a geometric origin, the $\mathbb{Q}$-algebraic group $\overline{\Gamma}^{\mspace{2mu}\mathbb{Q}}$ is \emph{semi-simple}.
\end{thm}

\noindent By general theory \cite{milnebook}, a semi-simple $\mathbb{Q}$-algebraic group is an almost-product of simple $\mathbb{Q}$-algebraic groups\footnote{\ Here $\approx$ means equality up to finite groups (i.e.\! modulo isogeny). Since we are working modulo finite groups, we shall not distinguish between an almost-simple group and its simple quotient.}
\be\label{kfasqw}
\overline{\Gamma}^{\mspace{2mu}\mathbb{Q}}\approx M_1\times M_2\times \cdots\times M_\ell\equiv P
\ee
where the $M_i$'s are \emph{simple} $\mathbb{Q}$-algebraic groups.  
Hence (replacing $S$ with a finite-cover if necessary) we may assume 
\be
\Gamma=\Gamma_1\times \Gamma_2\times\cdots\times\Gamma_\ell,\quad \text{with}\quad
\Gamma_i \subset M_i\subset P\equiv \overline{\Gamma}^{\mspace{2mu}\mathbb{Q}},
\ee  
which entails that the group $M_i(\R)$ of the real-valued points of $M_i$  is a  non-compact, simple, real Lie group. One can show that, modulo finite groups,
 $P\equiv \overline{\Gamma}^{\mspace{2mu}\mathbb{Q}}$ is a normal subgroup of the derived group $D\boldsymbol{M}$ of the Mumford-Tate ($\mathbb{Q}$-algebraic) group\footnote{\ For the theory of the Mumford-Tate groups, see e.g.: 
 
{\bf nice short  surveys:} $\mspace{5mu}$\cite{reva,revb}.

{\bf detailed treatments:} \cite{MT4,MT5,periods}.
}
$\boldsymbol{M}$ of the VHS, i.e.\! 
\be\label{mufttate}
\boldsymbol{M}=P\times R,\qquad R=M_{\ell+1}\times \cdots M_k\times A
\ee
where $M_i$ (resp.\! $A$) is a semi-simple 
(resp.\! a torus) $\mathbb{Q}$-algebraic group.
By definition of Zariski closure, the two groups $\Gamma$ and $\overline{\Gamma}^{\mspace{2mu}\mathbb{Q}}$ have the same rational tensor invariants  so, roughly speaking, they are ``algebraically indistinguishable''.

 All factors $M_i$ and $A$ of $\boldsymbol{M}$ (in particular, the first $\ell$ factors entering in $\overline{\Gamma}^{\mspace{2mu}\mathbb{Q}}$) have the property that the Lie groups of their real points, $M_i(\R)$ and $A(\R)$, contain a compact maximal torus (so the Abelian group $A(\R)$ is itself compact).
Therefore, not all simple $\mathbb{Q}$-algebraic groups may appear as factors of $\overline{\Gamma}^{\mspace{2mu}\mathbb{Q}}$, but only the ones whose underlying real Lie group $M_i(\R)$ has a Vogan diagram with trivial automorphism, that is, the ones in \textsc{Figure 6.1} and \textsc{Figure 6.2} of \cite{knapp}.
In our application to $\cn=2$ supergravity
the possible groups $M_i$ are further restricted  by the condition that the weight of the VHS is \emph{odd} (in fact 3), see \cite{reva,MT4} for the list of Lie groups satisfying this more stringent requirement.

For a general (that is, not necessarily of CY-type)
pure polarized VHS, the period map $p$ has target space 
the quotient of the Griffiths period domain $D$ by the
monodromy group
\be\label{pmapppa}
p\colon S\to \Gamma\backslash D,\qquad\text{with}\quad D\equiv G(\R)/H,
\ee 
where $G(\R)$ is the group of the real points of the $\mathbb{Q}$-algebraic group $G$ preserving the polarization; in the 3-CY case with Hodge numbers $h^{3,0}=1$, $h^{2,1}=m$ this is\footnote{\ In the equation we identify a $\mathbb{Q}$-algebraic group with the group of its $\mathbb{Q}$-valued points, for  clarity.}
\be 
G=G_m\equiv Sp(2m+2,\mathbb{Q}).
\ee
In eqn.\eqref{pmapppa} $H$ is the isotropy group of a reference 
Hodge decomposition with the given Hodge numbers $h^{p,q}$;  in our case, $h^{3,0}=1$, $h^{2,1}=m$, $H$ is \be
H\equiv H^{(2)}_m=U(1)\times U(m),
\ee 
and $G(\R)/H\equiv \mathscr{D}^{(2)}_m$.
We have the inclusion of $\mathbb{Q}$-algebraic groups
\be\label{ffrawql}
P\times R\equiv\boldsymbol{M}\subseteq G.
\ee
For ``generic'' VHS one has $\boldsymbol{M}= G$.
For instance, equality holds -- modulo finite groups -- for all universal families of complete intersections \cite{beauv}. To each factor in \eqref{ffrawql} we may associate a sub-domain of the period domain $D$
 \begin{align}
 &D_P\overset{\rm def}{=}P(\R)/[P(\R)\cap H]\equiv \prod_{i=1}^\ell M_i(\R)/[M_i(\R)\cap H]\\
 & D_R\overset{\rm def}{=}R(\R)/[R(\R)\cap H]\equiv\prod_{i=\ell+1}^k M_i(\R)/[M_i(\R)\cap H].
 \end{align}  
 We have the obvious embedding
 \be
 \\
D_P\times D_R \hookrightarrow D.
 \ee

The main result in the theory is the

\begin{thm}[Structure theorem \cite{reva,revb,periods,MT4}] In a VHS of algebro-geometric origin the \emph{global} period map $p$ (cfr.\! \eqref{pmapppa}) factorizes as in the commutative diagram
\be\label{strthemre}
\xymatrix{S\ar@/^1.7pc/[rr]^p\ar[r]_{q\phantom{mmmmm}} & \big(\Gamma\backslash D_P\big)\times D_R\; \ar@{^{(}->}[r] &\;\Gamma\backslash D}
\ee
with $q$ constant in the $D_R$ factor.
\end{thm}

We focus on the non-trivial factor and consider the 
map
\be\label{pMT}
\varpi\colon S\to \Gamma\big\backslash D_P\equiv \prod_{i=1}^\ell \Big(\Gamma_i\big\backslash M_i(\R)\big/H_i\Big),
\ee
where 
\be
H_i\overset{\rm def}{=} H\cap M_i(\R).
\ee
Note that the map $\varpi$ is necessarily non-constant in each one of the $\ell$ factors. Of course $\varpi$ should still obey the Griffiths infinitesimal period relations (i.e.\! should describe a $tt^*$ geometry, cfr.\!\! \S.2).
Let $\co(\mathfrak{m}_i^{-1,1})$ be the homogeneous vector bundle \cite{homore}\!\!\cite{periods,book} over the reductive homogeneous space $M_i(\R)/H_i$ associated to the $H_i$-module\footnote{\ The Lie subgroup $H_i\subset M_i(\R)$ acts on the complexified Lie algebra $\mathfrak{m}_i^\C$  of $M_i(\R)$ through the adjoint representation. As explained in the text, this action preserves the Hodge decomposition of $\mathfrak{m}_i^\C$. }
\be
\mathfrak{m}_i^{-1,1}\overset{\rm def}{=} \mathfrak{m}_{i}^{\C}\cap \mathfrak{g}^{-1,1},
\ee
where $\mathfrak{m}_{i}^{\C}= \mathfrak{Lie}(M_i)\mspace{-2mu}\otimes_\mathbb{Q}\mspace{-2mu}\C$, and $\mathfrak{g}^{-1,1}$ is the component of type $(-1,1)$ in the associated adjoint Hodge structure on the Lie algebra $\mathfrak{g}$ of $G$ at a reference point $\star\in S$ 
\be
\mathfrak{g}\otimes_\mathbb{Q}\mspace{-1.5mu}\C =\bigoplus_k \mathfrak{g}^{-k,k}
\ee
(see \cite{reva,revb,MT4, periods} for details).
Then the infinitesimal period relations for $\varpi$ take the form \cite{reva,revb}
\be\label{eeqrt}
\varpi_\ast(TS)\subseteq \bigoplus_{i=1}^\ell \co(\mathfrak{m}_i^{-1,1})
\ee 
where $TS$ is the holomorphic tangent bundle to $S$ (cfr.\! eqn.\eqref{uuuuqw12}).

\subsection{The case of ``motivic'' special K\"ahler geometries}\label{motSK}
In the case that our VHS has the Hodge numbers of a special K\"ahler geometry (i.e.\! $h^{3,0}=1$, $h^{2,1}=m$), we can say a bit more on the number $\ell$ of factors in $P$, eqn.\eqref{kfasqw}:

\begin{lem}\label{lllemmaq} {\bf(1)} A special K\"ahler geometry $S$ of algebro-geometric origin has $\ell=0,1,2,3$. 
\begin{itemize}
\item $\ell=0$ if and only if $\dim_\C S=0$ (in this case the corresponding algebro-geometric object is \emph{rigid}, e.g.\! a \emph{rigid 3-CY).} 
\item $\ell=2,3$ if and only if $S$ is locally isometric to the Hermitian symmetric space
\be\label{pzq1x}
SL(2,\R)/U(1)\times SO(2,k)/[SO(2)\times SO(k)],\qquad k=1,2,\cdots,
\ee  
with $\ell=3$ iff $k=2$.
\item in all other cases $\ell=1$.
\end{itemize}
{\bf(2)} $\varpi$ is onto iff  $\Gamma\backslash D_P$ is locally symmetric,
in which case $\varpi$ is an isomorphism of complex manifolds, $S\cong\Gamma\backslash D_P$.
\end{lem}

\begin{proof} {\bf (1)} By ``global Torelli'' \cite{glo-torelli,glo-torelli1,glo-torelli2} we may identify $S$ and its image $\varpi(S)$. Hence if $\ell>1$ the special K\"ahler geometry is locally isometric to a product; we then apply the Ferrara-van Proeyen theorem \cite{ferrara} which says that \emph{all} simply-connected reducible special K\"ahler geometries have the form \eqref{pzq1x}. {\bf (2)} If $\varpi$ is onto, the Griffiths infinitesimal period relations \eqref{eeqrt} become
tautological, and this is equivalent to $\varpi(S)$ being locally symmetric \cite{reva}. 
\end{proof}

We write
\be
M_i(\Z)\overset{\rm def}{=} M_i\ \cap G_m(\Z) \equiv M_i\cap Sp(2m+2,\Z).
\ee
By construction $M_i(\Z)$ is an arithmetic subgroup \cite{morris,borel} of the rational algebraic group  $M_i$ and one has 
\be
\Gamma_i\subset M_i(\Z)\equiv \overline{\Gamma}_i^{\mspace{2mu}\mathbb{Q}}\cap G_m(\Z)
\ee 
by Dirac quantization of charge.
We have two possibilities:
\begin{itemize}
\item the monodromy (sub)group $\Gamma_i$ is of finite index in $M_i(\Z)$ and hence it an arithmetic group on its own right;
\item $\Gamma_i$ is of infinite index in the arithmetic subgroup $M_i(\Z)$ of its $\mathbb{Q}$-Zariski closure $M_i$. In this case the monodromy group is said to be \emph{thin} \cite{thin2}.  
\end{itemize}

Simple situations have arithmetic monodromies:
in particular, when $S$ is locally symmetric the monodromy is always arithmetic \cite{book}.  All families of complete intersection CY have arithmetic monodromy \cite{beauv}. However, there are geometric situations with thin monodromy: the simplest such example is the mirror of the quintic \cite{thinmirror}.

\subsection{The structural swampland  criterion}

In the (ungauged) $\cn=2$ case we propose to replace in the $\cn=2$ case the several swampland conjectures by the following single one:

\begin{cri}[Structural criterion]\label{tttqwaswam} A special K\"ahler geometry $S$ belongs to the swampland unless its period map $p$ satisfies the VHS structure theorem \eqref{strthemre}. 
\end{cri}

In particular the discrete gauge group $\Gamma$ should satisfy eqn.\eqref{kfasqw}.

\subsection{A direct physical proof?}\label{direct}

As sketched in \S.\,1.3 the structural swampland criterion may be actually proven (as a necessary condition) for low-energy effective theories with \emph{rigid} $\cn=2$ supersymmetry. Let us focus on the crucial ingredients of the proof in QFT: the first one is the existence of a well-defined quantum (complex) Hilbert space $\boldsymbol{H}$, carrying a linear representation of the $\cn=2$ \textsc{susy} algebra
\be
\{\bar Q_{A\,\dot\alpha}, Q^B_\beta\}= {\delta_A}^B\,\gamma_{\alpha\dot\beta}^\mu P_\mu,
\quad
\{Q^A_\alpha, Q^B_\beta\}= \epsilon^{AB}\,\epsilon_{\alpha\beta}\, Z, 
\quad 
\{\bar Q_{A\,\dot\alpha}, \bar Q_{B\,\dot\beta}\}= \epsilon_{AB}\,\epsilon_{\dot\alpha\dot\beta}\, \bar Z,
\ee
 and of an algebra $\ca$ of quantum operators acting on $\boldsymbol{H}$. This allows to define the chiral ring $\mathscr{R}_{4d}$ as the subfactor consisting of scalar\footnote{\ The restriction to scalar operators is for convenience.} operators commuting with
  $\bar Q_{A\,\dot\alpha}$ modulo the ones which may be written as  $\bar Q$-anticommutators. 
  From the \textsc{susy} algebra and locality we learn that $\mathscr{R}_{4d}$ is a commutative $\C$-algebra with unit. 
 The second crucial ingredient is that in a UV complete $\cn=2$ QFT the chiral ring $\mathscr{R}_{4d}$ is also finitely-generated. Indeed in this case we have a well-defined UV fixed-point, with \emph{finite} conformal central charges $a_\text{uv}$ and $c_\text{uv}$, whose chiral ring $\mathscr{R}_{4d}^\text{uv}$ is graded by the conformal dimension
  $d$. When $\mathscr{R}_{4d}^\text{uv}$
  is a \emph{free} polynomial ring we have the unitarity bound \cite{tack}
  \be
\#\text{(generators of $\mathscr{R}_{4d}^\text{uv}$)}\leq 
 4(2a_\text{uv}-c_\text{uv})\equiv\text{finite,} 
  \ee
  and one reduces to this case in all known $\cn=2$ SCFT
  (see discussion in \S.\,5 of \cite{disc}). By RG flow, $\mathscr{R}_{4d}$ is a deformation of $\mathscr{R}_{4d}^\text{uv}$ and inherits finite-generation from it.
   
 Then, by the Hilbert basis theorem, the spectrum of a quantum-consistent $\mathscr{R}_{4d}$ is an affine variety defined over $\C$, hence has the form $\overline{M}\setminus Y_\infty$ for some projective variety $\overline{M}$ and effective divisor $Y_\infty$.
  It is convenient to work with the smooth locus in the Coulomb branch,  $M\equiv \overline{M}\setminus (Y_\infty+Y_D)$, where $Y_D$ is the effective divisor with support on the locus where some extra degree of freedom becomes massless. 
  Then the low-energy description of a UV complete 4d $\cn=2$ theory is described by a weight-1 VHS over a smooth quasi-projective basis $M\equiv \overline{M}\setminus Y$, and we may assume without loss
  that $Y$ is a simple normal crossing divisor. By going to a finite cover, if necessary, we may also assume the monodromy group $\Gamma$ to be neat. Then the structure theorem holds (for a sketch of the argument, see e.g.\! \S.\,IV of \cite{reva}).
  \medskip
  
  From the above we see that we have a proof of
  the structural swampland criterion also in the $\cn=2$ quantum gravity setting subjected to two assumptions:
  \begin{itemize}
  \item[1)] the asymptotically-Minkowskian quantum states form a Hilbert space carrying a unitary representation of $\cn=2$ supersymmetry.
  In this case the chiral ring $\mathscr{R}_\text{sugra}$ is well defined;
  \item[2)] $\mathscr{R}_\text{sugra}$ is finitely generated.
  \end{itemize} 
From the  semi-classic viewpoint these two assumptions look reasonable.

\subsection{Recovering the Ooguri-Vafa swampland statements}\label{recovering}

We now show that the Ooguri-Vafa (OV) swampland conjectures \cite{OoV} are implied by the above \textbf{Criterion} i.e.\! by the VHS structure theorem (specialized to the appropriate Hodge numbers). Most arguments may already be found  in \cite{OoV}.

\paragraph{(1) $S$ is non-compact.} The first OV conjecture states that either $S$ is zero-dimensional or it is non-compact. Since a VHS with a compact base has a constant period map, if $S$ is compact $\varpi(S)$ reduces to a point, and since $S\cong\varpi(S)$ by ``global Torelli'', $S$ is also a point.
The non-compactness of $S$ is related to the fact that $\Gamma$ should contain non-trivial unipotent elements (i.e.\! the monodromies around the prime divisors in $Y$): compare to the Godement non-compactness criterion for finite-volume quotients of symmetric spaces\cite{morris}.

\paragraph{(2) The scalars' manifold is geodesically complete.} This is the most tricky statement. In \S.\ref{basic} we found convenient to define the covering special geometry $\widetilde{S}$ to be the completion of a certain ``moduli'' space with respect to the Hodge metric $K_{i\bar j}$. Then $S$ is geodesically complete for $K_{i\bar j}$ by construction.
However the OV conjecture refers to geodesic completeness with respect to the WP metric $G_{i\bar j}$ which is the one entering the scalars' kinetic terms.
Each one of these two metrics is associated to its own viewpoint about special geometry:
with reference to diagram \eqref{gasqwe4534}, the WP metric is tied to the Legendre map $z$, and  the Hodge metric to the period map $p$.
Their K\"ahler forms are restrictions to the respective images $L$ and $S$ of the curvature of the canonical homogeneous line bundles $\cl^{(a)}\to\mathscr{D}^{(a)}_m$ over the respective target domains.

In general $S$ is \textbf{not} complete for the Weil-Petersson metric.
Indeed it is known that the moduli spaces of Calabi-Yau 3-folds is typically non-complete for the Weil-Petersson metric, see \cite{noncom,klemm}. In particular, a conifold point is at finite WP distance \cite{klemm} while being at infinite Hodge distance. The period map 
\be
p\colon S\to \Gamma\backslash \mathscr{D}^{(2)}_m
\ee 
cannot be extended to such a point \cite{GII,periods}, whereas the Legendre map
\be
z\colon S\to \Gamma\backslash \mathscr{D}^{(3)}_m
\ee 
may be extended continuously to the conifold  point as we shall see momentarily. In other words, we may complete $L$ by adding to it the points at finite distance in the WP metric even if we cannot extend there $S$ (which is locally identified with its image under $p$). This peculiar situation is related to the fact that the fiber of the projection
\be
\mathscr{D}^{(2)}_m\xrightarrow{\ \varpi_3\ }\mathscr{D}^{(3)}_m
\ee 
is non-compact, so, a geodesic $\widetilde{\gamma}(t)\subset \widetilde{S}\subset \mathscr{D}^{(2)}_m$ which approaches a conifold point may stretch to infinite length  in the fiber direction while its projection
$\varpi_3(\widetilde{\gamma}(t))\subset\widetilde{L}\subset \mathscr{D}^{(3)}_m$ remains at finite distance in the base. We may be a little more precise: the Penrose map
\be
\mathscr{D}^{(2)}_m\longrightarrow \mathscr{D}^{(1)}_m\times \mathscr{D}^{(3)}_m
\ee 
is an embedding, while the fiber of $\varpi_3$ is
a copy of Siegel's upper half-space $Sp(2m,\R)/U(m)$ whose physical interpretation is easily understood from the Penrose correspondence.
A point in the fiber, seen as a symmetric  $m\times m$ symmetric matrix $\boldsymbol{\tau}_{ij}$ with positive-definite imaginary part, is nothing else that the matrix of complexified couplings of the \emph{matter} vector fields, whereas the point in $\mathscr{D}^{(1)}_m$ describes \emph{all} vector couplings $\tau_{IJ}$, including the one of the graviphoton. The point in $\mathscr{D}^{(3)}_m$ specifies which $\C$-linear combination of the electric and magnetic charges is the \textsc{susy} central charge $Z$ \cite{cec}.
Thus a point at finite distance in the WP metric but infinitely away in the Hodge one is a field configuration where the \textsc{susy} central charge $Z$ and the graviphoton couplings remain finite, but  
 \emph{some} coupling of the matter gauge vectors blows up (or vanishes, depending on the chosen duality frame). 
 
 Since the problem arises purely in the matter sector,\footnote{\ We stress that the separation \textsf{matter sector} vs.\! \textsf{gravitational sector} makes sense only at the linearized level, i.e.\! in the first order expansion around a fixed scalar field configuration.} the situation is akin the ones appearing in rigid $\cn=2$ effective Lagrangian, 
 i.e.\! in Seiberg-Witten theory \cite{SW1,SW2}. For comparison sake, let us recall what happens in the rigid case. The singularities of the scalars' metric arises at points in the Coulomb branch where some 
charged hypermultiplet becomes massless, making incomplete the low-energy effective description in terms of IR-free photon super-multiplets only. Let $a$ be the period associated to the charge of the light hypermultiplet so that its mass is proportional to $|a|$; near the singular point the dual period $a^D$ has the form
\be
a^D= \frac{q}{2\pi i}\, a \log(a/\Lambda)+\text{regular}
\ee
 for some integer $q\neq0$. Clearly the periods $(a, a^D)$ may be extended continuously in the limit $a\to0$, while the K\"ahler metric $ds^2\propto \log\!|a|\,da\,d\bar a$ develops a singularity at \emph{finite} distance.
 The Ricci tensor has the form of the standard Poincar\'e metric in the upper-plane coordinate $z$
(where $a=\exp(2\pi i z)$)
 \be\label{tqerw}
 R_{z\bar z}= \frac{1}{4\,\mathrm{Im}\,z}+O(a),\quad \text{as }a\to0\ \ \text{i.e.}\ \ \mathrm{Im}\,z\to\infty,
 \ee
 so that the point $a=0$ where the hypermultiplet becomes massless would be at \emph{infinite} distance in any K\"ahler metric of the form \begin{equation}
 C\mspace{2mu}G_{i\bar j}+R_{i\bar j}
 \end{equation} 
 ($C$ a suitable constant) for instance in the Hodge metric \eqref{cjqer}.
 We conclude that such finite-distance singular points correspond to loci in scalars' space where \emph{finitely many} states become light, spoiling our low-energy effective description
 in terms of light vector multiplets only.

The original OV conjecture thus holds with the specification that ``scalars' manifold'' should be understood to mean the WP completion of the Legendre manifold $L$. The non-trivial part of the statement is that a continuous extension of $z$ exists. The extended metric is not smooth since its Ricci curvature should blow up at the conifold points, see \eqref{tqerw}.

Let us discuss in more detail the behavior of the special geometry at infinity.
Since $S=\overline{S}\setminus Y$ with $\overline{S}$ compact and Y simple normal crossing, the asymptotic behaviour of the special geometry as we approach a point in $Y$ is described -- up to corrections which are exponentially small in the geodesic distance -- by the multi-variable version of the nilpotent orbit theorem \cite{schm,morrison}. More concretely, let $U\subset \overline{S}$ be a small open set (in the analytic topology) where $Y$ takes the form $z_1z_2\cdots z_s=0$,
so that 
\be
U\cap S\cong (\Delta^{\mspace{-2mu}*})^s \times \Delta^{m-s}\; \xleftarrow{\ \;\sigma\ } \;\mathcal{H}^s \times \Delta^{m-s}\cong \widetilde{U\cap S}
\ee 
where $\Delta$ is the open unit disk, $\Delta^{\mspace{-2mu}*}=\{z\in \C,\; 0<|z|<1\}$ the punctured unit disk, $\mathcal{H}$ the upper half-plane, and $\sigma$ the
 universal covering map
 \be
 \sigma\colon(\tau_1,\cdots,\tau_s, z_{s+1},\cdots, z_{m})\mapsto (q_1,\cdots,q_s,z_{s+1},\cdots, z_m),\quad\text{with } q_i= q(\tau_i)=e^{2\pi i\tau_i}.
 \ee 
The local lift of the period map, 
\be
\widetilde{p}_U\colon \widetilde{U\cap S}\to D,\qquad \widetilde{p}_U\equiv \widetilde{p}\big|_U
\ee
takes the form \cite{schm}
\be\label{gaqwe}
\widetilde{p}_U(\tau_i;z_a)=\exp\mspace{-4mu}\left(\sum_{i=1}^s \tau_i\, N_i\right)\mspace{-4mu}\cdot\mspace{-2mu} F(q_i,z_a)+O\mspace{-0.5mu}\big(e^{-2\pi \mathrm{Im}\,\tau_i}\big),
\ee 
where the $N_i$'s are a $s$-tuple of non-zero, commuting, nilpotent elements of $\mathfrak{sp}(2m+2,\mathbb{Q})$ such that the monodromy around the divisor $\{z_i=0\}$ is
$\gamma_i=\exp(N_i)$ and $F\colon \Delta^{m}\to \check{D}$ is a regular holomorphic map.\footnote{\ Here $\check{D}$ is the \emph{compact dual} of the Griffiths period domain $D$ \cite{GII,Gbook,reva,revb,periods}. 
$\check{D}$ is a compact projective complex manifold and a $G(\C)=Sp(2m+2,\C)$-homogeneous space such that $D\subset \check{D}$ as an open domain.}  
From this expression one gets the asymptotic behaviour of the metric at infinity in terms of the action of the $N_i$'s.
For instance, let us approach the $j$-th prime divisor $Y_j$, i.e.\! we take $\mathrm{Im}\,z_j\to\infty$ at fixed values of the other
coordinates $q_{i\neq j}, z_a$.
Set 
\be
\kappa_j=\max\{\ell\, :\, N_j^\ell \varpi_3 F|_{q_j=0}\neq0\}.
\ee 
Then \cite{noncom}
\be\label{pqaw12}
G_{\bar z_j z_j}= \frac{\kappa_j}{4(\mathrm{Im}\,z_j)^2}+\text{exponentially small}
\ee
so completeness follows from comparison with the Poincar\'e metric, unless
$\kappa_j=0$, in which case we see from eqn.\eqref{gaqwe} that  
 the map 
 \be
 \widetilde{z}\equiv \varpi_3\mspace{2mu}\widetilde{p}\colon S\to \mathscr{D}^{(3)}_m
 \ee 
 extends to $q_j=0$ \cite{noncom} sot that we may extend continuously the covering Legendre manifold $\widetilde{L}$ there.

\paragraph{(3) Asymptotically at infinity the curvature of $S$ is non-positive.} This is \textbf{Conjecture 3} of \cite{OoV}. It follows from the asymptotic formula \eqref{pqaw12}.
Note that near a conifold point the curvature is indeed positive (cfr.\! eqn.\eqref{tqerw}) but this very fact implies that the singular point is at \emph{finite} distance in the WP metric. 

\paragraph{(4) $S$ has finite volume.} 
The above asymptotic expressions, together with  compactness of $\overline{S}$, show that the volume of $S$ is finite \cite{GII}. Ref.\!\cite{GII} shows that the volume of $\varpi(S)$, computed with the Hodge metric $K_{i\bar j}$, is finite when the VHS has a compactifiable base $S=\overline{S}\setminus Y$ \cite{GII}.
The OV conjecture refers to the volume computed with the WP metric $G_{i\bar j}$. The argument of \cite{GII} applies to this metric as well, since the two metrics are simply related in the asymptotic regime as a consequence of eqn.\eqref{pqaw12}.
For more details and precise results on the WP volumes of CY moduli space -- which apply in general to all special geometries consistent with the structure theorem --  see \cite{cymoduli1,cymoduli2}.

\paragraph{(5) Properties of $\pi_1$.}
\textbf{Conjecture 4} of \cite{OoV} states:
\vskip 8pt

\noindent{\it 
In the scalars' manifold there is no non-trivial 1-cycle with minimum length within a given homotopy class.}
\vskip8pt

The precise wording of this conjecture requires some refinement.\footnote{\ I thank Cumrun Vafa for a discussion on the issue.}
There are $h^{2,1}=1$ Calabi-Yau 3-folds whose complex moduli space has the form
\be
S=\mathbb{P}^1\setminus\{n\geq 4\ \text{points}\}
\ee
(the monodromies around the $n$ punctures being non-semisimple). There are plenty of such examples. For instance, the double octic corresponding to the eight plane arrangement number 254 in Meyer's list \cite{meyer}, i.e.\! a certain crepant resolution of the one-parameter family of
\emph{singular} degree 8 hypersurface in the weighted projective space $\mathbb{P}(1,1,1,1,4)$ \cite{draco}
\be
x_1x_2x_3x_4(x_1+x_2+x_3+x_4)\mspace{-2mu}(x_4+x_2+s\mspace{1mu}x_3)\mspace{-2mu}
(s\mspace{1mu} x_3+s\mspace{1mu} x_4+x_1+x_2)\mspace{-2mu}(x_1+s\mspace{1mu}x_2+s\mspace{1mu}x_3)=x_5^2
\ee
where
\be
s\in S\equiv \mathbb{P}^1\setminus \left\{0,1,\infty,\frac{1}{2}, \frac{3-\sqrt{5}}{2}, \frac{3+\sqrt{5}}{2}
\right\}.
\ee
$\{0,1,\infty\}$ are MUM points, whereas the other three punctures are conifolds points, see \cite{draco}. 
Alternatively we may take as ``scalars' manifold''  the metrically completed Legendre manifold 
\be
L^\text{comp.}\cong \mathbb{P}^1\setminus\{0,1,\infty\}
\ee 
in which the three conifold punctures are filled in. Let $\sigma_x\in \pi_1(L^\text{comp})$ be a loop encircling the puncture $x\in\{0,1,\infty\}$.
The length of a path in the class $\sigma_0\sigma_1^{-1}$
is  below by a positive constant, so the wording of the conjecture needs some minor modification.
 
Here we adopt a conservative attitude: we take as ``scalars' space'' the nicer $S$ and slightly modify the statement of the conjecture by replacing  ``homotopy class'' with  ``homology class''.
Since, as \underline{abstract} groups, $\pi_1(S)\cong\Gamma$ (because $\Gamma$ is assumed torsion-free)
the conservative version of the conjecture is equivalent to saying that, as a \underline{concrete} matrix group, $\Gamma$ is generated by unipotent elements. This is automatically true \cite{vanP} whenever $\Gamma$ is an arithmetic group of rank at least 2 (as it happens in most ``elementary'' examples such as all complete intersections), or when $\overline{S}$ may be chosen 
simply-connected. I am not aware of an example where neither conditions apply.

Let us be general. Let $\Gamma_u\triangleleft \Gamma$ be the subgroup generated by all unipotent elements of $\Gamma$: when $\dim_\C S>0$, $\Gamma_u$ is an infinite normal subgroup. Taking the $\mathbb{Q}$-Zariski closure, we have $\overline{\Gamma}_u^{\mspace{2mu}\mathbb{Q}}\subseteq \overline{\Gamma}^{\mspace{2mu}\mathbb{Q}}$; we may assume $\overline{\Gamma}^{\mspace{2mu}\mathbb{Q}}$ to be simple (otherwise $\Gamma$ is automatically arithmetic by \textbf{Lemma \ref{lllemmaq}}). Then $\overline{\Gamma}_u^{\mspace{2mu}\mathbb{Q}}\equiv \overline{\Gamma}^{\mspace{2mu}\mathbb{Q}}\equiv P$. In other terms, 

\begin{corl} $S$ is a special geometry with period map  $p$ as in \eqref{strthemre}.
Then $\pi_1(S)$ -- seen as a concrete matrix group $\Gamma$ acting on $\mathbb{Q}^{2m+2}$ --
has the same tensor invariants as a discrete group $\Gamma_u$ generated by nilpotent elements. In other words: the Ooguri-Vafa statement above holds (after refinement) for $S$ \emph{at least} in the algebraic sense.
\end{corl}
 
\noindent We feel that this algebraic version is the most natural formulation of \textbf{Conjecture 4} of \cite{OoV}.

\paragraph{(6) The distance conjecture.}
When the special geometry satisfies our structural swampland criterion, its asymptotic behaviour at infinity is described by the nilpotent-orbit and $sl_2$-orbit theorems \cite{schm,morrison}. It is known that these results imply the distance conjecture. For a rough sketch of the argument see \cite{book}, for nice and detailed analyses see \cite{Grimm:2018ohb,klemm,grimm2,grimm3,grimm4,grimm5}.

We stress that, in the present set-up, the distance conjecture is almost tautological. 
Our viewpoint is that -- whenever our special geometry \textbf{does not} belong to the swampland -- we have a natural compactification $\overline{S}$ of the scalars' manifold with a sub-locus $Y\subset \overline{S}$ where the special geometry gets singular. All points in the prime divisors $Y_i\subset Y$ are at infinite distance in the nicer Hodge metric $K_{i\bar j}$ but not necessarily in the WP metric $G_{i\bar j}$ (cfr.\! \S.\textbf{(2)}). 
As discussed around eqn.\eqref{tqerw}, the points on the divisor $Y_j$ are at \emph{finite} WP distance if the singularity arises from \emph{finitely many}
particles becoming massless along the locus $Y_j$. In order to have a more severe singularity of the metric, strong enough to make the WP distance infinite, the theory 
should have at $Y_j$ infrared divergences worse than those produced by any finite number of particles becoming massless. Clearly, the only possibility is that an \emph{infinite} tower of states get light.

\paragraph{(7) No global symmetries.} All global symmetries should act non-trivially on the bosonic fields
since the $R$-symmetry is a gauge symmetry in \textsc{sugra}. From the discussion in \S.\ref{preliminary} we see that the global symmetry group is
\be\label{glosymgg}
G_\text{glob}= \cn_{\mathsf{Sym}_\Z}(\Gamma)/\Gamma,
\ee
where $\cn_{\mathsf{Sym}_\Z}(\Gamma)$ is the normalizer of $\Gamma$ in $\mathsf{Sym}(\widetilde{S})_\Z$.
Modulo commensurability, the group in the \textsc{rhs} is automatically trivial.\footnote{\ Indeed one has
$$\mathsf{Vol}(\Gamma\backslash \widetilde{S})=\big[\cn_{\mathsf{Sym}_\Z}(\Gamma)\colon \Gamma\big]\;\mathsf{Vol}\big(\cn_{\mathsf{Sym}_\Z}(\Gamma)\backslash \widetilde{S}\big)
$$
and, since $\mathsf{Vol}(\Gamma\backslash \widetilde{S})<\infty$, one has $[\cn_{\mathsf{Sym}_\Z}(\Gamma)\colon \Gamma]<\infty$.} Note that from the structure theorem 
\be
{\overline{\mathsf{Sym}(\widetilde{S})_\Z}}^\mathbb{\,Q}\equiv P,
\ee
so there is no algebraic invariant which may detect a difference between the group of \emph{all} symmetries of $\widetilde{S}$,
$\mathsf{Sym}(\widetilde{S})_\Z$, and the actual \emph{gauge} symmetry group $\Gamma$.

\paragraph{(8) Completeness of charge spectrum.}
The group $\Gamma$ acts irreducibly on the $\mathbb{Q}$-space $V=\Lambda\otimes_\Z\mspace{-1mu}\mathbb{Q}$,
where $\Lambda$ is the symplectic lattice of electric-magnetic charges. Thus, if there is a state of charge $v\neq0$, we have states of charges $\{\Gamma v\}$ and these span $V$. 
Then the physically realized electro-magnetic charges make a sublattice of finite-index in $\Lambda$. Since we are working modulo finite groups, this statement is equivalent to saying that all possible charges are physically realized provided \emph{one} charged state exists. That charged states exist follows, say, from the validity of the distance conjecture. Alternatively the compleness of electro-magnetic spectrum is a formal consequence of the absence of global symmetries
\cite{McNamara:2019rup}\!\!\cite{Rev2}\footnote{\ Since in this note we work modulo commensurability (and so use rational VHS rather than integral VHS) we may only conclude that the occupied charge lattice $\Lambda_\text{occ}$ is a finite-index sublattice of the electromagnetic lattice $\Lambda$ and not $\Lambda_\text{occ}\equiv \Lambda$. Equivalently we only show $G_\text{glob}=1$ modulo finite groups.} and hence follows from our previous discussion.

\paragraph{(9) No free parameter.} This conjecture states that the couplings in the effective Lagrangian $\mathscr{L}$ are either the v.e.v.\! of light fields (so they are non-trivial functions on the moduli $S$), or they are frozen to some very specific ``magic'' numerical value.
 In the UV completed QG these ``magic'' numerical values become the v.e.v.'s of heavy fields \cite{Cecotti:2018ufg}, and so are given by the isolated critical points of some high energy potential, which is also subjected to quite strong consistency requirements. 
Then any small perturbation away from the ``magical''  numerical values makes the effective theory inconsistent at the quantum level. 
\textbf{Fact \ref{tubearg}} gives a first illustration of this state of affairs. Saying that the pre-potential $\cf$ is purely cubic, is equivalent to saying that the $\cn=2$ Pauli couplings (or the $\cn=1$ heterotic Yukawa couplings) are field-independent numerical parameters. Then \textbf{Fact \ref{tubearg}} asserts that in a consistent theory such numerical Pauli/Yukawa couplings should be the cubic determinant form of a (possibly reducible) rank-3 real Jordan algebra whose classification is provided by the Freudenthal-Rozenfeld-Tits magic square \cite{magicsquare} (so the adjective ``magical'' is technically accurate for these couplings). We shall elaborate a bit more on these aspects in the next subsection.

\paragraph{(10) the weak gravity conjecture.}
This conjecture states that there must exists states for which the electromagnetic repulsion is larger than the gravitational attraction.  This means that the squared-mass should be less than the 
square of the electromagnetic charge both measured in intrinsic normalizations.
The invariant square of the electromagnetic charge
$\boldsymbol{q}$ is its Hodge squared-norm $Q(C\boldsymbol{q},\boldsymbol{q})$. 
The square of the susy central charge (which is the mass for a BPS state) is the Hodge norm of the
$(3,0)$ projection of $\boldsymbol{q}$, which is smaller by the Schwarz inequality.

\paragraph{(11) dS conjecture, AdS conjecture, \emph{etc.}} These conjectures \cite{Obied:2018sgi,Agrawal:2018own,Lust:2019zwm}\!\!\cite{Rev1,Rev2} refer to properties of the scalar potential $V(\phi)$. In the present context -- low-energy effective theories which are \emph{ungauged} $\cn=2$ supergravities -- 
the scalar potential is identically zero and all
these conjectures are trivially satisfied.

\subsection{More on ``no free parameter''}

Let us revisit the example in \textbf{(9)} of numeric (i.e.\! field-independent) Pauli/Yukawa couplings exploiting the fact that the polarized VHS's over a connected complex manifold $S$  form a semi-simple Tannakian
category over $\mathbb{Q}$ (see \textsc{Proposition} 2.16 in \cite{milne3}).
Restricting to the fiber over a point $s\in S$, this entails that a polarized Hodge structure
over the fixed $\mathbb{Q}$-space $V$,
specified by a Hodge decomposition
(say, of pure weight $n$)
\be\label{hhiosegwa}
V_\C\equiv V\otimes \C=\bigoplus_{p+q=n}H^{p,q}_s,
\ee
induces functorially Hodge structures on all its tensor spaces\footnote{\ Since $V$ is polarized $T^{k,l}\cong T^{k+l,0}$.} $T^{k,l}\equiv V^{\otimes k}\otimes (V^\vee)^{\otimes l}$ given by the Hodge decomposition
\cite{reva,revb,MT4} 
\be
T^{k,l}\otimes\C= \bigoplus_{p+q=(k-l)n} (T_s^{k,l})^{p,q}.
\ee

Let us specialize to a VHS underlying an $\cn=2$ \textsc{sugra}.
 Fix a point $s\in S$ in the scalars' manifold; its image
$p(s)\in\mathscr{D}^{(2)}_m$ specifies a particular Hodge decomposition as in eqn.\eqref{hhiosegwa} of the fixed $\mathbb{Q}$-space $V$. The Pauli/Yukawa coupling
at $s$ is\footnote{\ The Pauli/Yukawa is the cubic form which is the only invariant of the corresponding \emph{infinitesimal variation of Hodge structure}, see \cite{Gbook,bgrif}.}
\be
\partial_{X^i}\partial_{X^j}\partial_{X^k}\cf\Big|_s \in \left(\odot^3\mathrm{End}(H_s^{3,0},H_s^{2,1})\right)\otimes \big(H_s^{0,3}\big)^{\otimes 2}\subset 
(T_s^{5,3})^{3,3}.
\ee
More generally, we may consider other couplings defined as suitable invariant combinations of higher derivatives of $\cf$; such couplings will correspond to elements of some higher tensor space $T^{k,l}$
which have \emph{pure type} 
\be\label{reqwtu}
(p,q)=\left(\frac{3}{2}(k-l),\frac{3}{2}(k-l)\right).
\ee

In view of the VHS rigidity theorems,
saying that the Pauli/Yukawa coupling, or any higher tensor coupling $\lambda$, is a numerical parameter independent of the scalar v.e.v.\!\! $s$ is equivalent to saying that it is fixed by the monodromy $\Gamma$. Since $\lambda$ has pure type \eqref{reqwtu}, $\lambda$ is\footnote{\ Here we use the fact that the linear space of tensors $T^{\bullet,\bullet}$ invariant under $\Gamma$ is defined over $\mathbb{Q}$.} an element of the complexified space of \emph{Hodge tensors} \cite{reva,revb,MT4,periods}
\be
\lambda\in \mathsf{Hg}^{\bullet,\bullet}\otimes_\mathbb{Q}\mspace{-1mu}\C
\ee 
at the generic point of $S$. One shows \cite{reva,revb,MT4} that the Mumford-Tate group $\boldsymbol{M}$, eqn.\eqref{mufttate}, is precisely the subgroup of $G_m\equiv Sp(2m+2,\mathbb{Q})$ which fixes the elements of $\mathsf{Hg}^{\bullet,\bullet}$. 
 
In particular, when the dimension of the Hodge tensors of the appropriate type and symmetry is at most 1 (as it happens for the ``magical'' Pauli/Yukawa couplings) 
 such couplings -- if field-independent -- should be  \emph{integers} in some suitable normalization. They cannot be arbitrary integers, however, since the presence of a non-trivial Hodge tensor $\lambda$ implies 
strong restrictions on the group $P$ appearing in the structure theorem: $P$ should leave all Hodge tensor invariants, so the more numerical field-independent couplings there are, the smaller $P$ is. On the other hand, when $S$ has positive dimension (i.e.\! we are not in the \emph{rigid} case), the group $P$ cannot be too small since $\varpi$ is an embedding and hence
\be
\dim_\C P(\R)/H_P\geq \dim_\C S
\ee
with equality iff $S$ is locally symmetric and $\mathsf{Iso}(\widetilde{S})=P(\R)$.
 In particular, a \emph{generic} field-independent coupling $\lambda$ would imply $P$ trivial, and then $\dim_\C S=0$, leading to a contradiction in presence of light vector-multiplets. 

The fact that field-independent couplings require the existence of non-trivial Hodge tensors leads to a classification of their possible ``magical'' values in terms of the finite list of $\mathbb{Q}$-algebraic subgroups of $G_m\equiv Sp(2m+2,\mathbb{Q})$ which have Hodge representations \cite{reva,MT4} of the appropriate kind. This yields back our classification of the allowed numerical Yukawas in terms of determinant forms for rank-3 Jordan algebras. 

The situation looks very much in line with the arguments of  \cite{Heckman:2019bzm}.

\section{Functional equations for quantum-consistent $\cf$'s}\label{functeqns}

\emph{A priori,} giving a pre-potential $\cf$
only specifies (locally) 
 a covering period map $\widetilde{p}\colon \widetilde{S}\to \mathscr{D}^{(2)}_m$
which satisfies the Griffiths infinitesimal period relations \cite{GII,Gbook,periods}.
With some abuse of language, we shall say that \textit{the pre-potential $\cf$ belongs to the swampland} if there is no subgroup 
\be
\Gamma\subset \mathsf{Sym}(\widetilde{S})_\Z
\ee 
such that the quotient special K\"ahler geometry
$S\equiv \Gamma\backslash \widetilde{S}$
satisfies our swampland \textbf{Criterion \ref{tttqwaswam}}.
Said differently, in order $\cf$  \textbf{not} to belong to the swampland (i.e.\! to be quantum-consistent) there must be a monodromy group
$\Gamma$ such that
 the 
induced quotient period map
\be
p\colon S\equiv\Gamma\backslash \widetilde{S}\to \Gamma\backslash \mathscr{D}^{(2)}_m 
\ee 
satisfies the VHS structure theorem \eqref{strthemre}.
 In principle one may translate the swampland criterion
\eqref{strthemre} into a set of functional equations for  the  analytic function $\cf(X^I)$. While these equations do not look promising as an effective tool for explicit computations, they are conceptually relevant.
\medskip

Let us summarize our previous discussions of the properties of  a quantum-consistent covering geometry $\widetilde{S}$ in the following 
commutative diagram
\be\label{ddrigr}
\begin{gathered}
\xymatrix{\widetilde{S}\; \ar@<0.2ex>@/_2.2pc/[drrrr]_{\widetilde{z}}\ar@<1ex>@/^2.5pc/[rrrr]^{\widetilde{p}}\ar@{^{(}->}[rr]^(0.4){\varpi} && P(\R)/H_P\;\ar[drr]^{\sigma\widetilde{\iota}}\ar@{^{(}->}[rr]^(.6){\widetilde{\iota}} && \mathscr{D}^{(2)}_m\ar@{->>}[d]^(.4){\sigma}\\
&& &&\mathscr{D}^{(3)}_m}
\end{gathered}
\ee
The inclusion $\widetilde{\iota}$ is induced by an  
\emph{irreducible}, faithful, real, symplectic, Hodge representation \cite{reva,revb,MT4} $\varrho$ of the real Lie group $P(\R)$ which makes the following diagram to commute:
\be\label{hodgrep}
\begin{gathered}
\xymatrix{P(\R)\ar@{->>}[d]\ar[r]^{\varrho} & G_m(\R)\ar@{->>}[d]
\\
P(\R)/H_P \ar[r]^{\widetilde{\iota}}& \mathscr{D}^{(2)}_m\ar@{->>}[r]^{\varpi_3} &\mathscr{D}^{(3)}_m.}
\end{gathered}
\ee
By definition of the $\mathbb{Q}$-algebraic group $P\mspace{-2mu}$,
$\varrho$ is \textit{defined over $\mathbb{Q}$.}
The image of the composite map $\varpi_3\widetilde{\iota}$ is
a $P(\R)$-homogeneous complex manifold \be
\mathscr{D}_P^{(3)}=P(\R)/[U(1)\times J]\equiv P(\R)/H^{(3)}_P
\ee 
which may be easily constructed with the help of the weight-3 Hodge representation $\varrho$. When (as it looks to be the general case) $\varrho$ is a fundamental representation associated to  a node in the Vogan diagram of $P(\R)$,
$J$ is the real Lie group whose diagram is obtained by deleting from the Vogan graph of $P(\R)$ the $\varrho$ node.

\begin{exe}[Figure \ref{fig2}] For instance, when $P(\R)=E_{7(-25)}$ and $\varrho$ is the $\boldsymbol{56}$, so that $m=27$, we get the holomorphic domain
\be\label{pqw12lll}
\mathscr{D}_{E_{7(-25)}}^{(3)}= E_{7(-25)}\big/[U(1)\times E_6],
\ee
which in this particular example coincides with the MT domain $P(\R)/H_P$. Since the isotropy group in \eqref{pqw12lll} is compact,  $\mathscr{D}_{E_{7(-25)}}^{(3)}$ is a Hermitian symmetric space. By general theory we conclude that when $P(\R)=E_{7(-25)}$ the special K\"ahler manifold $S$
is locally isometric to the Cartan symmetric domain
\eqref{pqw12lll}. \end{exe}

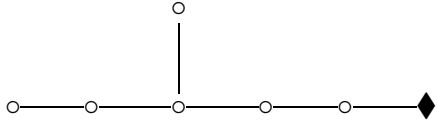
\begin{figure}
$$
\xymatrix{&& \circ\ar@{-}[d]\\\circ\mspace{-5mu}\ar@{-}[r]&\mspace{-5mu}\circ\mspace{-5mu}\ar@{-}[r]&\mspace{-5mu}\circ\mspace{-5mu}\ar@{-}[r]&\mspace{-5mu}\circ\mspace{-5mu}\ar@{-}[r]&\mspace{-5mu}\circ\mspace{-5mu}&\mspace{-5mu}\blacklozenge\ar@{-}[l]}
$$
\caption{\label{fig2} The diagram of the Klein geometries for the example $P(\R)=E_{7(-25)}$ and $\varrho=\boldsymbol{56}$.
The white/black color specifies the particular real form of the Lie group $E_7$, the lozenge stands for the highest weight of the representation $\varrho$. Since $\varrho$ corresponds to the black node, which  is an extension node, $H^{(3)}_{E_{7(-25)}}\subset E_{7(-25)}$
is a maximal compact subgroup, so $\mathscr{D}^{(3)}_{E_{7(-25)}}$,
and hence the special geometry, is locally symmetric. 
}
\end{figure} 

\medskip 

From the analysis in \S.\,\ref{preliminary} we know that the universal cover $\widetilde{S}$ is identified with a complex submanifold of the Griffiths domain $\mathscr{D}^{(2)}_m$, and that its \emph{naive} symmetry group $\mathsf{Sym}(\widetilde{S})$ is the subgroup of the automorphism group $Sp(2m+2,\R)$ of the ambient space $\mathscr{D}^{(2)}_m$ which fixes (set-wise) $\widetilde{S}$.
From diagram \eqref{ddrigr} we learn that $\widetilde{S}$ is completely contained in a \emph{single} orbit of the subgroup $P(\R)\subseteq Sp(2m+2,\R)$,  so that
\be\label{syminclusi}
\mathsf{Sym}(\widetilde{S})\subseteq P(\R),
\ee
and all naive symmetries arise from automorphisms of the smaller homogeneous ambient space
\be
 \mathscr{D}^{(2)}_P\overset{\rm def}{=}P(\R)/H_P.
 \ee 
From the diagram we also see that $\mathsf{Sym}(\widetilde{S})$ is, tautologically, a group of automorphisms of the contact manifold $\mathscr{D}_{m}^{(3)}$ (in facts of its subspace\footnote{\ Note that $\mathscr{D}_P^{(3)}$ is not a contact manifold in general, as the \textbf{Example} in figure \ref{fig2} shows.} $\mathscr{D}_P^{(3)}$) which fixes the Legendre submanifold
$\widetilde{L}=\widetilde{z}(\widetilde{S})$
whose generating function is the homogeneous pre-potential $\cf$ 
\be\label{whichformFF1}
\cf(X^0,X^i)= (X^0)^2\, F(X^i/X^0),\qquad z^i=X^i/X^0,\quad i=1,2,\cdots, m.
\ee
In particular,
$\widetilde{L}\subset\mathscr{D}_{m}^{(3)}$ should be invariant under
$\Gamma\subset Sp(2m+2,\R)$ acting as a group of automorphisms of the ambient contact domain.
Then, for each $(2m+2)\times (2m+2)$ matrix
\be\label{matgamma}
\gamma=\begin{pmatrix}A & B\\ C & D\end{pmatrix}\in \Gamma\subset Sp(2m+2,\mathbb{Z}),
\ee
we get a system of $(m+1)$ functional equations
for the holomorphic functions $\cf_I(X^K)$: 
\be\label{kqawert1}
\cf_I\mspace{1mu}\big(C^{KL}\cf_L+{D^K}_LX^L\big)={A_I}^J\mspace{1mu}\cf_J\big(X^K\big)+B_{IJ}\mspace{1mu}X^J.
\ee
The meaning of these equations is a bit subtle because  the function $\cf$ may be multi-valued.
In facts this is the generic case. 
The functional equations \eqref{kqawert1} refer to the \emph{global}
analytic continuation of $\cf$ to all of $\widetilde{L}$.

 In most applications, the putative pre-potential $\cf$ (or rather the function $F(z^i)$ in eqn.\eqref{whichformFF1}) has a local analytic expression as a series which converges only in some small domain $U\subset \C^{m}$. Then the  functional equations associated to elements of the subgroup $\Lambda_U\subset \Gamma$ which maps $U$ into itself may be used rather straightforwardly to constrain the terms in the series, but for most elements $\gamma\in \Gamma$ the functional equations \eqref{kqawert1} relate $U$ to far away regions of $\widetilde{L}$ where the  analytic expression of $\cf$ is not known.
Typically we do not even know if a global analytic continuation exists,
and when it exists, it is hard to establish whether its has the right branching properties.

The existence of the global analytic continuation of $\cf$ to the full $\widetilde{L}$  should be seen as a very subtle part of our swampland criterion.
As far as its uniqueness goes, there is some evidence that the correct statement is that the difference between two determinations of the global holomorphic function $\cf$
should be \emph{physically invisible,} that is, the values of all observables are independent of the choice of determination of $\cf$. This should also be seen as
 a fine point in our swampland criterion. 
 \medskip

We conclude 
\begin{cri}\label{critfunct} Let the holomorphic function $\cf(X^J)$,
homogeneous of degree 2, be the global analytic continuation of a local pre-potential which does \textbf{not}
belong to the swampland. Then $\cf$ satisfies the system of functional equations \eqref{kqawert1}
for all elements of a finitely-generated group $\Gamma$ such that
\be
\Gamma\subset Sp(2m+2,\Z),\qquad
\overline{\Gamma}^{\mspace{2mu}\mathbb{Q}}=P\subseteq Sp(2m+2,\mathbb{Q}).
\ee
\end{cri}

Since $\Gamma$ is infinite this looks like a huge set of equations. However not all of them are independent:
it suffices to impose the ones corresponding to the finitely-many generators of $\Gamma$. 
For instance,
consider the generic situation  where 
$\overline{\Gamma}^{\mspace{2mu}\mathbb{Q}}\equiv Sp(2m+2,\mathbb{Q})$. If the monodromy is arithmetic, $\Gamma$ is a finite-index quotient of the maximal arithmetic subgroup $Sp(2m+2,\Z)$.
The simplest possibility is $\Gamma\equiv Sp(2m+2,\Z)$: this happens, say, for the universal family of complete intersections of two cubics in $\mathbb{P}^5$ \cite{beauv}.
Then we have one independent system of functional equations of the form \eqref{kqawert1} per generator of the Siegel modular group. Ref.\!\!\cite{gen} yields an economic set of generators
for $Sp(2m+2,\Z)$ given by (at most) 3 explicit matrices. We get a set of $3m+3$ functional equations for the $\cf_I$. 
However, typically there exist other elements $\gamma\in \Gamma$ (or even infinite subgroups $\Lambda\subset \Gamma$) which lead to simpler relations which yield elementary but useful constraints on  $\cf$.
\textbf{Criterion \ref{critfunct}} is a concrete (if unpractical)  realization of the physical idea that in a consistent quantum gravity the gauge group (i.e.\! $\Gamma$) determines the Lagrangian (i.e.\! $\cf$).

\section{Dicothomy}

\subsection{Proof of dicothomy}

In this section we prove \textbf{Fact \ref{cido}} (dicothomy), that is,

\begin{fact}\label{ddiiye} Let  $S\equiv \Gamma\backslash \widetilde{S}$ be a special K\"ahler manifold (with $\widetilde{S}$ smooth simply-connected). 
Assume, in addition, that $S$ satisfies our  \emph{swampland structural criterion:} namely, its underlying Griffiths period map 
\be
p\colon S\to \Gamma\backslash D
\ee
satisfies the VHS structure theorem \eqref{strthemre}. 
We write $P$ for the \emph{semi-simple} $\mathbb{Q}$-algebraic group $\overline{\Gamma}^{\mspace{2mu}\mathbb{Q}}$
and $\mathfrak{p}^\R\equiv\mathfrak{p}\otimes_\mathbb{Q}\mspace{-2mu}\R$ for the real Lie algebra of the group of its real points $P(\R)$. Let $\mathfrak{sym}(\widetilde{S})\subseteq \mathfrak{p}^\R$
 be the  real Lie algebra of the naive symmetry group
 $\mathsf{Sym}(\widetilde{S})$ of the covering special K\"ahler geometry $\widetilde{S}$.
Then
\be
\text{either}\ \ \mathfrak{sym}(\widetilde{S})=0\quad\text{or}\ \   
\mathfrak{sym}(\widetilde{S})=\mathfrak{iso}(\widetilde{S})=\mathfrak{p}^\R.
\ee 
In the second case $\widetilde{S}$ is isometric to the Hermitian symmetric space $P(\R)/K$ where 
$P(\R)$ is either  $SU(1,m)$ or one of the `magic' isometry groups in \textsc{Table} \ref{tableiso}.
\end{fact}
That is, either the naive symmetry group of $\mathsf{Sym}(\widetilde{S})\subseteq \mathsf{Iso}(\widetilde{S})$ is purely discrete, or $\widetilde{S}$ is a Hermitian symmetric space and hence, being special K\"ahler, is one of the spaces in \textsc{table} \ref{tableiso}. In this second case the moduli space $S$ is a Shimura variety. As already mentioned, \textbf{Fact \ref{ddiiye}} is similar in spirit  to the DG statement \textbf{Fact \ref{weakre}} (but more general).

\begin{proof}
We may assume $\widetilde{S}$ to be irreducible of positive dimension, since simply-connected reducible special K\"ahler manifolds have been classified (see \cite{ferrara}\!\!\cite{book}) and are  all symmetric spaces,  so the statement is trivially true in the reducible case.
Then the $\mathbb{Q}$-Zariski closure $\overline{\Gamma}^{\mspace{2mu}\mathbb{Q}}\equiv P$ of the monodromy group $\Gamma$
is an almost-simple $\mathbb{Q}$-algebraic group (cfr.\! \textbf{Lemma \ref{lllemmaq}}),
and its $\mathbb{Q}$-Lie algebra $\mathfrak{p}$ is non-zero and simple. The structure theorem yields factorizations
of the period map $p$ and of its lift $\widetilde{p}$ as in the commutative diagram:
\be\label{oooowq}
\begin{gathered}
\xymatrix{\widetilde{S}\ar@{->>}[d]\ar@/^1.6pc/[rr]^{\widetilde{p}}\ar[r]_{\widetilde{\varpi}\phantom{mmmm}}&  P(\R)/H_P\;\ar@{->>}[d] \ar@{^{(}->}[r]_{\phantom{mmn}\widetilde{\iota}} &\; \mathscr{D}_m^{(2)}\ar@{->>}[d]\\
S\ar@/_1.6pc/[rr]_{p}\ar[r]^{\varpi\phantom{mmmm}}& \Gamma\backslash P(\R)/H_P\; \ar@{^{(}->}[r]^{\phantom{mmn}\iota} &\;\Gamma\backslash \mathscr{D}_m^{(2)},} 
\end{gathered}
\ee
where the vertical double-headed arrows are canonical projections,
and $\iota$, $\widetilde{\iota}$ closed inclusions.
$\widetilde{\iota}$ is induced by the irreducible representation $\varrho$ as in the diagram \eqref{hodgrep}.
It is easy to see that when $\varrho$ is irreducible, but not absolutely irreducible, the special K\"ahler manifold $\widetilde{S}$ is symmetric, in facts a complex hyperbolic space $H\C^m\equiv SU(1,m)/U(m)$ \cite{book}. Then we may assume $\varrho$ to be absolutely irreducible without loss.
\smallskip

 The inclusion \eqref{syminclusi} of the Lie groups induces the inclusion of the corresponding (real) Lie algebras
 \be
\mathfrak{sym}(\widetilde{S})\subseteq \mathfrak{p}^\R.
 \ee
On the other hand, 
\be
\Gamma\subset\mathsf{Sym}(\widetilde{S})\cap P(\Z),
\ee 
so that the subalgebra $\mathfrak{sym}(\widetilde{S})$ is preserved by the adjoint action of $\Gamma$
\begin{equation}
x\in \mathfrak{sym}(\widetilde{S})\quad\Rightarrow\quad \gamma x\gamma^{-1}\in \mathfrak{sym}(\widetilde{S})\ \ \text{for all }\gamma\in \Gamma.
\end{equation}
Thus $\mathfrak{sym}(\widetilde{S})$ is an 
$\mathrm{ad}\,\Gamma$-invariant  real Lie subalgebra of the simple real Lie algebra $\mathfrak{p}^\R$. 
 \textbf{Fact \ref{ddiiye}} then follows from the following \textbf{Lemma}.
\end{proof}

\begin{lem} Let the $\mathbb{Q}$-algebraic group $P\equiv \overline{\Gamma}^{\mspace{2mu}\mathbb{Q}}$ be simple, and let $\mathfrak{p}$ be its
$\mathbb{Q}$-Lie algebra. Set $\mathfrak{p}^\R=\mathfrak{p}\otimes_\mathbb{Q}\mspace{-2mu}\R$ for the Lie algebra of $P(\R)$.
Suppose that $\mathfrak{s}\subseteq \mathfrak{p}^\R$ is an $\mathrm{ad}\,\Gamma$-invariant 
real Lie subalgebra. Then either $\mathfrak{s}=0$ or $\mathfrak{s}=\mathfrak{p}^\R$.
\end{lem}

\begin{proof} The inclusion of $\R$-spaces $\mathfrak{s}\hookrightarrow\mathfrak{p}^\R$ defines an element of $\mathrm{End}_{\R[\Gamma]\textsf{-mod}}(\mathfrak{p}^\R)$,
that is, a $\Gamma$-invariant tensor
which belongs to the $\R$-space
\be
T_\R\overset{\rm def}{=} (\mathfrak{p}\otimes \mathfrak{p}^\vee)\mspace{-1mu}\otimes_\mathbb{Q}\mspace{-2mu}\R\subset (V\otimes V^\vee)^{\otimes 2}\mspace{-1mu}\otimes_\mathbb{Q}\mspace{-2mu}\R,
\ee where $V$ is the underlying $\mathbb{Q}$-space of the VHS, namely the representation space of the \emph{rational,} irreducible, faithful, symplectic Hodge representation $\varrho$, see
eqn.\eqref{hodgrep}.
Therefore $\mathfrak{p}\subset\mathrm{End}(V)$ \emph{via}
$\varrho$. 
Since $\Gamma\subset GL(V)$, the real vector subspace $T_\R^\Gamma\subset T_\R$ of vectors fixed by $\Gamma$ is defined over $\mathbb{Q}$, 
that is,
\be
T_\R^\Gamma=(\mathfrak{p}\otimes \mathfrak{p}^\vee)^\Gamma\mspace{-2mu}\otimes_\mathbb{Q}\mspace{-2mu}\R.
\ee
All tensors in $(\mathfrak{p}\otimes \mathfrak{p}^\vee)^\Gamma$ are also fixed by the $\mathbb{Q}$-algebraic group $P\equiv\overline{\Gamma}^{\mspace{2mu}\mathbb{Q}}$, hence by the group of its real points $P(\R)$. Then
$\mathfrak{s}\hookrightarrow\mathfrak{p}^\R$ is also an element of $\mathrm{End}_{P(\R)\textsf{-mod}}(\mathfrak{p}^\R)$, i.e.\! a $P(\R)$ intertwiner,
 hence zero or an isomorphism because the Lie group $P(\R)$ is simple.
\end{proof}

\subsection{Existence of quantum-consistent symmetric  geometries}

By dicothomy, the quantum-consistent 4d $\cn=2$ \textsc{sugra} are divided in two classes: (1) models whose 
scalars' manifold $S$  is locally symmetric;
and (2) models whose $S$  has no \emph{local} Killing vector at all.

We wish to show that the first class is not empty.
We distinguish two cases, namely quadratic versus cubic pre-potentials. We consider only geometries of positive dimension.

\subparagraph{Quantum-consistent quadratic special geometries.} 

There are known examples of geometric families of
Calabi-Yau 3-folds whose period map is 
given by a purely quadratic pre-potential 
\be
\cf=\frac{1}{2}\,A_{IJ}\,X^IX^J.
\ee 
By the structure theorem, this happens if and only if 
the $\boldsymbol{(2m+2)}$-dimensional representation of $\Gamma$ (or equivalently of the real Lie group $P(\R)$) is irreducible but not absolutely irreducible.
A mechanism which guarantees reducibility over $\C$ is described in ref.\!\cite{maxsy} (see, in particular, his \textbf{Theorem 2.5}). The idea is that there is a global automorphism $\alpha$ of the universal deformation space of the CY which acts on $H^{3,0}$ as multiplication by $\eta\neq \pm1$. Clearly $\alpha$ centralizes $\Gamma$, so that its representation becomes reducible over $\C$. 
In \cite{maxsy} there are 6 explicit examples of such CY 3-folds  
 with Hodge numbers
 \be
1\leq h^{2,1}\leq 6,\quad h^{1,1}=84-11\cdot h^{2,1}\quad \text{and}\quad h^{1,1}=61,\quad h^{2,1}=1.
\ee
Another source of examples with $h^{2,1}=1$ arises from the classification of Picard-Fuchs operators for 
Calabi-Yau manifolds with one-dimensional moduli spaces. It is clear that $\cf$ is quadratic precisely when the Picard-Fuchs operator has order 2 instead of the generic 4, so that over $\C$ there are only two linear independent periods so that $\cf_0$ and $\cf_1$ are complex linear combinations of $X^0$ and $X^1$. Seven explicit examples of this situation are described in ref.\!\cite{orphans}: they corresponds to the resolved double octics arising from  the Meyer arrangements of eight lines \cite{meyer}
of numbers 
\be
\bf 4,\ 13,\ 34,\ 72,\ 261,\ 264,\ 270.
\ee

\subparagraph{Quantum-consistent `magic' cubic pre-potentials.} The Calabi-Yau manifolds
which are finite quotients of either an Abelian variety $A$ or a product of a K3 surface with an elliptic curve \cite{OSaku}, have cubic pre-potentials corresponding to arithmetic quotients of the reducible symmetric special geometry
\be
SL(2,\R)/U(1)\times SO(2,m-1)/[SO(2)\times SO(m-1)].
\ee    
 These examples are discussed in detail in the next two sections. In addition to these ones, we have the consistent truncation of higher supergravities to the subsector invariant under a discrete group.
Such truncated models arise, for instance, in the compactifications of  Type II on a Kummer CY 3-fold of the form $T^6/\Sigma$,
 where $\Sigma$ acts non-freely, when we froze the twisted sector degrees of freedom to zero. The sub-sector 
 special geometries arising in this way are listed in the tables of ref.\!\!\cite{CFG}.

\section{Swampland criterion $\leadsto$  mirror symmetry}\label{mirror:sec}

In this section we study in slight more detail the asymptotic behavior  at infinity of a special K\"ahler geometry which satisfies our swampland criterion, and see how it automatically reproduces most of the properties we usually associate with mirror symmetry \cite{candelas,mirrorbook} without any need to assume that the special geometry arises from a pair $X$, $X^\vee$ of mirror Calabi-Yau 3-folds. The relevant properties just 
follow from quantum consistency of the 4d $\cn=2$ effective gravity theory
described by that special geometry.
For a deeper perspective on the asymptotic structure of a quantum-consistent pre-potential $\cf$, see Deligne \cite{deligneL}.
\medskip

The validity of the VHS structure theorem implies, in particular, that the asymptotic behavior of the period map at infinity is described by the (multi-variable)
nilpotent-orbit theorem \cite{schm}. 
As we approach infinity in $S$ along a certain direction,
we may end up with 
 nilpotent orbits of several different kinds: they are classified in terms of the degenerating mixed Hodge structures which may arise \cite{multisl2}.
Here we focus on just one particularly simple possibility, namely what happens when we approach a MUM (maximal unipotent monodromy) point \cite{morrison}.
The special situation we have in mind is as follows:
our special K\"ahler manifold has the form
$S=\overline{S}\setminus Y$, for $\overline{S}$ a compact complex space and $Y$ a \textsc{snc} divisor, and there is a small open set
$U\subset \overline{S}$ and local coordinates $x_i$ such that 
\be
U\cap S= U\setminus \{x_1x_2\cdots x_m=0\},\qquad \dim_\C S=m,
\ee
while the monodromy $\gamma_i\in\Gamma$ around the
$i$-th component $\{x_i=0\}$ of $Y$ is non-trivial.
We may assume with no loss that $\Gamma$ is neat, so that $\gamma_i=\exp(N_i)$ for some
nilpotent element $N_i\in \mathfrak{sp}(2m+2,\mathbb{Q})$. The $N_i$ commute between themselves, and 
 $(\sum_i\lambda_i N_i)^4=0$ for all choices of the coefficients $\lambda_i\in\R$ by the strong monodromy theorem.
The point 
\be
\big\{x_1=x_2=\cdots=x_m=0\big\}\in \overline{S}
\ee 
is a MUM point iff $(\sum_i\lambda_i N_i)^3\not\equiv 0$. 
 
Consider the commutative, unipotent, $\mathbb{Q}$-algebraic subgroup $J\subset \overline{\Gamma}^{\mspace{2mu}\mathbb{Q}}$
\be\label{whattttL}
J(\mathbb{Q})=\left\{\exp\mspace{-4mu}\left[\sum\nolimits_i q^i N_i\right]\;\colon\; (q^i)\in\mathbb{Q}^m\right\}\subset P(\mathbb{Q})\subseteq Sp(2m+2,\mathbb{Q}).
\ee
Its group of integral points $J(\Z)\equiv J(\mathbb{Q})\cap G_m(\Z)$ contains the finite-index subgroup  
\be
\Lambda\overset{\rm def}{=}\left\{\prod_i\gamma_i^{n_i}\ :\ n_i\in\Z\right\}\subset\Gamma.
\ee 
The group $\Lambda\subset Sp(2m+2,\Z)$
acts on the covering Legendre manifold $\widetilde{L}\equiv\widetilde{z}(\widetilde{S})$ through
automorphisms of the ambient space $\mathscr{D}^{(3)}_{m}$, i.e.\! by integral symplectic rotations of its homogeneous coordinates $(\cf_I,X^J)$.
 
The universal cover of $U\cap S\cong (\Delta^*)^m$ is
\be
\widetilde{U\cap S}=\ch^m,\quad\text{where}\quad \ch=\big\{w\in \C\;\big|\; \mathrm{Im}\,w>0\big\},
\ee
with cover map 
 \be
 w^i\mapsto x_i=\exp(2\pi i\mspace{1mu} w^i),\quad w^i\in\ch,\quad i=1,\dots,m.
 \ee 
The MUM point corresponds to the limit $\mathrm{Im}\,w^i\to\infty$ for all $i$. The $i$-th monodromy element 
$\gamma_i\in C$ acts on the covering local coordinates $w^j$
as $w^j\to w^j+{\delta^j}_{i}$:
this reflects in an integral symplectic action on the $X^I$.
Then, for a suitable choice of symplectic homogeneous coordinates $X^I$, the action of the $\gamma_i$'s  takes the form 
\be\label{somegammma2}
\gamma_i \colon \begin{cases}
X^j\mapsto X^j+ {f^{\mspace{2mu}j}}_i\mspace{3mu} X^0\\
X^0\mapsto X^0
\end{cases}\ \ \text{with}\quad z^{\mspace{1mu}j}\overset{\rm def}{=}\frac{X^j}{X^0}={f^{\mspace{2mu}j}}_i\, w^i,\qquad i,j=1,\dots,m,
\ee
where ${f^{\mspace{2mu}j}}_i$ is an integral matrix
such that 
\be\label{deetrq}
\det f=[J(\Z)\colon \Lambda].
\ee 
 In most examples $\det f=1$ and there is no need to distinguish between $z^i$'s and $w^i$'s. 

The subgroup $\Lambda\subset \Gamma$ maps small neighborhoods $U$ of the MUM point into themselves.  
Then we
may use directly the functional equations \eqref{kqawert1} with respect to elements of $\Lambda$ to constrain the asymptotic structure of $\cf$ at the MUM point. The requirement that each $\gamma_i$ acts by an integral $Sp(2m+2,\Z)$-rotations
on the full vector
$(\partial_{X^I}\cf, X^J)$ (cfr.\! eqn.\eqref{kqawert1})
requires that $\cf(X^I)$ satisfies for all $i=1,\dots,m$
a functional equation of the form (cfr.\! eqn.\eqref{changeF})
\be\label{pqw123}
\cf(X^0,X^j+{f^{\mspace{2mu}j}}_i\mspace{3mu} X^0)=
\cf(X^I)+\frac{1}{2} S^{(i)}_{IJ}\mspace{2mu} X^IX^J
\ee
for some symmetric integral matrix $S_{IJ}^{(i)}$.
This condition restricts the function $F(z^j)$ 
(see eqn.\eqref{whichformFF1}) to the general form\footnote{\ Here $Q=\{(n_1,\dots,n_m)\in \Z^m\;:\; n_i\geq0\}$ is the positive $m$-tant.} 
\be\label{whichformFFX}
F(z^i)= -\frac{d_{ijk}}{3!}\, z^iz^jz^k+ \frac{b_{ij}}{2}\, z^i z^j+ \frac{c_i}{24}\, z^i + \frac{\zeta(3)}{2(2\pi i)^3}\,e+\overbrace{\sum_{\vec n\in Q\setminus 0} \lambda_{\vec n}\;e^{2\pi i\mspace{1.5mu} \vec n\cdot\vec w}}^{\text{Fourier series}},
\ee
($z^{\mspace{1mu}j}\equiv {f^{\mspace{2mu}j}}_i\, w^i$)
where the coefficients $d_{ijk}$, $b_{ij}$ and $c_i$ are rational numbers which need to satisfy strict integrality
conditions  in order to satisfy \eqref{pqw123} with ${f^j}_i$ and $S^{(i)}_{IJ}$ integral. In particular,
\be
d_{ijk}\mspace{2mu}{f^k}_\ell\in\Z.
\ee
The constant coefficient in the \textsc{rhs} of \eqref{whichformFFX} has been written with a particular normalization for later convenience. From the nilpotent orbit theorem we see that 
\be
N_iN_jN_k=d_{ijk}\,M
\ee
for a certain non-zero matrix $M$, so MUM point $\Rightarrow$ the cubic form $d_{ijk}$ is non-zero.
Eqn.\eqref{whichformFFX} may be seen as a strengthened version of the nilpotent-orbit theorem. At this stage the coefficients $e$ and $\lambda_{\vec n}$ are still arbitrary complex numbers.

Now we appeal (a bit heuristically)
to some more \emph{arithmetic} Hodge theory.
The Fourier series in eqn.\eqref{whichformFFX} stands for a generic holomorphic function, defined in the neighborhood $U$ (possibly after restricting it), and periodic under the integral shifts $w^i\to w^i+{\delta^i}_j$, which yields the general solution to the functional equations \eqref{pqw123} associated to the Abelian subgroup $\Lambda\subset \Gamma$.
Such general function was written as a sum
over elements of the standard basis for the periodic holomorphic functions in $\ch^m$, namely $\{e^{2\pi i\mspace{1.5mu} \vec n\cdot\vec w}\}_{\vec n\in Q}$.
However this way of writing the solution is not intrinsic, and the expression may be made more illuminating  by going to a more natural basis of such functions $\{\boldsymbol{L}(\vec n\mspace{-2mu}\cdot\mspace{-2mu}\vec w)\}_{\vec n\in Q}$.

Which basis is the ``intrinsic'' one in the present context?

In the MUM limit $\mathrm{Im}\,w^i\to\infty$ the VHS
degenerates to a mixed Hodge structure (more precisely to a mixed Hodge-Tate structure \cite{deligneL}) of maximal weight 3,
so that the  natural candidate is the function $\boldsymbol{L}(w)$ which canonically represents the Tate Hodge structure $\mathbb{Q}(3)$, i.e.\! the natural function is the one which satisfies the appropriate Hodge-theoretic functional equations.
From, say, \textbf{Theorem 7.1} in the review \cite{polylog} we learn that such function $\boldsymbol{L}(w)$ is the trilogarithm
\be
\boldsymbol{L}(w)\overset{\rm def}{=}\frac{1}{(2\pi i)^3}\,\mathrm{Li}_3(e^{2\pi i w})=\sum_{k\geq 1}\frac{e^{2\pi i k\mspace{1mu} w}}{(2\pi i\,k)^3}.
\ee
Formally, all Fourier series of the form \eqref{whichformFFX} 
may be rewritten
as a series in the $\boldsymbol{L}(z)$.
Therefore we shall rewrite
\be\label{seriesrep}
\sum_{\vec n\in Q\setminus 0} \lambda_{\vec n}\;e^{2\pi i\mspace{1.5mu} \vec n\cdot\vec w}\to
\sum_{\vec n\in Q\setminus 0} N_{\vec n}\;\boldsymbol{L}(\vec n\mspace{-2mu}\cdot\mspace{-2mu}\vec w),
\ee
for the new (\emph{a priori}  complex) coefficients $N_{\vec n}$
\be
N_{\vec n}=(2\pi i)^3\!\sum_{d\,\mid\, \gcd\{\vec n\}} \mu(d)\,\frac{\lambda_{\vec n/d}}{d^3}\qquad \text{($\mu$ $\equiv$ M\"obius function).}
\ee 
The structural swampland criterion suggests that the intrinsic coefficients in the expansion of $\cf$ at a MUM point are the $N_{\vec n}$ not the $\lambda_{\vec n}$.
In particular the swampland philosophy suggests that the $N_{\vec n}$'s -- contrary to the $\lambda_{\vec n}$'s -- \textit{should have interesting arithmetic properties} in quantum-consistent $\cn=2$ supergravities with a MUM point.
\smallskip

Let us see how these arithmetic conditions arise.
While $\boldsymbol{L}(\vec n\mspace{-3mu}\cdot\mspace{-3mu}\vec w)$ may be defined to be univalued
in the neighborhood\footnote{\ $\overline{\widetilde{L}}$ stands for a suitable compactification of the covering  Legendre manifold $\widetilde{L}$ \cite{morrison}.} $U\subset \overline{\widetilde{L}}$ around the MUM point, its global analytic continuation is certain \emph{not} univalued. Indeed $\boldsymbol{L}(w)$ has an interesting monodromy  \cite{polylog} which is  the basic reason why this special function describes the relevant Hodge structure: this is the very  property which makes $\boldsymbol{L}(w)$ prominent  in the swampland program. 
From the discussion in the previous section, we know that the swampland condition refers to the global analytic continuation of the pre-potential $\cf$, not to its
local expression in a particular region of  $\overline{\widetilde{L}}$. Thus, to procede with the swampland program, we are forced to
understand the physical implications of the multivaluedness of $\boldsymbol{L}(\vec n\mspace{-1mu}\cdot\mspace{-1mu}\vec w)$.
\smallskip

Consistency of the swampland scenario we are proposing leads to the following principle:\vglue6pt

{\it In a quantum-consistent 4d $\cn=2$ supergravity, the ambiguity in the global definition of each term in the sum in the \textsc{rhs} of \eqref{seriesrep} should be \emph{physically invisible,}
that is, no physical observable should depend on the particular determination we choose. }\vglue6pt

Two determinations
of $\boldsymbol{L}(\vec n\mspace{-1mu}\cdot\mspace{-1mu}\vec w)$ differ by an integral multiple of $(\vec n\mspace{-1mu}\cdot\mspace{-1mu}\vec w)^2/2$. For, say, $f=1$, this reflects to an indeterminacy of $\cf$ of the form
\be
\cf(X^I)+\frac{1}{2}\text{(integer)} N_{\vec n}\, (n_i X^i)^2
\ee
and this is invisible precisely when it can be compensated by a $Sp(2m+2,\Z)$-rotation of the electro-magnetic frame. This requires the $N_{\vec n}$ to be integers.\footnote{\ More precisely, one has the slightly weaker condition that $\gcd\{{f_i}^kn_k {f_j}^ln_l\}\mspace{2mu}N_{\vec n}$ should be integral (${f_i}^j$ is the inverse of the matrix ${f^j}_i$). }

 The same token leads to the condition $e\in 2\mspace{1mu}\Z$.
Indeed the constant term in eqn.\eqref{whichformFFX} may be simply written as
\be
\frac{e}{2}\mspace{2mu}\boldsymbol{L}(0),
\ee 
so it may be absorbed in the sum by extending the summation index $\vec n$ to zero and setting
$N_{\vec n=0}\equiv e/2$. Applying to $N_{\vec n=0}$ the same integrality condition which holds for the $\vec n\not=0$ terms then yields $e/2\in\Z$.
Then \textbf{Criterion \ref{critfunct}}, which follows from our swampland condition, implies that around a MUM point with $f=1$ ($w^i\equiv z^i$)
\be\label{whichformFFY}
F(z^i)= -\frac{d_{ijk}}{3!}\, z^iz^jz^k+ b_{ij}\, z^i z^j+ \frac{c_i}{24}\, z^i + \overbrace{\frac{\zeta(3)}{2(2\pi i)^3}\,e}^{\text{``loop''}}+\overbrace{\sum_{\vec n\in Q\setminus 0} N_{\vec n}\;{\boldsymbol{L}(\mspace{1.5mu} \vec n\cdot\vec z})}^{\text{``instanton corrections''}},
\ee
with integral coefficients $d_{ijk}$, $b_{ij}$, $c_i$, $e/2$ and $N_{\vec n}$.
\medskip

The large-volume expansion of the pre-potential $\cf$ for type IIA compactified on a Calabi-Yau 3-fold $X$ has exactly the form \eqref{whichformFFY} where $d_{ijk}\in\Z$ is the triple intersection of an integral basis $\{\omega_i\}$ of harmonic (1,1) forms, $c_i=\int c_2\wedge \omega_i\in\Z$, while \be
e=\int_X c_3\equiv 2(h^{1,1}-h^{2,1})\in 2\mspace{2mu}\Z
\ee 
is the Euler characteristic of $X$ \cite{candelas}.
We dubbed 
 the last term in \eqref{whichformFFY} the ``instanton corrections'' since this term has this physical interpretation in this particular class of models; for the same reason we denoted the constant term as the ``loop'' ($\equiv$ perturbative) contribution. 
  In the Type IIA  large CY volume set-up, the coefficients $N_{\vec n}$ are (non-negative) integers because they count
  the rational curves on $X$ of class $\vec n\mspace{-1mu}\cdot\mspace{-1mu}\vec \omega$ \cite{candelas}. 
  
 The coefficient $c_i$ in eqn.\eqref{whichformFFY}
 drops out of the special K\"ahler metric $G_{i\bar j}$ so it is ambiguous at the differential-geometric level. However, in the full VHS its value (mod 24) is crucial to reproduce the correct $K$-theoretic quantization of electro-magnetic charges,
see e.g.\! the detailed analysis in \S.\, 4.1.2 of
\cite{piopio}. For simply-connected 3-CY one gets
$c_i=\int c_2\wedge \omega_i\in\Z$ (mod 24). 
 \medskip 
  
 We see that essentially all the general predictions that we may infer from the existence of an actual pair $X$, $X^\vee$ of mirror Calabi-Yau 3-folds 
 are valid in full generality for any
 \emph{abstract} 4d $\cn=2$ effective supergravity (with a $f=1$ MUM point at infinity) \emph{provided}
 it satisfies the structural swampland criterion.
 Morally speaking, quantum consistency of gravity {implies} mirror symmetry as a special case. 
 
 \medskip
 
 \subparagraph{The set $N\mspace{-4mu}E$ of non-empty instanton sectors.} There is one counter-intuitive point to be stressed. 
 At first sight only positive instanton-charges $\vec n$ may contribute
to the local sum \eqref{whichformFFY} in $U$:
  
 \be
 N\mspace{-4mu}Z_U \overset{\rm def}{=}\Big\{ \vec n\in \Z^m\;\colon\; N_{\vec n}\neq0\Big\}\subset Q
 \ee  
 In the geometric set-up of Type IIA on a 3-CY, this restriction to positive instanton-charges is obvious: only \emph{effective} curves -- whose volume is positive -- may possibly contribute in the large-volume limit.
Thus, naively, the only instanton-charge sectors which do contribute to $\cf$ are the ones with instanton charge in $C_M\subset \Z^m$, the Mori strict convex cone of  effective curves.

However the general story is subtler: in order to apply the swampland criterion, we should look at the covering  Legendre manifold $\widetilde{L}$ \emph{globally,} and this requires analytic continuation of $\cf$ outside the small domain $U$ in which the sum \eqref{whichformFFY} makes sense.

To make our point sharp, let us consider the example discussed in \S.\ref{functeqns} of a quantum consistent $\cn=2$ supergravity with $\Gamma=Sp(2m+2,\Z)$ which, in view of eqn.\eqref{deetrq}, implies $f=1$. In particular, $\Gamma$ contains the element
\be
\mathrm{diag}(+1,-1,\cdots,-1,+1,-1,\cdots, -1)\in Sp(2m+2,\Z),
\ee
which acts as $(X^0,X^i)\to (X^0,-X^i)$ that is $z^i\equiv X^i/X^0\to -z^i$. 
This  operation formally ``invertes the sign'' of the  instanton charges. 

Of course, the region
$\mathrm{Im}\, z^i\to -\infty$ is outside the domain of validity of the representation \eqref{whichformFFY} for the function $\cf$. We may however, use the functional equation of the function $\boldsymbol{L}(z)$
\be
\boldsymbol{L}(-z)=\boldsymbol{L}(z)+\frac{z^3}{3!}-\frac{z}{24},
\ee
to formally flip the sign of the topological charge for some primitive instanton from positive to negative. 
The price we pay is 
a redefinition of the polynomial part of $\cf$
which is perfectly compatible with the required  integrality properties of its coefficients.

In view of this situation, it looks natural to define the set 
$N\mspace{-4mu}E$ of non-empty instanton-charge sectors to be the {global} one, containing the instanton charges which virtually contribute in all of the asymptotic regions of $\widetilde{L}$ and in all determinations of the multivalued $\cf$.
In particular $N\mspace{-4mu}E$ should contain
instantons contributing to all asymptotic expansion
\be
\left\{ \vec n\in \Z^m\;\colon\;
\text{\begin{minipage}{190pt}the Fourier expansion of $\cf$ in \textbf{\emph{some}}
asymptotic region at $\infty$ of $\widetilde{S}$ contains $e^{2\pi i \vec n\cdot\vec w}$ with a non-zero coefficient\end{minipage}}\;\right\}\subset N\mspace{-4mu}E \subseteq \Z^m.
\ee

Let $\cn_\Gamma(J(\Z))$ be the normalizer in $\Gamma$ of the unipotent subgroup $J(\Z)$, and set 
\be
\Upsilon \overset{\rm def}{=} \cn_\Gamma(J(\Z))/J(\Z).
\ee
Note that $\Upsilon$ has the same formal structure as the would-be global symmetry group
\eqref{glosymgg}. $\Upsilon$ acts linearly on the $\vec z$, hence on the instanton-charges; morally speaking it plays the same role for the instanton-charges that $G_\text{glob}$ plays for the electro-magnetic ones. 
Therefore, \emph{to the very least} $N\mspace{-4mu}E$ should contain
\be\label{nnnwerq}
N\mspace{-4mu}E\supset \bigcup_{d=1}^\infty \bigcup_{\upsilon\in\Upsilon} d\,N\mspace{-4mu}Z_{\upsilon(U)},
\ee
where $d\,N\mspace{-4mu}Z_U$ correspond to $d$-fold cover instantons contributing to the local expansion in the asymptotic region $U$.
The physical expectation \cite{vvvaf} is that $N\mspace{-4mu}E$ is essentially the full lattice $\Z^m$. To establish this fact is our next task.

\subsection{Necessity of non-perturbative corrections}\label{jasqw12}

We note that when  there are no ``instanton corrections'' in eqn.\eqref{whichformFFY}, that is, when $N_{\vec n}=0$ for all $\vec n\neq0$,
the real shifts
\be
X^i\to X^i+y^i\, X^0,\quad y^i\in\R^m
\ee 
form a dimension-$m$, commutative, unipotent, real Lie subgroup $J(\R)\subset \mathsf{Sym}(\widetilde{S})$ of symmetries of the special geometry.
In particular, its Killing vectors
\be
X^0\partial_{X^i}\in \mathfrak{j}^\R\subset \mathfrak{sym}(\widetilde{S})\qquad i=1,2,\dots,m
\ee 
are non-zero. From \textbf{Fact \ref{ddiiye}} we conclude that a pre-potential of the form \eqref{whichformFFY} with $N_{\vec n}=0$ ($\vec n\neq0$)
 belongs to the swampland \emph{unless} its covering  space $\widetilde{S}$ is
 a Hermitian symmetric manifold, hence
 the tube domain $T(V)$ of a rank-3 symmetric convex cone $V$ ($\equiv$ the positive cone of a rank-3 Euclidean Jordan algebra), see the right-hand side of  \textsc{Table}  \ref{tableiso}.
 In this exceptional case $\cf$ is uniquely determined by  the corresponding Jordan algebra, modulo the trivial ambiguity
 due to possible different choices of integral-symplectic frames for the homogeneous coordinates of $\mathscr{D}_{m}^{(3)}$. For $T(V)$ symmetric we may choose the frame so that $6!\mspace{2mu}F(z^i)$ is a homogenous cubic form over $\Z$ equivalent over $\R$ to  the determinant in the Jordan algebra whose explicit expression is given in eqns.\eqref{det1} and \eqref{det2}.
 
 In particular $N_{\vec n}=0$ for all $\vec n\neq0$  implies that the constant term in \eqref{whichformFFY}  vanishes as well,
 that is, the ``Euler characteristic'' $e$ should be zero 
 \be\label{ezero}
\text{$N_{\vec n}=0$ for all $\vec n\neq0$}\quad\Longrightarrow\quad e=0.
 \ee
 Naively, $c_i=0$ also, but the coefficients $b_{ij}$,  $c_i$ are not uniquely defined since they depend on the choice of the duality frame. Indeed the change of the duality frame
 \be
 \begin{cases}\cf_I\to \cf_I+n_i(\delta_I^i\,X^0+\delta_I^0 X^i)\\
 X^I\to X^I\end{cases}\quad\Rightarrow\quad \cf\to\cf +n_i\, X^0X^i,\quad n_i\in \Z,
 \ee
 has the effect $c_i\to c_i+24\, n_i$ (while keeping fixed all other coefficients), so that,  in absence of instanton corrections one can assert \emph{at most} the weaker condition $c_i=0\bmod24$. The actual condition seems to be even weaker: 
 if the cohomology groups of $X$ have some $k$-torsion, 
 the actual congruence seems to be something like $k\,c_i=0\bmod24$. 
 \medskip
 
In facts, dicothomy yields more: if $\widetilde{S}$ is not symmetric, the quantum corrections should break the continuous unipotent symmetry
$J(\R)$ down to $\Lambda\subset J(\Z)\subset \Gamma$,
because, in this case, \emph{no} Killing vector
is allowed to survive quantization in a consistent theory. This condition  requires that the set $\vec n\in N\mspace{-4mu}Z_U$ of contributing instanton-charge vectors is big enough  
to generate the full lattice $\Z^m$.

 \begin{rem} Let us compare this result with the differential-geometric \textbf{Fact \ref{tubearg}} in  section \ref{dummies}. 
There we shown that a non-symmetric asymptotically-cubic special geometry -- which does \emph{not} belong to the swampland -- should receive \emph{some} quantum correction, perturbative or non-perturbative.
Here we see that we need \emph{non-perturbative} corrections, the perturbative ones being not enough to satisfy the swampland criterion. Moreover, the absence of non-perturbative correction in some exceptional cases \emph{implies} that the perturbative loop corrections vanish as well. We shall elaborate more on this below.
 \end{rem} 
 
 \subsection{Non-empty instanton-charge sectors}
 
 In the previous section we have seen that when $\widetilde{S}$ is not symmetric, the instanton-charges of the non-empty  sectors should generate the lattice $\Z^m$. However physically we expect \cite{vvvaf} a stronger condition, that is, that all instanton-charges $\vec n \in \Z^m$ correspond to a non-empty topological sector (in some asymptotic region of $S$).
 
 That infinitely many instanton-charges should be realized is already clear from the functional equations \eqref{kqawert1} in view of the fact that the $\mathbb{Q}$-Zariski closure of $\Gamma$ is semisimple, and hence the matrices \eqref{matgamma} of  some of its generators $\{\gamma_t\}$ necessarily have non zero left-bottom block, $C^{IJ}\neq0$,
 and the corresponding functional equation cannot be solved by any finite sum of the form 
 \eqref{whichformFFY}. Indeed, the structure of these equations strongly suggests that instantons of all charges will appear.

One may make the above heuristics a little more explicit (we take $f=1$ for simplicity).     
 Consider the subgroup $GL(m,\Z)\subset Sp(2m+2,\Z)$ given by the block-diagonal embedding
\be
GL(m,\Z)\ni \boldsymbol{A}\mapsto\text{Diag}\big(1,\boldsymbol{A},1,(\boldsymbol{A}^t)^{-1}\big)\in Sp(2m+2,\Z),
\ee
and let 
\be
\Xi\overset{\rm def}{=} \Gamma\cap GL(m,\Z)\subset Sp(2m+2,\Z).
\ee 
Clearly 
\be
\Xi\subseteq \cn_{\Gamma}(J(\Z))\quad\text{and}\quad \Xi\hookrightarrow \cn_\Gamma(J(\Z))/J(\Z)\equiv\Upsilon
\ee
then by eqn.\eqref{nnnwerq} 
\be\label{nnnwerq2}
 \bigcup_{d=1}^\infty \bigcup_{\xi\in\Xi} d\,N\mspace{-4mu}Z_{\xi(U)}\subset N\mspace{-4mu}E\subseteq \Z^m,
\ee
so to show that the set of non-empty instanton-charge sectors $N\mspace{-4mu}E$ is the full lattice $\Z^m$ it suffices to show that $\Xi$ is ``big enough''.
This clearly holds  when $\Gamma=Sp(2m+2,\Z)$
since in this case $\Xi = GL(n,\Z)\equiv\pm SL(n,\Z)$, by \S.\,\ref{jasqw12}. $N\mspace{-4mu}Z_{U}$ contains a \emph{primitive} element of the lattice $\Z^m$, while its $SL(n,\Z)$-orbit
consists of all primitive elements of 
$\Z^m$. Since $N\mspace{-4mu}E$
contains all multiples of its primitive elements, in this case $N\mspace{-4mu}E$ is the full lattice $\Z^m$.

The situation when $\Gamma\subset Sp(2m+2,\Z)$ is arithmetic reduces to the previous one.
We  have the commutative diagram (first row exact, $i$ and $\iota$ mono)
\be
\begin{gathered}
\xymatrix{1\ar[r]& \Gamma\ar[r]& Sp(2m+2,\Z)\ar[r]^{\phantom{mmm}\sigma}& \textsf{Finite}\ar[r] & 1\\
1 \ar[r]& \Xi\ar[u]^i \ar[r]& GL(m,\Z)\ar[u]^\iota \ar@/_1.5pc/[ru]_{\sigma\,\iota}}
\end{gathered}
\ee
and since $\Xi\equiv\mathrm{ker}\, \sigma\mspace{1mu}\iota$ we see that $\Xi$ is still $GL(m,\Z)$ modulo a finite group.

The case of a thin monodromy, is less clear, although 
the same argument will typically lead to a large set
$N\mspace{-4mu}E$ because 
in many respects thin subgroups of $Sp(2m+2,\Z)$ ``behave as they were the full group''. For instance:\footnote{\ There exists a more precise result, namely 
the \emph{super-}strong approximation of monodromy groups \cite{superstrong}.}

\begin{proper}[Strong approximation for  monodromy  \cite{MVW}\footnote{\ See also \textbf{Theorem A} in  \cite{strongapprox}}] $\Gamma\subset GL(n,\Z)$ a monodromy group consistent with the VHS structure theorem, that is, $\Gamma$ is finitely-generated and Zarinski-dense in $P(\mathbb{Q})\subset GL(n,\mathbb{Q})$, with $P$ a Zariski-connected, simply-connected, semi-simple $\mathbb{Q}$-group. Then
for \emph{almost all} prime $p$ the residue map
\be
\Gamma\xrightarrow{\ \pi_p\ } P(\Z/p\Z)\quad\text{is onto.}
\ee
\end{proper}

\subsection{Applications to Type IIA compactifications}
We apply the previous analysis to Type IIA compactifications:
 
 \begin{corl} $X$ a 3-CY with a mirror $X^\vee$. If the large-volume limit of the pre-potential for Type IIA on $X$ has no instanton correction, then
 \be
 c_3(X)=0,\quad c_2(X)=0\bmod 24,
 \ee
 and the quantum K\"ahler space $S$ is an arithmetic quotient of a Hermitian symmetric manifold of the ``magic'' type (in particular, $S$ is a Shimura variety).  
 \end{corl}
 
 We recall that a 3-CY has $c_2=0$ if and only if its Ricci flat K\"ahler metric is flat \cite{koba}.

\begin{exe} We check the \textbf{Corollary} in the two classes of known examples in which the pre-potential does not receives instanton corrections, that is, when \textit{(i)} $X$ is  a complex 3-torus, \textit{(ii)} $X$ is an elliptic curve times a K3, or \textit{(iii)} $X$ is a 
finite free quotient of \textit{(i)} or \textit{(ii)} (see next section for more details).
The 
pre-potential for all Calabi-Yau of type \textit{(iii)} was computed (in some convenient frame) in ref.\!\!\cite{type-K2} getting a purely cubic polynomial, in agreement with the prediction from \textbf{Fact \ref{ddiiye}} that there should exist an electro-magnetic frame with this property. Indeed, in these cases it is fairly obvious that the Weil-Petersson metric should be locally symmetric to
\be\label{pqwertt}
SL(2,\R)/U(1)\times SO(2,\rho-1)/[SO(2)\times SO(\rho-1)]
\ee 
($\rho\geq2$ being the Picard number)
and all such locally symmetric special geometry have a purely cubic pre-potential in a suitable frame.   
Note that the space \eqref{pqwertt} is a power $\ch^\rho$ of the upper half-plane $\ch$ if and only if $\rho=2$ or $3$. This observation is related to our next topic, the the Oguiso-Sakurai question.
In the next section we shall be mathematically more precise.
\end{exe}

\section{Answering the Oguiso-Sakurai question}
\label{sec:os}

In the rest of this note we adopt the notion of ``Calabi-Yau 3-fold'' which is natural in Algebraic Geometry.

\begin{defn} A Calabi-Yau 3-fold  (3-CY) $X$ is a compact K\"ahler manifold $X$ of dimension 3 with
\be
\ck_X\cong \co_X,\qquad H^1(X,\co_X)=0,
\ee
where $\ck_X$ is the canonical sheaf and $\co_X$ the structure sheaf. All 3-CY (in this sense) are smooth projective algebraic varieties over $\C$.
\end{defn}

Given this definition, there are two possibilities: either $\pi_1(X)$ is finite or it is infinite. When $|\pi_1(X)|<\infty$ the holonomy Lie algebra $\mathfrak{hol}(X)$ of the Ricci-flat metric is exactly
$\mathfrak{su}(3)$ \cite{beauv2}, and $X$ is a 3-CY in the \emph{strict} sense. 
\vglue8pt

A 3-CY with $|\pi_1(X)|=\infty$  has a finite unbranched cover which is either an Abelian variety $A$ or the product of an elliptic curve $E$ and a K3 surface $K$ \cite{beauv2}. 
A 3-CY of the form $A/\Sigma$ is said to be of \emph{A-type}, while one of the form $(E\times K)/\Sigma$ is called of \emph{K-type} \cite{OSaku,type-K1,type-K2}. In both cases $\Sigma$ is a finite group of automorphisms  acting freely.  
\vglue8pt

There are six deformation types of A-type 3-CY explicitly constructed in refs.\!\!\cite{OSaku,type-K1}. Their Ricci-flat K\"ahler metric is flat, so $\chi(X)=0$ and $c_2(X)=0$.  Indeed, a CY has $c_2=0$ if and only if its Ricci-flat metric is flat \cite{koba,OSaku,type-K1} so 
\be
c_2(X)=0\quad\Leftrightarrow\quad \text{$X$ is A-type}.
\ee
The Picard number $\rho\equiv h^{1,1}\equiv h^{2,1}$ of an A-type 3-CY is either 2 or $3$
\cite{OSaku}.
An A-type 3-CY $X$ contains no rational curve for obvious reasons.\footnote{\ \label{footfoot}In facts, let $\pi\colon A\to A/\Sigma\equiv X$ be the unbranched cover, and suppose (by absurd) that $C\subset X$ is a rational curve. Since $\pi$ is unbranched, $\pi^{-1}(C)$ is the disjoint union of connected curves $C_i$ in $A$,  each of which is an unbranched cover of $C\cong\mathbb{P}^1$. But any connected unbranched cover of $\mathbb{P}^1$ has genus zero by Gauss-Bonet. Hence $C_i$ is a rational curve in $A$, but this contradicts the fact that an Abelian variety does not contain any rational curve.}

There are eight deformation types of K-type 3-CY again explicitly
constructed in refs.\!\!\cite{OSaku,type-K1}. 
A K-type 3-CY $X$ has $\chi(X)=0$  in agreement with eqn.\eqref{ezero}.
Their possible Picard numbers 
 are \cite{OSaku,type-K1}
\be
\rho\equiv h^{1,1}=h^{2,1}=\text{$11,7,5,4$ or $3$.}
\ee
 A K-type 3-CY contains several rational curves, but all of them appear in one-parameter families parametrized by the curve $E$ (same argument as in footnote \ref{footfoot}), so the corresponding instanton corrections to the pre-potential $\cf$ vanish because of too many fermionic zero-modes; this fact may be stated  more intrinsically: one lifts the 2d \textsc{susy} $\sigma$-model to the finite-cover target $E\times K$ and uses the (4,4) non-renormalization theorems to show the absence of corrections.
\vglue8pt

In ref.\!\!\cite{OSaku} Oguiso and Sakurai raise the following

\begin{queos}
Is it true that all 3-CY without rational curves have Picard number $\rho=2$ or $3$?
\end{queos}

A stronger statement would be that a 3-CY $X$
has no rational curve if and only if it is A-type.
An even stronger result will be a positive answer to the following

\begin{que}\label{uuuquest} Is it true that the ``instanton corrections''
in eqn.\eqref{whichformFFY} for Type IIA on $X$ vanish if and only if $\pi_1(X)$ is infinite? \emph{That is, (in physical language):} is it true that the ``instanton corrections'' vanish if and only if they are forbidden by the non-renormalization theorem of a \textbf{higher} $(p,q)>(2,2)$
world-sheet supersymmetry?
\end{que}

In this section we answer \textbf{Question \ref{uuuquest}} in the positive assuming that $X$ has a mirror $X^\vee$ and Picard number $\rho\geq2$. The case of $\rho=1$ will be settled in the next section below using less rigorous physical arguments.
The results of the present section are mathematically rigorous (under the stated assumptions); indeed they follow directly from the VHS structure theorem \cite{reva,revb,MT4,periods} applied to the universal deformation of the mirror Calabi-Yau $X^\vee$ assumed to exist.
From a physical perspective, the hypothesis that $X$ has a mirror $X^\vee$ \textit{should be dropped:} indeed the only thing we use in the argument below is that the period map of the VHS which describes the quantum K\"ahler moduli satisfies the VHS structure theorem, 
as required by our swampland criterion.   

The positive answer to \textbf{Question \ref{uuuquest}} implies
\be
\text{\begin{minipage}{80pt}\centering no instanton corrections\end{minipage}}\ \Longleftrightarrow\ 
\text{\begin{minipage}{110pt}\centering higher 2d  \textsc{susy}\\ non-renormalization\end{minipage}}\ \Longrightarrow\ \text{\begin{minipage}{100pt}\centering no perturbative loop corrections\end{minipage}}
\ee
so that the absence of non-perturbative world-sheet correction also entails the absence of the perturbative ones, as we found in the previous section by the dichotomy argument. In particular, absence of rational curves implies $\chi(X)=0$.

\subsection{The mathematical argument}
 We consider the projective cubic hypersurface 
\be
W\subset P H^{2}(X),\qquad \rho\equiv \dim H^2(X),
\ee 
whose affine cone is
\begin{equation}\label{affcone}
C(W)=\Big\{D\in H^{2}(X)\;\Big|\; D^3=0\Big\}\subset H^{2}(X).
\end{equation}
Since the intersection numbers are integers, $C(W)$ (resp.\! $W$) is an  affine (resp.\! projective) algebraic hypersurface \emph{defined over $\mathbb{Q}$.} 
We recall the 
\begin{pro}[P.M.H. Wilson \cite{wilson1,wilson2}]\label{proWW} If $X$ has finite fundamental group and the projective hypersurface $W$ satisfies the condition that its rational points $W(\mathbb{Q})$ are dense in its real locus $W(\R)$, then there exists a rational curve on $X$. 
\end{pro}

\begin{rem} Note that the \textbf{Proposition} says nothing when $\rho\equiv \dim H^2(X)=1$, since in this case $W$ is not defined. 
\end{rem}

If $W$ contains a hyperplane it should be rational\footnote{\ For this claim and the list of properties in items (a),(b),(c) see  comment after \textbf{Lemma 4.2} in  ref.\!\cite{wilson2}.}, that is, if $W$ is absolutely reducible it is already reducible over $\mathbb{Q}$.
Then we remain with three possibilities:
\begin{itemize} 
\item[(a)] $W$ is an irreducible cubic defined over $\mathbb{Q}$;
\item[(b)] $W$ contains a hyperplane
and an irreducible quadric, both defined over $\mathbb{Q}$;
\item[(c)] $W$ consists of 3 hyperplanes defined over $\mathbb{Q}$, possibly counted with multiplicity.
This cannot happen for  $\rho>3$ \cite{wilson2}.
\end{itemize} 

The rational points are trivially dense along the irreducible components which are hyperplanes. Using general facts about quadratic and cubic integral forms representing zero, one shows:
\begin{lem}[P.M.H. Wilson \cite{wilson2}]\label{jaswelem} The rational points are dense in the real locus:
\begin{itemize} 
\item[\bf(1)] In case {\rm(b)} when $\rho>5$ (Meyer's theorem);
\item[\bf(2)] If $W$ is irreducible, $\rho>5$,
and $W$ contains a rational linear space of dimension 3;
\item[\bf(3)] If $W$ is irreducible and $\rho>19$.
\end{itemize}
\end{lem}

The swampland structural  criterion sets severe constraints on the possible integral cubic forms $D^3$
which may arise from the triple intersection of divisors in a 3-CY without rational curves.
Then, in view of \textbf{Proposition \ref{proWW}}, to answer the  Oguiso--Sakurai question we need to study the rational points of only these very \emph{specific}  cubic hypersurfaces.
\medskip

Below we show the following:

\begin{cla}\label{pq123} $X$ a 3-CY (with mirror and $\rho\geq2$) without rational curves or, more generally, such that its A-model TFT has no instanton corrections $\Rightarrow$ the rational points $W(\mathbb{Q})$ are dense in the real locus $W(\R)$.
\end{cla}

Together with \textbf{Proposition \ref{proWW}} this \textbf{Claim} gives

\begin{corl} $X$ a 3-CY (with mirror and $\rho\geq2$) such that its A-model TFT has no instanton correction. Then
$|\pi_1(X)|=\infty$.
\end{corl}

\noindent From this \textbf{Corollary} all other claims follow provided the Picard number $\neq1$.

\subparagraph{Proof of Claim.} As we saw sections 6 and 7, the structure theorem of VHS (applied to the complex moduli of the mirror $X^\vee$) implies that, in the absence of instanton corrections, the K\"ahler moduli of $X$ do not get any quantum correction, and that  their special K\"ahler geometry $S$ (after replacing it with a finite cover, if necessary) 
is 
the product of irreducible Shimura varieties
in one-to-one correspondence with the irreducible components $C_\alpha$ of the affine cone $C(W)$, eqn.\eqref{affcone}
\be
S=\prod_\alpha S_\alpha,\qquad S_\alpha\ \text{Shimura variety of appropriate type.}
\ee  By \cite{wilson2}
the $C_\alpha$'s 
 are defined over $\mathbb{Q}$. Each irreducible Shimura variety $S_\alpha$ is the quotient of a Hermitian symmetric space $M_\alpha(\R)/K_\alpha$ by the arithmetic subgroup $\Gamma_\alpha\subset M_\alpha(\Z)$, where $M_\alpha$ is the \emph{simple} $\mathbb{Q}$-algebraic group 
\be
M_\alpha\equiv \overline{\Gamma}_\alpha^{\mspace{2mu}\mathbb{Q}}\subset Sp(2m+2,\mathbb{Q})\quad\text{with}\quad
\Gamma_\alpha\subset M_\alpha\cap Sp(2m+2,\Z).
\ee
\be
\mathrm{rank}\,\Gamma_\alpha=\mathbb{Q}\text{-rank}\,M_\alpha=\mathbb{R}\text{-rank}\,M_\alpha(\R),
\label{rrrianks}\ee
where, as before, $\overline{\Gamma}_\alpha^{\mspace{2mu}\mathbb{Q}}$ stands for the $\mathbb{Q}$-Zariski closure of $\Gamma_\alpha\subset Sp(2m+2,\mathbb{Q})$ in the ambient $\mathbb{Q}$-algebraic group. 
The Lie group $M_\alpha(\R)$ of real points is
\be
\begin{aligned}
&SL(2,\R) &\quad &\text{for hyperplanes}\\
&SO(2,\rho-1) &\quad &\text{for quadrics}\\
&Sp(6,\R),
\ U(3,3),\ SO^*\mspace{-1mu}(12),\ E_{7(-25)} &\quad & \text{for irreducible  cubics}.
\end{aligned}
\ee 
The density of the rational points in the real locus of the irreducible component $C_\alpha$ is obvious when $C_\alpha$ is an hyperplane. From now on we shall consider only quadric and cubic components.
For quadrics the only open case is $\rho=4$, since for $\rho=2,3$ the special  geometry becomes reducible (so we are back to the case of hyperplanes) while when $\rho\geq 5$ it is settled by Meyer's theorem, see \textbf{Lemma \ref{jaswelem}} \textbf{(1)}. Part \textbf{(3)} of the \textbf{Lemma} implies the \textbf{Claim} for  the last cubic case which has $\rho=27>19$.

From sections 3, 4, 6 and 7 we know that each irreducible Hermitian symmetric space
$M_\alpha(\R)/K_\alpha$ is (biholomorphic to) a tube domain $T(V_\alpha)$ for a \emph{symmetric}, open, strict, convex cone $V_\alpha\subset \R^m$ where $m=\rho-1$ or  $\rho$ for quadrics or, respectively, cubics. The open cone $V_\alpha$ is a connected component of the open domain in $\R^m$
\be\label{poqwertb}
d_{ij}\,x^ix^j>0\quad\text{
respectively}\quad  d_{ijk}\,x^ix^jx^k>0
\ee
with $d_{ij}$, resp.\! $d_{ijk}$ integers.
The integer quadratic form $d_{ij}$ has Lorentzian signature $(1,m-1)$ so, in the quadric case, the convex cone $V_\text{quad.}$ is the set of ``forward time-like'' vectors, whereas in the cubic case
$V_\text{cub.}$ is the cone of positive-definite elements in the $\R$-space $\mathsf{Her}_3(\mathbb{F})$ of $3\times 3$ Hermitian matrices  with entries in $\mathbb{F}=\R,\C,\mathbb{H}$, and $\mathbb{O}$, respectively, see \S.\,\ref{3tube}. The classical K\"ahler cone -- which coincides with the quantum one under the present hypothesis -- is the product of the convex cones $V_\alpha$ associated to the various irreducible components $C_\alpha$.

The equation of the associated irreducible component $C_\alpha$ of the hypersurface $C(W)\subset \R^m$ is obtained by replacing in eqn.\eqref{poqwertb} $>0$ with $=0$. Thus $C_\text{quad.}(\R)$ is the full  light-cone, whereas $C_\text{cub.}(\R)$ is identified with  the space of Hermitian matrices in $\mathsf{Her}_3(\mathbb{F})$ of rank $\leq2$. 

Let $G_\alpha$ be the
 $\mathbb{Q}$-algebraic group
\be
G _\alpha=M_\alpha\cap GL(m,\mathbb{Q})
\ee
where both groups in the \textsc{rhs} are seen as subgroups of the $\mathbb{Q}$-algebraic group $Sp(2m+2,\mathbb{Q})$
through the block-diagonal embedding
\be
g\mapsto \mathrm{diag}(1,  (g^{-1})^t, 1, g)\in Sp(2m+2,\mathbb{Q}),\quad g\in GL(m,\mathbb{Q}).
\ee 
Modulo finite groups, one has
\be\label{pqwertv}
G_\alpha=\mathbb{G}_m\times L_\alpha,
\ee 
where $\mathbb{G}_m$ is the multiplicative group
and $L_\alpha=M_\alpha\cap SL(m,\mathbb{Q})$ is the simple $\mathbb{Q}$-algebraic group  
which leaves invariants the rational tensors
$d_{ij}$ resp.\! $d_{ijk}$
\be
\begin{split}
L_\text{quad.}=\;&\big\{{g^i}_j\in SL(m)\;:\; d_{ij}\,{g^i}_k\,{g^j}_l=d_{kl}\big\}\\ 
L_\text{cub.}=\;  &\big\{{g^i}_j\in SL(m)\;:\; d_{ijk}\,{g^i}_l\,{g^j}_m\,{g^k}_n=d_{lmn}\big\}.
\end{split}
\ee
The connected group $G_\alpha(\R)^\circ$ of real points of $G_\alpha$
is nothing else than the automorphism group of the cone $V_\alpha\subset \R^m$. Then
\begin{equation}\label{jasqwerpp}
\begin{aligned}
G_\text{quad.}(\R)=\;&\R^\times \times SO(1,m-1),\\
G_\text{cub.}(\R)=\;&\R^\times\times SL(3,\mathbb{F}),
\end{aligned}
\end{equation}
where $SL(3,\mathbb{O})\equiv E_{6(-26)}$. 

\begin{rem} The $\mathbb{Q}$-algebraic group $\prod_\alpha L_\alpha$ has a simple physical interpretation: its real points $\prod_\alpha L_\alpha(\R)$ form  the \emph{naive} symmetry group of the 5d $\cn=1$ \textsc{sugra} obtained by compactifying M-theory on the Calabi-Yau $X$. $\prod_\alpha L_\alpha(\R)$ is the isometry group of the universal cover of the 5d vector-multiplet scalars' space  \cite{magicsquare}, and its $\mathbb{Q}$-algebraic structure is induced by the 5d flux quantization. In particular, the 5d swampland condition yields
\be\label{pqwertxx}
\mathbb{Q}\text{-rank}\, L_\alpha=\R\text{-rank}\,L_\alpha(\R)=\begin{cases} 1 &\text{quadrics}\\
2 & \text{cubics,}\end{cases}
\ee 
a fact already implied by \eqref{rrrianks}.  
\end{rem}
\medskip

For each irreducible component $C_\alpha$ of $C(W)$ we may find a finite set of real-valued  points
$\{x_a\}\in C_\alpha(\R)$ such that the union of the (analytic) closure of their orbits under $G_\alpha(\R)$
is the full real locus $C_\alpha(\R)$ of the  component
\be\label{pqw12}
C_\alpha(\R)=\bigcup_a\mspace{2mu} \overline{G_\alpha\mspace{-1mu}(\R)\mspace{-3mu}\cdot\mspace{-2mu} x_a}\mspace{4mu}.
\ee
Indeed, in the quadric case we may take as $\{x_a\}$ just a single non-zero point $x$ lying in the light-cone, while in the cubic case we may take a pair of rank 2 Hermitian matrices of different signature, say $x_\pm= \mathrm{diag}(1,\pm1,0)\in\mathsf{Her}_3(\mathbb{F})$.

Suppose for the moment that it is possible to choose the points $\{x_a\}$ in \eqref{pqw12} to be $\mathbb{Q}$-valued. In this case
\be\label{ratorbit}
\bigcup_a\mspace{2mu} G_\alpha\mspace{-1mu}(\mathbb{Q})\mspace{-3mu}\cdot\mspace{-2mu} x_a\subset C_\alpha(\mathbb{Q}).
\ee
We claim that the union of the rational orbits in the \textsc{lhs} of \eqref{ratorbit} is dense in the real locus $C_\alpha(\R)$. This follows from eqn.\eqref{pqw12} together with the fact that $G_\alpha\mspace{-1mu}(\mathbb{Q})$ is dense in $G_\alpha\mspace{-1mu}(\R)$.
 Indeed, $G_\alpha=\mathbb{G}_m\times L_\alpha$. Density of the rational points is obvious for the multiplicative group $\mathbb{G}_m$. For the simple $\mathbb{Q}$-algebraic group $L_\alpha$ we invoke 
 the following fact\footnote{\ From  \textbf{Theorem}  (4.6.3) of \cite{morris2}
we know that all connected, semi-simple subgroup
$L\subset SL(n, \R)$ is almost $\R$-Zariski closed
in $SL(n,\R)$. Then apply
\textbf{Proposition} (5.1.8) of \cite{morris}:  let $L$ be a connected subgroup of $SL(n,\R)$ which is almost $\R$-Zariski closed; then the group $L$ is defined over $\mathbb{Q}$ if and only if $$L_\mathbb{Q}\overset{\rm def}{=} L\cap SL(n,\mathbb{Q})$$ is dense in $L$.}
\vskip6pt

\begin{lem}
Let $L_\R\subset SL(m,\R)$ be a semi-simple and connected (closed) Lie subgroup.
Some finite cover of $L_\R$ is the Lie group of real points of a $\mathbb{Q}$-algebraic group if and only if
$L_\R\cap SL(m,\mathbb{Q})$ is dense in $L_\R$.
\end{lem} 

Putting everything together, the density of the rational points $C(W)(\mathbb{Q})$ in the real locus $C(W)(\R)$ will follow if we can show that we may choose the points $\{x_a\}$ in \eqref{pqw12} to be $\mathbb{Q}$-valued.
\medskip

For quadrics it is enough to show that the light-cone $d_{ij}x^ix^j=0$ contains a non-zero rational point. In view of eqn.\eqref{pqwertxx}, this is a direct consequence of, say, \textbf{Proposition} (5.3.4)
of \cite{morris}.  In other words: since the quotient
$\Gamma_\text{quad.}\G_\text{quad.}(\R)/K_\text{quad}$ is non-compact of finite volume,
and $\Gamma_\text{quad.}$ is arithmetic, the
$\mathbb{Q}$-rank of $\G_\text{quad.}(\R)$ is at least one by the Godement criterion; for a Lorentz group this is equivalent to the existence of a non-zero rational point in the light-cone \cite{morris}.

More generally, in all instances -- quadrics as well as cubics -- the existence of the rational points $\{x_a\}$ will follow if we can show that the rational structure of the algebraic group $L_\alpha$ is the ``obvious'' one: for instance,
for the first cubic case ($\rho=6$), where  $L_\text{cub.}\mspace{-1mu}(\R)=SL(3,\R)$,
it would follow if $L_\text{cub.}\mspace{-1mu}(\mathbb{Q})$ is the plain $\mathbb{Q}$-group $SL(3,\mathbb{Q})$ rather than some fancier rational structure on the real group $SL(3,\R)$.
Indeed, if this was the case, the rational cubic form $d\colon \mathsf{Her}_3(\mathbb{Q})\to \mathbb{Q}$ would be given by the plain 
determinant formula\footnote{\ The equality \eqref{simplydet} must hold up to $\mathbb{Q}$-equivalence, the expression needs not to  be valid over $\Z$.}
\be\label{simplydet}
\frac{1}{6!}\,d_{ijk}\mspace{2mu}x^ix^jx^k=\det\mspace{-5mu}\begin{pmatrix} x^1 & x^4 & x^5\\
x^4 & x^2 & x^6\\
x^5 & x^6 & x^3
\end{pmatrix}\in \Z[x^1,\dots, x^6],\qquad (x^1,\dots, x^6)\in\mathbb{Q}^6.
\ee
and the ``obvious'' points \be
x_\pm=\mathrm{diag}(1,\pm1,0)\in C_\text{cub.}(\mathbb{Q})\subset \mathsf{Her}_3(\mathbb{Q}),
\ee 
would be defined over $\mathbb{Q}$, and we would be done. 

That the ``obvious'' $\mathbb{Q}$-structure is the correct one for this first cubic example, follows from classification of the $\mathbb{Q}$-algebraic groups whose real locus is the Lie group $SL(3,\R)$. We recall the relevant facts, see e.g.\! the proof of \textbf{Proposition} (18.6.4) in \cite{morris}. There are infinitely many non-isomorphic such rational groups, but only the ``obvious'' one $SL(3,\mathbb{Q})$ has $\mathbb{Q}$-rank equal 2 as required by \eqref{rrrianks}. A more direct argument which does not use eqn.\eqref{rrrianks} is as follows: all rational algebraic groups with underlying real Lie group $SL(3,\R)$ have a natural rational representation on $\mathbb{Q}^6$. Tensoring this representation with $\R$ we get the irreducible representation $\odot^2\R^3$ for rank 2 and the reducible one $\R^3\oplus (\R^3)^\vee$ for rank 1 (see the proof of \textbf{Proposition} (6.6.1) in \cite{morris}). Since we know that the underlying real representation is $\odot^2\R^3$ (\S.\,3.3) we conclude that $SL(3,\mathbb{Q})$ is the rational structure on $SL(3,\R)$ selected by quantum gravity.
So the cubic $\rho=6$ instance is settled. 

In  the second cubic case ($\rho=9$), where $L_\text{cub.}\mspace{-1mu}(\R)=SL(3,\C)$, the same argument leads us to 
\be
L_\text{cub.}\mspace{-1mu}(\mathbb{Q})= SL(3,\mathbb{Q}(\sqrt{-r}))\ee 
for some square-free positive integer $r$
whose precise value is not relevant for our present purposes, although it should be seen as an important physical invariant of the corresponding 5d quantum gravity -- \emph{if it exists}.
Again the points $x_\pm=\mathrm{diag}(1,\pm1,0)\in \mathsf{Her}_3(\mathbb{Q}(\sqrt{-r}))$ are defined over $\mathbb{Q}$ independently of the value of $r$. Alternatively
the rank 2 matrices
\be
\begin{pmatrix}x_1 & x_2+\sqrt{-r}\,x_3& 0\\
x_2-\sqrt{-r}\,x_3 & x_4 & 0\\
0 & 0 & 0 \end{pmatrix}\in \mathsf{Her}_3(\mathbb{Q}(\sqrt{-r}),
\ee
where $(x_1,x_2,x_3,x_4)\in\mathbb{Q}^4$ form a dimension 3 linear space in $W(\mathbb{Q})$
and we may apply \textbf{Lemma 5(2)}.

On the basis of these two examples, let us 
consider the general cubic case.
We have
\be
L_\text{cub.}\mspace{-1mu}(\mathbb{Q})=SL(3,\mathbb{D}_\mathbb{F})
\ee
where $\mathbb{D}_\mathbb{F}$ is a division algebra over $\mathbb{Q}$, with a unique (canonical) positive-definite
Rosati (anti)involution $x\leftrightarrow x^\star$,
such that 
\be
\mathbb{D}_\mathbb{F}\otimes_\mathbb{Q}\mspace{-2mu}\R=\mathbb{F}.
\ee
For $\mathbb{F}=\R$, $\C$, and $\mathbb{H}$, the division algebra $\mathbb{D}_\mathbb{F}$ has, respectively, Shimura type \cite{shi1963}
I, IV, and III over the totally real field $\mathbb{Q}$.  
The case $\mathbb{F}=\mathbb{O}$ is not covered by Shimura, since the octonions do not form an associative algebra. 
The argument formally extends also to this case; anyhow we do not need it, since $\mathbb{F}=\mathbb{O}$ corresponds to Picard number $\rho=27>19$ and this case is already settled by \textbf{Lemma 5(3)}.

The canonical involution
$x\leftrightarrow x^\star$ has the property
\cite{shi1963}
\be\label{ponnbv0}
x+x^\star\in \mathbb{Q},\qquad \forall\;x\in\mathbb{D}_\mathbb{F}
\ee
so that it makes sense to talk about the vector $\mathbb{Q}$-space $\mathsf{Her}_3(\mathbb{D}_\mathbb{F})$ of Hermitian matrices
with entries in $\mathbb{D}_\mathbb{F}$ as well as the positive rational cone 
\be
V_\mathbb{Q}\subset \mathsf{Her}_3(\mathbb{D}_\mathbb{F})\approx \mathbb{Q}^{3+3\dim \mathbb{F}}
\ee 
i.e.\! the would-be rational K\"ahler cone of the putative Calabi-Yau $X$ without rational curves.
By \eqref{ponnbv0} the entries along the main diagonal of any $m\in \mathsf{Her}_3(\mathbb{D}_\mathbb{F})$
are rational numbers. Again, our arguments are independent of the precise division $\mathbb{Q}$-algebra $\mathbb{D}_\mathbb{F}$, whose isomorphism class is an important datum of the consistent quantum gravity (if it exists). 
The $\mathbb{Q}$-algebraic group $SL(3,\mathbb{D}_\mathbb{F})$ acts on the  
$\mathbb{Q}$-space $\mathsf{Her}_3(\mathbb{D}_\mathbb{F})$ as
\be
m\mapsto a\mspace{2mu} m\mspace{2mu} a^\dagger, \quad m\in\mathsf{Her}_3(\mathbb{D}_\mathbb{F}),\ a\in SL(3,\mathbb{D}_\mathbb{F}), \quad a^\dagger\overset{\rm def}{=} (a^\star)^t.
\ee
The cubic form is just the determinant of the matrix $m$ (defined by the Vinberg prescription \eqref{det2} when $\mathbb{D}_\mathbb{F}$ is non-commutative/non-associative)
\be
\det \colon \mathsf{Her}_3(\mathbb{D}_\mathbb{F})\to \mathbb{Q}\qquad\text{(by \eqref{ponnbv0}).}
\ee 
The rational points of the cubic affine cone $C_\text{cub.}(W)(\mathbb{Q})$ are then  identified with the elements of $\mathsf{Her}_3(\mathbb{D}_\mathbb{F})$ of rank $\leq 2$. This set is clearly non-empty, e.g.\!
$\mathrm{diag}(1,\pm1,0)\in\mathsf{Her}_3(\mathbb{D}_\mathbb{F})$, and we are done.

\begin{rem} For $\rho=2,3$ we have 3-CY without rational curves i.e.\! the A-type ones, and in addition there are 3-CY with rational curves but no quantum corrections to the K\"ahler moduli, namely the K-type whose Picard numbers are $\rho=3,4,5,7,11$.
The existence of these special cases is consistent with our findings. In all these cases the special K\"ahler geometry is an arithmetic quotient of
\be
SL(2,\R)/U(1)\times SO(2,\rho-1)/[SO(2)\times SO(\rho-1)],
\ee  
and $|\pi_1(X)|=\infty$.
\end{rem}

\section{The case of Picard number 1}

The argument of the previous section cannot be applied to Calabi-Yau 3-folds with
\be
\rho\equiv h^{1,1}=1.
\ee
In this section we argue that -- as everybody expects \cite{wilson1} -- a 3-CY with $\rho=1$ has (infinitely many) rational curves. To attach this residual case, we use  physical ideas and the discussion is meant to be just heuristic.
\medskip

We assume (by absurd) that $X$ has no rational curve.
Then, by the previous result
\be
F(z)=-\frac{N}{3!}\,z^3,\qquad N\in\mathbb{N},
\ee
so that the covering special geometry $\widetilde{S}$ is the upper half-plane with  the Poincar\'e metric $G_{z\bar z}$ normalized to that 
\be
R_{z\bar z}=-\frac{2}{3}\,G_{z\bar z},
\ee
that is, 3 times the Poincar\'e metric in the standard normalization.
The Hodge metric is
\be\label{hodg12}
K_{z\bar z}\equiv 4\mspace{3mu} G_{z\bar z}+R_{z\bar z} =\frac{10}{3} G_{z\bar z},
\ee
and it does not receive the contact-term  quantum correction described in \cite{Bershadsky:1993ta} since $\chi(X)=0$ in absence of instanton corrections.

 The period map factors through the upper half-plane, so that
$\Gamma$ must be a finite-index subgroup
of $SL(2,\Z)$ containing a power of $T$. Then the global symmetry group is\footnote{\ We use the fact that a normal subgroup $N\triangleleft SL(2,\Z)$ which contains $T$ is the full $SL(2,\Z)$. Indeed, from $S^2=-1$ and $(TS)^3=-1$ we get
$S^{-1}\equiv T(S^{-1}TS)T\in N$ while $T$ and $S^{-1}$ generate $SL(2,\Z)$.}
\be
\mathcal{N}_{SL(2,\Z)}(\Gamma)/\Gamma=SL(2,\Z)/\Gamma,
\ee
and if we assume that it is trivial (as required by the usual swampland arguments) we conclude that
\be\label{hodg13}
\Gamma\equiv SL(2,\Z).
\ee
Using the equations of one-loop holomorphic anomaly \cite{Bershadsky:1993ta} and eqn.\eqref{hodg12}, one finds that the genus one
index $F_1$ is
\be
F_1(z,\bar z)=-\log\!\Big((\mathrm{Im}\,z)^{10}\, |f(z)|^2\Big),
\ee
for some holomorphic function $f(z)$ without zero or poles in the upper half-plane; the expression should be modular invariant by \eqref{hodg13}. Then, following the argument around eqn.(8) of ref.\!\cite{Bershadsky:1993ta}, we get
\be
f(z)=\eta(\tau)^{20}.
\ee
From \cite{Bershadsky:1993ta} we know
\be
-\frac{1}{12}\int_X \omega\wedge c_2=\lim_{\mathrm{Im}\,z\to\infty}\!\left(\frac{1}{2\pi i}\,\partial_z F_1\right)  =-\frac{10}{12}
\ee
which looks as a contradiction since
\be
\int_X\omega\wedge c_2=10\not\equiv 0\bmod 24.
\ee

\section*{Acknowledgments} 
I am very grateful to Cumrun Vafa for sharing his profound insights on the consistency conditions for Quantum Gravity. I have benefit of discussions with E. Palti, S. Katz and I. Zadeh.

\appendix
\section{Pluri-harmonic $\Rightarrow$
solution to $tt^*$ PDEs}\label{revtau}

\textbf{In this appendix:} $X$ is a K\"ahler manifold  which in the standard application of $tt^*$ geometry to 2d (2,2) QFTs  plays the role of the space of coupling constants associated to the chiral primaries \cite{tt*}; all arguments and quantities will be independent of the chosen K\"ahler metric on $X$. $G$ is a semi-simple real Lie group and $K\subset G$ is a maximal compact subgroup. 
$\mathfrak{g}=\mathfrak{k}\oplus \mathfrak{p}$ is the orthogonal decomposition\footnote{\ With respect to the Killing form.} of the Lie algebra of $G$ in the Lie algebra of $K$ plus the complementary space $\mathfrak{p}$. We write $g^{-1}dg\in \Lambda^1(G)\otimes\mathfrak{g}$ for the Maurier-Cartan form on the group manifold $G$,  and $(g^{-1}dg)_\mathfrak{k}$, 
$(g^{-1}dg)_\mathfrak{p}$ for its summands with respect to the decomposition $\mathfrak{g}=\mathfrak{k}\oplus \mathfrak{p}$.
$\Gamma\subset G$ is a discrete subgroup.
\medskip

We wish to show the following \cite{jap, dubrovin}:

\begin{lem} Let 
\be
f\colon X\to \Gamma\backslash G/K
\ee 
be a smooth map. Let $\widetilde{X}$ be the universal cover of $X$ and $\phi\colon \widetilde{X}\to G$ be a lift of $f$. Define the definite-type one-forms on $\widetilde{X}$
\be\label{AAASWE}
\begin{aligned}
A&= \phi^*\mspace{-1mu}(g^{-1}dg)_\mathfrak{k}|_{(1,0)},& \bar A&= \phi^*\mspace{-1mu}(g^{-1}dg)_\mathfrak{k}|_{(0,1)}, \\
 C&=\phi^*\mspace{-1mu}(g^{-1}dg)_\mathfrak{p}|_{(1,0)},& \bar C&=\phi^*\mspace{-1mu}(g^{-1}dg)_\mathfrak{p}|_{(0,1)},
\end{aligned}
\ee
whose coefficients are seen as square matrices acting on some representation space $V$ for $G$,
and consider the $K$-covariant Dolbeault differentials
\be
D=\partial+A,\qquad \bar D=\bar\partial+\bar A.
\ee
Then $D$, $\bar D$, $C$, $\bar C$ solve the $tt^*$ equations {\rm\cite{tt*}}
\be\label{PDESS}
\begin{gathered}
DC=\overline{D}C=D\overline{C}=\overline{D}\mspace{3mu}\overline{C}=C\wedge C=\overline{C}\wedge \overline{C}=0\\
D^2=\overline{D}^2=D\overline{D}+\overline{D}C+C\wedge\overline{C}+\overline{C}\wedge C=0.
\end{gathered}
\ee
\emph{if and only if} $f$ is pluri-harmonic, i.e.\!\! 
in local coordinates $\bar D_{\bar k}\partial_i f^a=0$. Conversely, all
solutions to $tt^*$ arise in this way.
\end{lem}

\begin{rem}
Two local lifts $\phi_1$ and $\phi_2$ differ by a $K$-gauge transformation, so define the same $tt^*$ geometry which depends only on the underlying pluri-harmonic map $f$.
\end{rem}

\begin{proof}
Consider the Maurier-Cartan identity 
\be
\phi^*\mspace{-1mu}(d+g^{-1}dg)^2=0
\ee 
and split it according to Lie algebra decomposition $\mathfrak{g}=\mathfrak{k}\oplus\mathfrak{p}$ and the form type: 
\be
\begin{aligned}
& D^2+C\wedge C= \bar D^2+\bar C\wedge \bar C=0,\\
&DC=\bar D\bar C=0\\
& D\bar D+\bar D D+C\wedge \bar C+ \bar C\wedge C=0\\
& D\bar C+\bar D C=0. 
\end{aligned}
\ee
These equations reduce to the $tt^*$ equations \eqref{PDESS} if we may prove the two equalities $\bar D C=0$ and $C\wedge C=0$. The first one is just the condition that $f$ is pluri-harmonic. The consistency of the condition $\bar D C=0$ and $D C=0$ yields
\begin{equation}\label{pqw12xz}
\begin{aligned}
0= (D\bar D+\bar D D)C&=(C\wedge \bar C+\bar C\wedge C) \wedge C- C\wedge (C\wedge \bar C+\bar C\wedge C)=\\
&=
\bar C\wedge (C\wedge C)-(C\wedge C)\wedge \bar C.
\end{aligned}
\end{equation}
Writing $C\wedge C\equiv \tfrac{1}{2} [C_i, C_j] dt^i \wedge dt^j$, this yields
\be
\big[[C_i,C_j],[C_k,C_l]^\dagger\big]=0\quad \text{for all }i,j,k,l.
\ee 
Hence the matrices $\mathcal{C}_{ij}\equiv [C_i,C_j]$ are normal and may all be simultaneously diagonalized. Let $V$ be a simultaneous eigenspace where $\mathcal{C}_{ij}$ acts by multiplication by $\lambda_{ij}$. By \eqref{pqw12xz} $V$ is left invariant by all $C_k$, $\bar C_k$, then $[C_i,C_j]|_V=0$ since the trace of a commutator vanishes. So all eigenvalues $\lambda_{ij}$ of $[C_i,C_j]$ are zero; since $[C_i,C_j]$ is a normal matrix it must vanish, and $C\wedge C=0$. 

Conversely, write the $tt^*$ equations as the integrability conditions of the linear system
\be
\big(D+\bar D+\zeta^{-1}C+\zeta \bar C\big)\Psi(\zeta)=0
\ee 
for all values of the spectral parameter $\zeta\in\mathbb{P}^1$.
From the reality condition \cite{tt*}, the fundamental solution $\Psi(\zeta)$ may be chosen to take value in $SL(n,\C)$ and $\Psi(1/\zeta^*)=\Psi(\zeta)^*$.
The $SL(n,\C)$-connection 
\be
D+\bar D+\zeta^{-1}C+\zeta \bar C
\ee 
is flat, so its pull-back on the universal cover $\widetilde{X}$ is pure gauge
\be
d+A+\bar A+\zeta^{-1}C+\zeta \bar C=d-\big(d\Psi(\zeta)\big)\Psi(\zeta)^{-1}=d+\Psi(\zeta)\,d\Psi(\zeta)^{-1}.
\ee
Specializing to $\zeta=1$ and taking $g=\Psi(1)^{-1}$, we get the identifications \eqref{AAASWE}.
Then the composition
\be
\widetilde{X}\xrightarrow{\ \Psi(1)^{-1}\ } SL(n,\R)\xrightarrow{\ \pi\ } SL(n,\R)/SO(n),
\ee
yields a pluri-harmonic map $\widetilde{f}$
with $G=SL(n,\R)$. Since the branes satisfy the 
equivariant condition
\be
\xi^\ast\Psi(1)= \Psi(1)\gamma_\xi,\quad \gamma_\xi\in\Gamma\quad\text{for all }\xi\in\mathsf{Deck}(\widetilde{X}\to X),
\ee
$\widetilde{f}$ descends to a pluri-harmonic map
$f\colon X\to \Gamma\backslash SL(n,\R)/SO(n)$.

\end{proof}

\begin{rem} 
Under suitable algebraic conditions on the chiral ring $\mathcal{R}$, i.e.\! the $\C$-algebra generated  by 1 and the matrices $C_i$, the image $\Psi(1)(X)$
lays in a proper subgroup $G(\R)\subset SL(n,\R)$ and we may write $\widetilde{f}\colon X\to G(\R)/K$ (with $K=G(\R)\cap SO(n)$), $\widetilde{f}$ pluri-harmonic. 

\end{rem}

\section{Arithmetics of {\textit{superconformal}} $tt^*$ branes}

The arithmetic properties of the $tt^*$ brane amplitudes of a 2d superconformal (2,2) are much richer and more beautiful than  those of a generic (2,2) QFT, in particular more elegant than the ones for a gapped (2,2) model where the BPS branes are
simple Lefschetz thimbles whose monodromy is governed by the classical  Picard-Lefschetz theory \cite{Cecotti:1992rm,Hori:2000ck}.
\medskip

To describe the \emph{arithmetic} characterization of the branes of a (2,2) SCFT inside the space of \emph{all} $tt^*$ branes
of (2,2) QFTs, we consider the brane amplitudes $\Psi(\zeta)$ along the twistor equator $|\zeta|=1$ and normalize the basis elements of the chiral ring $\mathcal{R}_s$ so that the determinant of the topological 2-point function $\eta_{ij}$ is 1 (such a normalized basis always exists). Then, in a \emph{generic} 2d (2,2) QFT, the brane amplitude $\Psi(\zeta)$ takes values in $SL(n,\R)$, where $n$ is the Witten index (see \S.\,3 of \cite{dubrovin}). More precisely,
for a fixed $\zeta$ along the equator, $\Psi(\zeta)$ is a map
\be
\Psi(\zeta)\colon S\to  SL(n,\R)/\Gamma, 
\ee
with $\Gamma$ the monodromy group.
Composing with the canonical projection
\be
\pi_\text{can}\colon SL(n,\R)\to SO(n)\backslash SL(n,\R)
\ee 
we get a $\zeta$-independent
 pluri-harmonic map 
 \be
 S\to SO(n)\backslash SL(n,\R)/\Gamma
 \ee
which, through the algorithm in appendix
\ref{revtau},  yields the corresponding solution to the $tt^*$ PDEs.

 If the family $\{\mathcal{R}_s\}_{s\in S}$ of chiral rings has special physical properties it may happen
that the ``brane group'',
that is, the image $G(\R)$ of the universal  lift
\be
\widetilde{\Psi}(\zeta)\colon\widetilde{S}\to SL(n,\R)\qquad\quad |\zeta|=1,
\ee
is a proper real Lie subgroup $G(\R)\subsetneq SL(n,\R)$. This reduction of the ``brane group'' has important implications for the quantum theory.  

We focus on a family of superconformal (2,2) models, i.e.\!
the rings $\{\mathcal{R}_s\}_{s\in S}$ are local and graded by the $U(1)_R$ charge operator $\boldsymbol{Q}$ (see eqns.\eqref{pqawe11},\eqref{pqawe12}) while $S$ is a space of exactly marginal deformations. For simplicity of notation we also assume $\hat c\in \mathbb{N}$ (the extension to fractional $\hat c$ being straightforward).
$\boldsymbol{Q}$ yields a $U(1)_R$ grading of the
Lie algebra $\mathfrak{g}$ of  $G(\R)$ of the form
\be
\mathfrak{g}\otimes \C=\bigoplus_{q=-\hat c}^{\hat c}\mathfrak{g}^{-q,q},\qquad X\in \mathfrak{g}^{q,-q}\ \Leftrightarrow\ \big[\boldsymbol{Q},X]=q\,X.
\ee

\begin{lem} The real Lie group $G(\R)$ is of ``Mumford-Tate type'' i.e.\! it contains a compact maximal torus, that is,
\be
\mathrm{rank}\, G(\R)=\mathrm{rank}\,K, \qquad \text{$K\subset G(\R)$ a maximal compact subgroup.}
\ee 
Hence $G(\R)$ is a real Lie groups in \textsc{Figure 6.1} and \textsc{Figure 6.2} of {\rm\cite{knapp}.}
Moreover, if
\be
q_\text{max}\overset{\rm def}{=}\max_q\mspace{-2mu}\Big\{ \mathfrak{g}^{q,-q}\neq0\Big\}
\ee
is \emph{odd} $G(\R)$ must be
in the list of odd-weight groups on page 24 of {\rm\cite{revb}.}
\end{lem}

\begin{proof} From \S.\,X.3 of the second edition of \cite{knapp} we know that
$\boldsymbol{Q}\in \mathfrak{g}$.
Then $\exp(\theta \boldsymbol{Q})\subset G(\R)$ is a compact one-parameter subgroup since it preserves the positive-definite $tt^*$ metric. Let $T$ a maximal torus containing
$\exp(\theta \boldsymbol{Q})$, and $\mathfrak{t}$ its Lie algebra. Clearly
$[\boldsymbol{Q},\mathfrak{t}]=0$, that is,
$\mathfrak{t}\subset \mathfrak{g}^{0,0}\cap\mathfrak{g}$.
Since the compact part $\mathfrak{k}\subset\mathfrak{g}$ is given by
\be
\mathfrak{k}\otimes \C=\bigoplus_{q\ \text{even}}\mathfrak{g}^{-q,q},
\ee
$\mathfrak{t}\subset \mathfrak{g}^{0,0}\cap\mathfrak{g}\subset\mathfrak{k}$ is compact.
\end{proof}

The generic ``brane group''
$SL(n,\R)$ does \emph{not} have the above property (for $n\geq3$) 
\be
n-1\equiv \mathrm{rank}\, SL(n,\R)\neq \mathrm{rank}\, SO(n)\equiv [n/2].
\ee

The above result shows that the conformal brane amplitudes are very different from the well-known one for the massive 2d models. They have \emph{higher arithmetics.} Indeed the  real Lie groups of the ``Mumford-Tate''
type are characterized by their interesting arithmetic: in particular they are the \emph{only} groups which have discrete series automorphic representations
\cite{disseries,disseries1}.
In view of eqn.\eqref{pqawe12}
we have

\begin{corl} The $tt^*$ pluriharmonic map $w$ of a \emph{superconformal} 2d (2,2) model, as a function of the conformal manifold $S$, factorizes through a quotient of a Mumford-Tate domain $G(\R)/H$,
where $H=\exp(\mathfrak{g}^{0,0}\cap \mathfrak{g})$
\be
\xymatrix{S\ar[r]_(.3){p}\ar@/^2pc/[rr]^w& \Gamma\backslash G(\R)/H\ar@{->>}[r]& \Gamma\backslash G(\R)/K}
\ee
$p$ holomorphic with $p_\ast(TS)\subseteq \co(\mathfrak{g}^{-1,1})$.
\end{corl}
The (arithmetic) quotients of Mumford-Tate domains are the natural generalization \cite{MT5} of the Shimura varieties, the ``paradise'' of arithmetics.

\section{Symmetric rigid special K\"ahler manifolds}
\begin{lem} A symmetric rigid special K\"ahler manifold $M$ is flat.
\end{lem}
\begin{proof} Let $M$ be a symmetric rigid special K\"ahler manifold. Without loss we may assume $M$ to be irreducible. Then $M$ is Einstein $R_{i\bar j}=\lambda\, G_{i\bar j}$. If $\lambda=0$ the symmetric manifold is flat, the period map $p$ is constant, and the QFT is free. Otherwise we see from eqn.\eqref{hodgeversusG} the Hodge metric $K_{i\bar j}$ coincides, up to a factor, with the special K\"ahler metric $G_{i\bar j}$. Then the Ricci curvatures of the two metrics are equal and non-singular. But the Ricci curvature of the Hodge metric is non-positive \cite{GII,periods} while the Ricci curvature of $G_{i\bar j}$ is the Hodge metric and hence non-negative. We got a contradiction.
\end{proof}


\end{document}